\documentclass[conference,compsoc]{IEEEtran}
\ifCLASSOPTIONcompsoc
  \usepackage[nocompress]{cite}
\else
  \usepackage{cite}
\fi

\usepackage[hyphens,spaces,obeyspaces]{url}
\usepackage{xcolor}
\ifdefined\isnotextended\def\separateproof{1}\fi
\ifdefined\isproofdoc\def\separateproof{1}\fi
\ifdefined\separateproof
\usepackage{xr-hyper}
\ifdefined\isproofdoc
\else
\makeatletter
\newcommand*{\addFileDependency}[1]{\typeout{(#1)}\@addtofilelist{#1}\IfFileExists{#1}{}{\typeout{No file #1.}}}
\makeatother
\addFileDependency{proof.tex}\addFileDependency{proof.aux}
\fi
\fi
\usepackage{hyperref}
\hypersetup{
  colorlinks,
  linkcolor={green!80!black},
  citecolor={red!70!black},
  urlcolor={blue!70!black}
}
\def\isnotanon{1}
\usepackage[acronym]{glossaries}
\makeglossaries

\usepackage{soul}
\usepackage{listings}

\usepackage{graphicx}
\usepackage{tabularx}
\usepackage{booktabs}
\usepackage{etoolbox}
\usepackage[normalem]{ulem}

\usepackage{placeins}

\usepackage{environ}

\makeatletter
\NewEnviron{myhideenv}{%
  \ifdefined\isnotextended
      \begingroup
        \renewenvironment{figure}[1][]{%
          \begingroup
            \def\@captype{figure}%
            \begin{minipage}{0pt}%
        }{%
            \end{minipage}%
          \endgroup}
        \renewenvironment{figure*}[1][]{%
          \begingroup
            \def\@captype{figure}%
            \begin{minipage}{0pt}%
        }{%
            \end{minipage}%
          \endgroup}

        \renewenvironment{subfigure}[2][]{%
          \begingroup
            \def\@captype{figure}%
            \begin{minipage}{0pt}%
        }{%
            \end{minipage}%
          \endgroup}
        \let\origcaption\caption      %
        \renewcommand*\caption[2][]{%
          \phantomcaption             %
          \origcaption*[{##1}]{}       %
        }

        \renewcommand*{\includegraphics}[2][]{\relax}       %

        \renewcommand*{\addcontentsline}[3]{}%

        \setbox0\vbox{%
   \let\xxwrite\write
   \protected\def\write{\immediate\xxwrite}%
   {\tiny XX\BODY XX}}
      \endgroup
  \else
      \BODY
  \fi}
\makeatother

\newcommand{\sortfirst}[1]{}
\ifdefined\isnotextended
\ifdefined\isnotanon

\else

\fi
\else

\fi

\usepackage{tcolorbox}
\tcbset{
    colframe=black, %
    colback=white, %
    fonttitle=\bfseries, %
    boxrule=0.8pt, %
    arc=1mm, %
    left=3mm, right=3mm, top=2mm, bottom=2mm, %
    toptitle=1mm, bottomtitle=1mm
}

\usepackage{tikz}
\usetikzlibrary{decorations.pathreplacing,arrows.meta,positioning,calc}

\usepackage{enumitem,amssymb}
\newlist{todolist}{itemize}{2}
\setlist[todolist]{label=$\square$}
\usepackage{pifont}

\usepackage[makeroom]{cancel}

\usepackage{glossaries}
\glsdisablehyper

\usepackage{graphicx}
\usepackage{subcaption}
\usepackage{comment}
\usepackage{xspace}

\usepackage{multirow}
\usepackage{soul}

\usepackage{balance}

\usepackage{amsthm}
\usepackage{thm-restate}
\usepackage{mathtools}
\usepackage[utf8]{inputenc}
\usepackage[operators,sets,primitives,keys,notions,advantage,adversary,logic,ff,asymptotics,probability]{cryptocode}
\createprocedureblock{procb}{center,boxed}{}{}{codesize=\normalsize,width=\textwidth}

\usepackage[capitalise]{cleveref}
\newcommand{\ldef}[1]{\label{def:#1}}
\makeatletter
\newcommand{\crefnames}[3]{%
  \@for\next:=#1\do{%
    \expandafter\crefname\expandafter{\next}{#2}{#3}%
  }%
}
\makeatother

\crefnames{part,chapter,section}{\S}{\S\S}

\usepackage{color}
\usepackage{xcolor}

\setlist[itemize]{noitemsep, topsep=0pt, leftmargin=14pt}
\setlist[enumerate]{noitemsep, topsep=0pt, leftmargin=14pt}

\newlist{algos}{itemize}{2}
\setlist[algos]{align=left,itemsep=2pt,left=0pt,label=•}

\usepackage{colortbl}
\definecolor{Gray}{gray}{0.65}
\definecolor{LightGray}{gray}{0.9}

\usepackage{array}

\usepackage{amsmath}
\usepackage{bm}
\usepackage{amssymb}
\usepackage{hyperref}
\usepackage{cleveref}

\usepackage{mdframed}

\newmdenv[
  backgroundcolor=white,
  linecolor=black,
  linewidth=1pt,
  roundcorner=5pt,
  innerleftmargin=10pt,
  innerrightmargin=10pt,
  innertopmargin=10pt,
  innerbottommargin=10pt
]{protocolbox}

\usepackage{xspace}

\newcommand{\todoCameraReady}[1]{\textcolor{green}{\small{TODO Camera Ready: #1}}}
\renewcommand{\todoCameraReady}[1]{}

\newcommand{\todoExtendedVersion}[1]{\textcolor{green}{\small{TODO Extended Version: #1}}}
\ifdefined\isnotextended
\renewcommand{\todoExtendedVersion}[1]{}
\fi

\newcommand{\artemisBlue}[1]{\blue{#1}}

\newcommand{\blue}[1]{{\color{blue}{#1}}}

\newcommand{\shorten}[1]{}

\newcommand{\lsec}[1]{\label{sec:#1}}

\newcommand{\figcaptionvspace}{} %

\newcommand{\subsecspacingtop}{}
\newcommand{\subsecspacingbot}{}

\newcommand{\oursystem}{Artemis\xspace} %
\newcommand{\ourlunar}{Apollo\xspace}

\newcommand{\fakeparagraph}[1]{\par\noindent\textbf{#1. }}
\newcommand{\protocolparagraph}[1]{\par\noindent\textbf{#1. }}

\newcommand{\fakeparagraphnovskip}[1]{\noindent\textbf{#1. }}

\newcommand{\githuburl}{https://github.com/pps-lab/artemis}
\newcommand{\githuburlanon}{https://anonymous.4open.science/r/artemis-735C}
\newcommand{\proofpdffile}{artemis\_security\_proof.pdf}
\newcommand{\proofrepourl}{\ifdefined\isnotanon\githuburl\else\githuburlanon\fi}

\theoremstyle{definition}
\newtheorem{definition}{Definition}[section]

\newtheorem{lemma}[definition]{Lemma}

\newtheorem{remark}{Remark}

\newtheorem{claim}{Claim}
\crefname{claim}{Claim}{Claims}
\Crefname{claim}{Claim}{Claims}

\newcommand{\sfield}{\ensuremath{\mathbb{F}}\xspace}
\newcommand{\sgroup}{\ensuremath{\mathbb{G}}\xspace}
\newcommand{\sgenerator}{\ensuremath{h}\xspace}
\newcommand{\sInstance}{\ensuremath{x}\xspace}

\renewcommand{\adv}{\ensuremath{\pcadvstyle{A}}\xspace} %

\newcommand{\sEvalPoint}{\ensuremath{\beta}\xspace}
\newcommand{\sEvalValue}{\ensuremath{\rho}\xspace}

\newcommand{\extr}{\ensuremath{\mathcal{E}}\xspace}
\newcommand{\extCompiler}{\ensuremath{\extr_{\oursystem}}\xspace}

\newcommand{\extPC}{\ensuremath{\extr_{\PCScheme}}\xspace}

\newcommand{\advPC}{\ensuremath{\adv_{\PCScheme}}\xspace}

\newcommand{\simu}{\ensuremath{\mathcal{S}}\xspace}

\newcommand{\negPC}{\ensuremath{\epsilon_{\PCScheme}}\xspace}

\newcommand{\sWitness}{\ensuremath{w}\xspace}

\newcommand{\sWitnessExtra}{\ensuremath{v}\xspace}

\newcommand{\sRandomness}{\ensuremath{r}\xspace}
\newcommand{\sProtocolRandomness}{\ensuremath{\sMask}\xspace}

\newcommand{\spolynomial}{\ensuremath{\mathbf{g}}\xspace}

\newcommand{\ck}{\key[ck]\xspace}
\newcommand{\td}{\key[td]}

\newcommand{\scomRandomness}{\ensuremath{r}\xspace}
\newcommand{\scomMessageSpace}{\ensuremath{\mathcal{M}}\xspace}
\newcommand{\scomRandomnessSpace}{\ensuremath{\mathcal{O}}\xspace}
\newcommand{\scomCommitmentSpace}{\ensuremath{\mathcal{C}}\xspace}

\newcommand{\pcDbound}{\ensuremath{d}\xspace}
\newcommand{\piCommitment}{\ensuremath{\pi_{\text{Link}}}\xspace}

\newcommand{\zkSARK}{\ensuremath{\Pi}\xspace}

\newcommand{\pol}[1]{\ensuremath{\mathbf{#1}}\xspace}

\newcommand{\PCScheme}{\textsf{PC}\xspace}
\newcommand{\PCSetup}{\textsf{\PCScheme.Setup}\xspace}
\newcommand{\PCCommit}{\textsf{\PCScheme.Commit}\xspace}
\newcommand{\PCProve}{\textsf{\PCScheme.Open}\xspace}
\newcommand{\PCVerify}{\textsf{\PCScheme.Verify}\xspace}
\newcommand{\PCCheck}{\textsf{\PCScheme.Check}\xspace}

\newcommand{\ourcompiler}{\textsf{\oursystem}\xspace}

\newcommand{\zkSetup}{\textsf{Setup}\xspace}

\newcommand{\ipk}{\ensuremath{\textsf{ipk}}\xspace}
\newcommand{\ivk}{\ensuremath{\textsf{ivk}}\xspace}

\newcommand{\sComLabel}{\textsf{Com}\xspace}
\newcommand{\sComSetup}{\textsf{Com.Setup}\xspace}
\newcommand{\sComCommit}{\textsf{Com.Commit}\xspace}
\newcommand{\sComVerify}{\textsf{Com.Verify}\xspace}

\newcommand{\sPedSetup}{\textsf{Ped.Setup}\xspace}
\newcommand{\sPedCommit}{\textsf{Ped.Commit}\xspace}
\newcommand{\sPedVerify}{\textsf{Ped.Verify}\xspace}

\newcommand{\sprover}{\ensuremath{\mathcal{P}}\xspace}
\newcommand{\sverifier}{\ensuremath{\mathcal{V}}\xspace}
\newcommand{\malProver}{\ensuremath{\mathcal{P}^{*}}\xspace}

\newcommand{\scom}{\ensuremath{c}\xspace}
\newcommand{\scomCount}{\ensuremath{\ell}\xspace}
\newcommand{\scomOne}{\ensuremath{c_1}\xspace}
\newcommand{\scomList}{\ensuremath{\scomOne, \ldots, \scom_\scomCount}\xspace}

\newcommand{\srandList}{\ensuremath{\sRandomness_1, \ldots, \sRandomness_\scomCount}\xspace}

\newcommand{\scommittedWitness}{\ensuremath{\sWitness}\xspace}

\newcommand{\commitmentCombination}{\ensuremath{\scom_\sProtocolRandomness + \textstyle\sum_{i=1}^\ell \scom_{i} \cdot \alpha^{i}}\xspace}

\newcommand{\nRounds}{\ensuremath{\textsf{k}}\xspace}
\newcommand{\nChal}{\ensuremath{n}\xspace} %

\newcommand{\pMsg}{\ensuremath{m}\xspace}
\newcommand{\dCircuit}{\ensuremath{D}\xspace}

\newcommand{\aip}{\ensuremath{\phi}\xspace}

\newcommand{\zkc}{\ensuremath{\gamma}\xspace}

\newcommand{\idxWit}{\ensuremath{\mathtt{idx}_\text{wit}}\xspace}
\newcommand{\idxRho}{\ensuremath{\mathtt{idx}_\rho}\xspace}
\newcommand{\idxMu}{\ensuremath{\mathtt{idx}_\mu}\xspace}

\newcommand{\Calpha}{\ensuremath{C_{\alpha}}\xspace}
\newcommand{\Cbeta}{\ensuremath{C_{\sEvalPoint}}\xspace}

\newcommand{\sRelation}{\ensuremath{\mathcal{R}}\xspace}
\newcommand{\sRelationClass}{\ensuremath{\mathfrak{R}}\xspace}         %

\newcommand{\sRelationCP}{\ensuremath{\sRelation^\sComLabel}\xspace}

\newcommand{\sIndexSet}{\ensuremath{I}\xspace}
\newcommand{\sIndexSetCom}{\ensuremath{\sIndexSet_\text{com}}\xspace}

\newcommand{\spolySingle}{\ensuremath{\mathbf{\scommittedWitness_i}}\xspace}
\newcommand{\spolySinglePlus}{\ensuremath{\mathbf{\scommittedWitness_{i+1}}}\xspace}

\newcommand{\spolyMask}{\ensuremath{\bm{\mu}}\xspace}

\newcommand{\cplink}{\ensuremath{\mathsf{CP}_{\text{link}}}\xspace}
\newcommand{\cplinkone}{\ensuremath{\mathsf{CP}_{\text{link}}^{(1)}}\xspace}
\newcommand{\cplinktwo}{\ensuremath{\mathsf{CP}_{\text{link}}^{(2)}}\xspace}

\newcommand{\numberCols}{\ensuremath{n_h}\xspace}

\newcommand{\polyGate}{\ensuremath{\mathbf{b}}\xspace}
\newcommand{\polyHornerGate}{\polyGate_{\textsf{h}}\xspace}

\newcommand{\sMask}{\ensuremath{\mu}\xspace}

\renewcommand{\secpar}{\ensuremath{\lambda}}

\newacronym{ml}{ML}{machine learning}
\newacronym{zkp}{ZKP}{zero-knowledge proof}
\newacronym{zkml}{zkML}{zero-knowledge machine learning}
\newacronym{poc}{PoC}{Proof-of-Consistency}
\newacronym[longplural=Succinct Non-Interactive Arguments of Knowledge]{snark}{SNARK}{Succinct Non-Interactive Argument of Knowledge}
\newacronym[longplural=zero-knowledge Succinct Non-Interactive Arguments of Knowledge]{zksnark}{zk-SNARK}{zero-knowledge Succinct Non-Interactive Argument of Knowledge}
\newacronym[longplural=zero-knowledge Succinct Arguments of Knowledge]{zksark}{zk-SARK}{zero-knowledge Succinct Argument of Knowledge}
\newacronym{cpsnark}{CP-SNARK}{Commit-and-Prove SNARK}
\newacronym{cpsark}{CP-SARK}{Commit-and-Prove Succinct Argument of Knowledge}
\newacronym{ahp}{AHP}{Algebraic Holographic Proof}
\newacronym{iop}{IOP}{Interactive Oracle Proof}
\newacronym{piop}{PIOP}{Polynomial Interactive Oracle Proof}
\newacronym{qap}{QAP}{Quadratic Arithmetic Program}
\newacronym{air}{AIR}{Algebraic Intermediate Representation}
\newacronym{aip}{AIP}{Algebraic Interactive Proof}

\newglossaryentry{mnist}{name={MNIST},description={}}
\newglossaryentry{resnet18}{name={ResNet-18},description={}}
\newglossaryentry{dlrm}{name={DLRM},description={}}
\newglossaryentry{mobilenet}{name={MobileNet},description={}}
\newglossaryentry{vgg}{name={VGG-16},description={}}
\newglossaryentry{diffusion}{name={Diffusion},description={}}
\newglossaryentry{gpt2}{name={GPT-2},description={}}

\newglossaryentry{no_com}{name=\texttt{No Commitment},description={}}
\newglossaryentry{poly}{name=\texttt{\oursystem},description={}}
\newglossaryentry{cp_link}{name=\texttt{Lunar},description={}}
\newglossaryentry{cp_link+}{name=\texttt{\ourlunar},description={}}
\newglossaryentry{poseidon}{name=\texttt{Poseidon},description={}}
\newglossaryentry{zkfft}{name=\texttt{zkFFT},description={}}
\newglossaryentry{legosnark}{name=\texttt{LegoSNARK},description={}}
\newglossaryentry{sha256}{name=\texttt{SHA-256},description={}}

\newsavebox\ARelation
\setbox\ARelation\hbox
 {$\left\{
        \begin{array}{c}
             ((\sInstance, \alpha, \sEvalPoint, \sEvalValue, \artemisBlue{\scom_\sProtocolRandomness}), (\sWitness, \spolyMask, \artemisBlue{\sRandomness_\sProtocolRandomness})) \tabularnewline
             : \tabularnewline
             (\sInstance, \sWitness) \in \sRelation \land \sEvalValue = (\spolyMask + \sum_i \spolySinglePlus \cdot \alpha^i)(\sEvalPoint) \tabularnewline
             \land~ \artemisBlue{\PCVerify(\ck, \scom_\sProtocolRandomness, \spolyMask, \sRandomness_\sProtocolRandomness) = 1}
        \end{array}
         \right\}$%
 }

\providecommand{\zkProver}{\ensuremath{P_{\zkSARK}}\xspace}                  %
\providecommand{\zkVerifier}{\ensuremath{V_{\zkSARK}}\xspace}                %
\providecommand{\aipObj}{\ensuremath{\mathsf{A}}\xspace}                     %
\providecommand{\arithK}{\ensuremath{\mathcal{K}}\xspace}                    %
\providecommand{\aipExt}{\ensuremath{\aip^{\prime}}\xspace}                  %
\providecommand{\hornerGate}{\textsf{AddHornersGate}\xspace}
\providecommand{\artemisMain}{\ensuremath{\textsf{Artemis}[\PCScheme]}\xspace}                 %
\providecommand{\artemisMainOut}{\ensuremath{\artemisMain(\zkSARK)}\xspace}                 %
\providecommand{\zkSARKCP}{\ensuremath{\zkSARK^{\mathsf{CP}}}\xspace}                      %
\providecommand{\artemisLink}{\ensuremath{\langle \mathcal{P}^{[\PCScheme]}, \mathcal{V}^{[\PCScheme]} \rangle}\xspace}                %
\providecommand{\msgProj}{\ensuremath{\mathsf{Proj}}\xspace}                 %
\providecommand{\msgProjI}{\ensuremath{\msgProj_{\sIndexSetCom}}\xspace}        %
\providecommand{\msgMap}{\ensuremath{\mathsf{pos}}\xspace}                  %
\providecommand{\pMsgs}{\ensuremath{\pMsg_1,\ldots,\pMsg_\nRounds}\xspace}
\providecommand{\aipChals}{\ensuremath{\zkc_1,\ldots,\zkc_{\nRounds-1}}\xspace}
\providecommand{\sRelationMu}{\ensuremath{\sRelation_{\sMask}}\xspace}        %
\providecommand{\revealRounds}{\ensuremath{H_{\zkSARK}}\xspace}              %
\providecommand{\hRound}{\ensuremath{h_i}\xspace}
\providecommand{\hIdx}{\ensuremath{h_j}\xspace}
\providecommand{\extSARK}{\ensuremath{\extr_{\zkSARK}}\xspace}
\providecommand{\extAIP}{\ensuremath{\extr_{\aip}}\xspace}
\providecommand{\negSARK}{\ensuremath{\epsilon_{\zkSARK}}\xspace}
\providecommand{\negAIP}{\ensuremath{\epsilon_{\aip}}\xspace}
\providecommand{\negReveal}{\ensuremath{\epsilon_{\mathsf{Rev}}}\xspace}
\providecommand{\malProver}{\ensuremath{\mathcal{P}^{*}}\xspace}

\ifdefined\separateproof\ifdefined\isproofdoc\else
\newcommand*{\aliaslabel}[2]{%
  \expandafter\let\csname r@#1\expandafter\endcsname\csname r@#2\endcsname
  \expandafter\let\csname r@#1@cref\expandafter\endcsname\csname r@#2@cref\endcsname}
\AtBeginDocument{%
  \aliaslabel{apx:compiler}{apx:securityproof}%
  \aliaslabel{apx:compiler:ext}{apx:securityproof}%
  \aliaslabel{def:chsnark}{S-def:chsnark}%
  \aliaslabel{def:message-ext}{S-def:message-ext}%
  \aliaslabel{def:vr-snark}{S-def:vr-snark}}
\fi\fi

\IEEEoverridecommandlockouts   %

\begin{document}

\ifdefined\isproofdoc

\title{Artemis: Efficient Commit-and-Prove SNARKs\\[0.3em]\Large Supplementary Material: Security Proof}
\ifdefined\isnotanon
\author{\IEEEauthorblockN{Hidde Lycklama, Alexander Viand, Nikolay Avramov, Nicolas K\"uchler, Anwar Hithnawi}}
\else
\author{\textit{Supplementary material for Paper \#1581}}
\fi
\date{}
\maketitle

\noindent
This document is Appendix~D of the paper: the properties of \glspl{zksark} that \oursystem requires and the proof of the main security theorem (\Cref{thm:artemis:aip} of the paper).
References to items outside Appendix~D point to the paper.

\begingroup
\makeatletter
\let\label\@gobble
\setbox0\vbox{
\begin{restatable}[\artemisMain yields a (zk-)\gls{cpsark}]{theorem}{thmArtemisMain}
\label{thm:artemis:aip}
    Let \PCScheme be an additively homomorphic polynomial commitment scheme over \sfield with binding and extractability (\Cref{def:pc,def:homomorphic_commitment}),
    and let $\zkSARK$ be a \gls{zksark} compiler for an interactive arithmetization \arithK for \sRelation that is challenge-exposing (\Cref{def:chsnark}), message-extractable (\Cref{def:message-ext}), and soundly-revealing (\Cref{def:vr-snark}) for every instance of \arithK.
    Let $\aip \in \arithK$ encode a complete and knowledge-sound \gls{aip} for \sRelation, and
    let $\zkSARKCP = \artemisMainOut$ be the output of the meta-compiler described in~\Cref{fig:compiler:artemis:meta}.
    \zkSARKCP is a zk-\gls{cpsark} compiler for \sRelation and \PCScheme (\Cref{def:cpsnark}).
    If \PCScheme is hiding, then this protocol is honest-verifier zero-knowledge.
\end{restatable}
}
\endgroup

\appendices
\crefalias{section}{appendix}
\crefname{appendix}{Appendix}{Appendices}
\Crefname{appendix}{Appendix}{Appendices}
\setcounter{section}{3}

\section{Security Proof for \oursystem}
\label{apx:compiler}
This appendix provides a security proof for the \oursystem meta-compiler.
We first formally define some properties of proof systems used by our meta-compiler and its security proof.
A technicality in the proof is that for knowledge soundness,
our extractor must be able to extract the randomness of the individual commitments,
even though we only have a single evaluation proof that is masked by a random value.
To do so, our extractor internally invokes the extractor of \PCScheme several times to reconstruct the randomness from
different evaluation proofs.


\makeatletter
\@ifundefined{remark}{\newtheorem{remark}[definition]{Remark}}{}
\makeatother
\crefname{remark}{Remark}{Remarks}
\Crefname{remark}{Remark}{Remarks}

\subsection{\Glspl{zksark} for Interactive Arithmetizations}
\label{apx:compiler:zksark}
We define the three properties that \oursystem requires of proof systems for interactive arithmetizations (cf.~\Cref{def:arithmetization}): the proof system exposes the challenges of the \gls{aip}, an extractor recovers the messages of the \gls{aip}, and the prover can soundly reveal values.

\fakeparagraph{Exposing the challenges}
\oursystem evaluates a polynomial at challenges of the \gls{aip}.
The verifier of the proof system must therefore send these challenges in the clear, in the order of the \gls{aip}.
We only require that the interaction contains the challenges of the \gls{aip}, but
the proof system may send further challenges of its own.
\begin{definition}[Challenge-Exposing]
\label{def:chsnark}
    A proof system \zkSARK for \arithK is \emph{challenge-exposing for} an instance $\aip \in \arithK$ with challenges \aipChals if, for $(\ipk, \ivk) \gets \zkSARK.\zkSetup(\aip)$, the verifier \zkVerifier samples \aipChals uniformly at random from the challenge spaces of \aip and sends them to \zkProver in the clear, in the order of \aip, where \zkVerifier sends $\zkc_i$ only after \zkProver has sent the messages of \zkSARK that correspond to round $i$ of \aip.
\end{definition}

\fakeparagraph{Extracting the transcript}
The standard notion of knowledge soundness extracts a witness for the relation.
\oursystem needs a similar statement on the level of \gls{aip}.
We therefore require an extractor that recovers the messages of the \gls{aip}, each determined by the challenges that precede it.
The extractor of the \gls{aip} can then be used to yield the witness (\Cref{claim:witness} in the proof of Theorem~\ref{thm:artemis:aip}).

\begin{definition}[Message Extractability]
\label{def:message-ext}
    Let $\aip \in \arithK$ have \nRounds rounds, and let \zkSARK be a proof system for \arithK that is challenge-exposing for \aip (\Cref{def:chsnark}).
    \zkSARK is \emph{message-extractable for \aip} with extraction error $\negSARK$ if for every PPT prover $\malProver$ there exists an expected polynomial-time extractor \extSARK as follows.
    \extSARK receives the state $\mathsf{st}$ of $\malProver$ (including its random coins) and first runs one complete interaction with $\malProver$ in the role of $\zkVerifier(\ivk, \sInstance)$, the \emph{main transcript} $\tau$.
    \extSARK then rewinds $\malProver$ to any point of $\tau$ and continues with fresh challenges, and finally outputs messages $\tilde{\pMsg}_1, \ldots, \tilde{\pMsg}_\nRounds$ of \aip.
    We write $(\tau; \tilde{\pMsg}_1, \ldots, \tilde{\pMsg}_\nRounds) \leftarrow \langle \malProver, \extSARK \rangle(\ipk, \ivk, \sInstance)$ for this experiment, with $(\ipk, \ivk) \gets \zkSARK.\zkSetup(\aip)$, and \aipChals for the challenges of \aip in $\tau$.
    For every $\malProver$ and \sInstance, with probabilities over the experiment:
    \begin{algos}
        \item \fakeparagraph{Extraction}
        \begin{align*}
            \Pr\bigl[\, &\zkVerifier \text{ accepts } \tau \;\land \\
            &\dCircuit_{\aip} \text{ rejects } (\sInstance, \tilde{\pMsg}_1, \ldots, \tilde{\pMsg}_\nRounds, \aipChals) \,\bigr] \le \negSARK .
        \end{align*}
        \item \fakeparagraph{Prefix Determination}
        There are functions $f_1, \ldots, f_\nRounds$ (not necessarily efficient) such that
        \begin{align*}
            \Pr\bigl[\, &\zkVerifier \text{ accepts } \tau \;\land \\
            &\exists i \in [\nRounds] : \tilde{\pMsg}_i \neq f_i(\mathsf{st}, \zkc_1, \ldots, \zkc_{i-1}) \,\bigr] \le \negSARK .
        \end{align*}
    \end{algos}
\end{definition}
\noindent
In other words, the extracted transcript is accepting whenever $\malProver$ convinces $\zkVerifier$, and in this case the state of the prover and the challenges that precede a message determine the extracted message.
In particular, $\tilde{\pMsg}_i$ is independent of $\zkc_i, \ldots, \zkc_{\nRounds-1}$ and of the coins of \extSARK, up to \negSARK.
\Cref{def:message-ext} does not require \emph{online} extraction, in which \extSARK outputs $\tilde{\pMsg}_i$ before \zkVerifier sends $\zkc_i$.
For Mirage+, \extSARK is the generic-group extractor of Ozdemir et al.~\cite[App.~A.4]{Ozdemir2024-vh}, which obtains the explanation of each proof element before it forwards $\zkc_i$ and is therefore online.
For \gls{piop}-based proof systems with a rewinding polynomial-commitment extractor, \extSARK runs the polynomial-commitment extractor on the evaluation proofs of $\tau$ after the interaction.
Prefix determination then follows from binding: two extractions of a commitment sent in round $i$ that differ open the commitment to two different polynomials.

\fakeparagraph{Soundly revealing values}
We capture the ability to soundly reveal values from messages to the verifier as \emph{sound revelation}.

\begin{definition}[Sound Revelation]
\label{def:vr-snark}
    Let \zkSARK be a proof system for \arithK that is challenge-exposing and message-extractable for an instance $\aip \in \arithK$ with \nRounds rounds, and let $\revealRounds \subseteq \mathbb{N}$ be a set of round indices, which may depend on \nRounds and contains every round except possibly the last, i.e., $[\nRounds - 1] \subseteq \revealRounds$.
    \zkSARK supports \emph{sound revelation} for \aip and rounds \revealRounds if for every round $i \in \revealRounds \cap [\nRounds]$ and every coordinate $j \in [|\pMsg_i|]$ there is an efficient and succinct interactive protocol
    \begin{align*}
    \langle P^{(i,j)}_\mathsf{Reveal}(\ipk, \sInstance, \sWitness, \tau, \sEvalValue),\
    V^{(i,j)}_\mathsf{Reveal}(\ivk, \sInstance, \tau, \sEvalValue) \rangle ,
    \end{align*}
    where $\tau$ is a transcript of \zkSARK and $\sEvalValue \in \sfield$ a claimed value, with the following properties.
    Below, $\pMsg_i[j]$ denotes the $j$-th coordinate of the round-$i$ message of the honest prover of \aip in the execution recorded in $\tau$.
    \begin{algos}
        \item \fakeparagraph{Reveal Completeness}
        For every $(\sInstance, \sWitness)$ that is valid for \aip, every honestly generated $\tau \gets \langle \zkProver(\ipk, \sInstance, \sWitness), \zkVerifier(\ivk, \sInstance) \rangle$ and $\sEvalValue := \pMsg_i[j]$, the reveal verifier $V^{(i,j)}_\mathsf{Reveal}$ accepts with probability $1$.

        \item \fakeparagraph{Reveal Soundness}
        For every PPT prover $\malProver$, with \extSARK its extractor from \Cref{def:message-ext} and $\tilde{\pMsg}_i$ the round-$i$ message that \extSARK outputs, and for every \sInstance:
        {\footnotesize
        \[
        \condprob{
            \begin{gathered}
                \zkVerifier \text{ accepts } \tau \\
                \land\ V^{(i,j)}_\mathsf{Reveal} \text{ accepts } (\sInstance, \tau, \tilde{\sEvalValue}) \\
                \land\ \tilde{\pMsg}_i[j] \neq \tilde{\sEvalValue}
            \end{gathered}
        }{
            \begin{gathered}
                (\ipk, \ivk) \gets \zkSARK.\zkSetup(\aip) \\
                (\tau; \tilde{\pMsg}_1, \ldots, \tilde{\pMsg}_\nRounds) \gets \\
                \quad \langle \malProver, \extSARK \rangle(\ipk, \ivk, \sInstance) \\
                (\tilde{\sEvalValue}, \tau_\mathsf{R}) \gets \\
                \quad \langle \malProver, V^{(i,j)}_\mathsf{Reveal}(\ivk, \sInstance, \tau, \cdot) \rangle
            \end{gathered}
        } \le \negReveal .
        \]
        }
        That is, no PPT prover passes the reveal protocol on a value that disagrees with the message extracted from its main proof.

        \item \fakeparagraph{Reveal Zero-Knowledge}
        There exists a PPT simulator $\sdv = (\sdv.\textsf{Setup}, \sdv.\textsf{Sim})$ such that for every PPT adversary $\adv = (\adv_1, \adv_2)$,
        $\bigl| \Pr[\mathsf{Real}_{\adv} = 1] - \Pr[\mathsf{Ideal}_{\adv} = 1] \bigr| = \negl$, where:

        \smallskip\noindent
        $\mathsf{Real}_{\adv}$:
        $(\ipk, \ivk) \gets \zkSARK.\zkSetup(\aip)$;
        $(\sInstance, \sWitness, \mathsf{st}) \gets \adv_1(\ipk, \ivk)$;
        abort if $(\sInstance, \sWitness)$ is not valid for \aip;
        sample verifier coins $\vec{\theta}$ for $\zkVerifier$ and $V_{\mathsf{Reveal}}$;
        $\tau \gets \langle \zkProver(\ipk, \sInstance, \sWitness), \zkVerifier(\ivk, \sInstance; \vec{\theta}) \rangle$;
        $\sEvalValue := \pMsg_i[j]$;
        $\tau_\mathsf{R} \gets \langle P^{(i,j)}_{\mathsf{Reveal}}(\ipk, \sInstance, \sWitness, \tau, \sEvalValue), V^{(i,j)}_{\mathsf{Reveal}}(\ivk, \sInstance, \tau, \sEvalValue; \vec{\theta}) \rangle$;
        return $\adv_2(\mathsf{st}, \tau, \sEvalValue, \tau_\mathsf{R})$.

        \smallskip\noindent
        $\mathsf{Ideal}_{\adv}$:
        $(\ipk, \ivk, \td) \gets \sdv.\textsf{Setup}(\aip)$;
        $(\sInstance, \sWitness, \mathsf{st}) \gets \adv_1(\ipk, \ivk)$;
        abort if $(\sInstance, \sWitness)$ is not valid for \aip;
        sample $\vec{\theta}$ as above;
        $\sEvalValue := \pMsg_i[j]$, computed by running the honest prover of \aip on $(\sInstance, \sWitness)$ with the challenges determined by $\vec{\theta}$;
        $(\tau, \tau_\mathsf{R}) \gets \sdv.\textsf{Sim}(\td, \sInstance, \vec{\theta}, \sEvalValue)$;
        return $\adv_2(\mathsf{st}, \tau, \sEvalValue, \tau_\mathsf{R})$.
    \end{algos}
\end{definition}
\noindent
Reveal zero-knowledge is stated for an honest verifier: the simulator receives the verifier coins $\vec{\theta}$.
After the Fiat--Shamir transform, this notion yields zero-knowledge in the random oracle model.
The set \revealRounds captures the rounds that a proof system can open.
The Halo2 \gls{piop} commits to every advice column of every phase separately and opens it with a polynomial evaluation.
Revealing a cell opens its advice polynomial at a point of $D^\prime$ (cf.~\Cref{thm:blinding} in \Cref{apx:blinding}).
The remaining $p$ openings are at points outside the domain, and with $p+1$ blinding rows \Cref{thm:blinding} shows that they are uniformly random.
The blinding term vanishes on $D^\prime$, so both hold jointly, and we get $\revealRounds = \mathbb{N}$.
Mirage+ folds its final message into the last proof element together with the quotient and the blinding terms, so $\revealRounds = [\nRounds - 1]$, and revealing a coordinate of an earlier proof element is a $\Sigma$-protocol opening.

\fakeparagraph{Revealing the Evaluation}
    Most proof systems support revealing from every round, e.g., Halo2 with $\revealRounds = \mathbb{N}$.
    \oursystem reveals from round $\nRounds+2$ of the extended \gls{aip} $\aipObj^\prime$, so it requires $\nRounds+2 \in \revealRounds$.
    A proof system that cannot reveal from the last round, such as Mirage+, meets this requirement after we append a round $\nRounds+3$ with an empty prover message to $\aipObj^\prime$.
    Round $\nRounds+2$ is then no longer the last round, and the empty message changes none of the arguments in the security proof.

\subsection{Security Proof}
\label{apx:compiler:ext}

We first define some notions used in the security proof.
We begin by arguing that our \gls{aip} extension, visualized in~\Cref{fig:aip:ext} and described in~\Cref{fig:compiler:artemis:meta}, preserves completeness and knowledge soundness.

\begin{lemma}[\oursystem' AIP extension preserves completeness and knowledge soundness]
\label{lemma:aipext}
    Let \aipObj be an \gls{aip} for \sRelation with knowledge error $\epsilon_{\aipObj}$, and let $\aipObj^\prime$ be the extension of \aipObj for \sIndexSetCom.
    Then $\aipObj^\prime$ is an \gls{aip} for
$$\sRelationMu := \{ (\sInstance, (\sWitness, \sMask)) : (\sInstance, \sWitness) \in \sRelation \land \sMask \in \sfield \}$$
    with knowledge error $\epsilon_{\aipObj}$.
    Moreover, every transcript $(\pMsgs, \sMask, \alpha, \sEvalPoint, \sEvalValue)$ accepted by $\dCircuit^\prime$ satisfies
    $\sEvalValue = \sMask + \sum_{i=1}^{\ell} \alpha^i \cdot \spolySingle(\sEvalPoint)$ with $\spolySingle$ defined by $\msgProjI(\pMsgs)$.
\end{lemma}
\begin{proof}
    Completeness: the honest prover of $\aipObj^\prime$ runs the honest prover of \aipObj, sends \sMask, and computes \sEvalValue as prescribed, so both conjuncts of $\dCircuit^\prime$ hold.
    Knowledge soundness: let \malProver be accepted by $\dCircuit^\prime$ with probability $\varepsilon$.
    $\dCircuit^\prime$ accepts only if \dCircuit accepts the first \nRounds messages and challenges, so the restriction of \malProver to rounds $1, \ldots, \nRounds$ is accepted by \dCircuit with probability at least $\varepsilon$.
    The extractor of \aipObj on the first \nRounds messages of this restriction outputs \sWitness with $(\sInstance, \sWitness) \in \sRelation$ with probability at least $\varepsilon - \epsilon_{\aipObj}$, and $(\sInstance, (\sWitness, 0)) \in \sRelationMu$.
    The second statement is the definition of $\dCircuit^\prime$.
\end{proof}

The meta-compiler assumes there exists a valid arithmetization of the extended \gls{aip} in \arithK.
The underlying assumption is that \arithK is \emph{universal}: every \gls{aip} over \sfield has an encoding in \arithK, possibly with auxiliary prover messages.
I-R1CS and I-Plonkish are universal, because both express arbitrary arithmetic circuits.
Universality says nothing about cost; \Cref{apx:compiler:inst} gives an efficient encoding of the extension for I-Plonkish.


\providecommand{\negBind}{\ensuremath{\epsilon_{\mathsf{bind}}}\xspace}
\providecommand{\negHide}{\ensuremath{\epsilon_{\mathsf{hide}}}\xspace}
\providecommand{\negRZK}{\ensuremath{\epsilon_{\mathsf{zk}}}\xspace}                %
\providecommand{\simSARK}{\ensuremath{\simu_{\zkSARK}}\xspace}
\providecommand{\polyRho}{\ensuremath{\bm{\rho}}\xspace}                      %
\providecommand{\simReveal}{\ensuremath{\simu_{\mathsf{Rev}}}\xspace}
\providecommand{\stA}{\ensuremath{\mathsf{st}}\xspace}                        %
\providecommand{\stAlpha}{\ensuremath{\mathsf{st}_{\alpha}}\xspace}           %
\providecommand{\polyIn}{\ensuremath{\hat{\mathbf{W}}}\xspace}                %
\providecommand{\polyBiv}{\ensuremath{\mathbf{F}}\xspace}                     %

\fakeparagraph{Committed Witnesses in \glspl{aip}}
\glspl{cpsnark} (cf.~\Cref{def:cpsnark}) use an index set \sIndexSetCom to relate the committed sub-witnesses to the witness \sWitness.
We write $\sIndexSetCom = (\sIndexSetCom^1, \ldots, \sIndexSetCom^\ell)$, where each $\sIndexSetCom^t$ is an ordered list of witness positions $u \in [|\sWitness|]$ and $\scommittedWitness_t = ( \sWitness[u] )_{u \in \sIndexSetCom^t}$.
Since our \gls{aip}'s decision circuit reasons over prover messages, we require a similar mapping of the committed witness to the prover messages.
Let
\begin{equation*}
    \msgMap : \textstyle\bigcup_{t} \sIndexSetCom^t \to \bigl\{ (i, j) : i \in [\nRounds],\, j \in [|\pMsg_i|] \bigr\}
\end{equation*}
be an injective map, fixed by the \gls{aip}, that assigns to each committed witness position $u$ the round $i$ and the position $j$ at which the honest prover sends $\sWitness[u]$.
We define the projection
\begin{equation*}
    \begin{aligned}
    \msgProj_{\sIndexSetCom}(\pMsgs) &:= (\scommittedWitness_1, \ldots, \scommittedWitness_\ell), \\
    \scommittedWitness_t &:= \bigl( \pMsg_i[j] \bigr)_{(i,j) = \msgMap(u),\ u \in \sIndexSetCom^t},
    \end{aligned}
\end{equation*}
where the order of each $\scommittedWitness_t$ follows the order of $\sIndexSetCom^t$.
For honestly generated messages, $\msgProj_{\sIndexSetCom}(\pMsgs)$ equals the committed sub-witnesses of \sWitness from \Cref{def:cpsnark}.

The security proof uses that an extractor can read the witness from the prover messages of the \gls{aip}, which we formalize as a witness-carrying \gls{aip}.
A witness-carrying \gls{aip} defines an extractor that, for every accepting transcript, outputs a witness whose committed sub-witnesses equal $\msgProj_{\sIndexSetCom}(\pMsgs)$.
The \glspl{aip} that we consider satisfy \Cref{def:witness-carrying} by construction, because the prover sends the witness as part of its messages and the extractor directly reads the witness from them.

\begin{definition}[Witness-Carrying \gls{aip}]
\label{def:witness-carrying}
    Let \aip be an \gls{aip} for \sRelation with index set \sIndexSetCom and map \msgMap.
    \aip is \emph{witness-carrying with respect to \sIndexSetCom} with knowledge error \negAIP if there is an efficient extractor \extAIP such that for every prover \malProver and every \sInstance, the following holds.
    Let $\tilde{\pMsg}_1, \ldots, \tilde{\pMsg}_\nRounds$ and $\zkc_1, \ldots, \zkc_{\nRounds-1}$ be the transcript of one execution of \malProver with the honest verifier, and let $\tilde{\sWitness} := \extAIP(\sInstance, \tilde{\pMsg}_1, \ldots, \tilde{\pMsg}_\nRounds)$.
    Then
    \begin{align*}
        \Pr\bigl[\, &\dCircuit \text{ accepts the transcript} \;\land \\
        &\bigl( (\sInstance, \tilde{\sWitness}) \notin \sRelation \;\lor \\
        &\;\; \exists\, u \in \textstyle\bigcup_{t} \sIndexSetCom^t : \tilde{\sWitness}[u] \neq \tilde{\pMsg}_i[j] \bigr) \,\bigr] \le \negAIP ,
    \end{align*}
    where $(i, j) = \msgMap(u)$.
\end{definition}

We restate the security theorem of \oursystem (\Cref{thm:artemis:aip}) and prove it.

\thmArtemisMain*

\begin{proof}
    We prove completeness, knowledge soundness, and honest-verifier zero-knowledge of \zkSARKCP for the relation \sRelationCP.

    \fakeparagraphnovskip{Setting}
    Fix $\aip \in \arithK$ as in the theorem, an index set \sIndexSetCom, and a commitment key \ck output by \PCSetup.
    The prover of \aip sends the witness as part of its messages, so \aip is witness-carrying with respect to \sIndexSetCom (\Cref{def:witness-carrying}); let \extAIP and \negAIP be its extractor and knowledge error.
    Let \aipObj be the \nRounds-round \gls{aip} that \aip encodes, with decision circuit \dCircuit and challenges \aipChals.
    Let $\aipObj^\prime$, \aipExt, and $(\textsf{pk}, \textsf{vk})$ be as in the setup of \Cref{fig:compiler:artemis:meta}.
    \arithK is universal, so $\aipObj^\prime$ has an encoding $\aipExt \in \arithK$.
    The three properties of \zkSARK hold for every instance of \arithK, so they hold for \aipExt.
    In $\aipObj^\prime$, the prover sends $\pMsg_{\nRounds+1} = \sMask$ and $\pMsg_{\nRounds+2} = \sEvalValue$, the verifier sends no challenge after round \nRounds, and $(\alpha, \sEvalPoint)$ is the challenge of round $\nRounds+1$.
    \zkSARK can reveal from round $\nRounds+2$, if necessary after appending an empty round to $\aipObj^\prime$.
    For messages \pMsgs, let $(\scommittedWitness_1, \ldots, \scommittedWitness_\ell) := \msgProjI(\pMsgs)$ and let \spolySingle be the polynomial with coefficient vector $\scommittedWitness_i$.
    The decision circuit $\dCircuit^\prime$ accepts if and only if \dCircuit accepts and
    \begin{equation}
        \label{eq:proof:inside}
        \pMsg_{\nRounds+2} = \pMsg_{\nRounds+1} + \textstyle\sum_{i=1}^{\ell} \alpha^i \cdot \spolySingle(\sEvalPoint).
    \end{equation}

    All commitments use the degree bound \pcDbound, and every committed sub-witness has at most $\pcDbound + 1$ entries.
    We use four properties of \PCScheme.
    \begin{enumerate}[label=(\roman*),leftmargin=18pt,topsep=2pt]
        \item \PCScheme is $\sfield$-linear (\Cref{def:homomorphic_commitment}).
        \item $\PCVerify(\ck, \scom, \pcDbound, \spolynomial, \sRandomness) = 1$ if and only if $\deg \spolynomial \le \pcDbound$ and $\scom = \PCCommit(\ck, \spolynomial, \pcDbound, \sRandomness)$, as for KZG~\cite{Kate2010-px} and Pedersen commitments.
        We call such a pair $(\spolynomial, \sRandomness)$ an \emph{opening} of \scom.
        \item \PCScheme is binding with error \negBind: a PPT adversary outputs a commitment together with two openings to different polynomials with probability at most \negBind.
        \item \PCScheme is extractable with error \negPC: for every PPT adversary \advPC that outputs $(\scom, z, v, \pi)$ on input \ck, there is an expected polynomial-time extractor \extPC that receives the state of \advPC, including its random coins, and outputs $(\spolynomial, \sRandomness)$ such that
        \begin{align*}
            \Pr\bigl[\, &\PCCheck(\ck, \scom, \pcDbound, z, v, \pi) = 1 \;\land \\
            &\lnot \bigl( (\spolynomial, \sRandomness) \text{ opens } \scom \,\land\, \spolynomial(z) = v \bigr) \,\bigr] \le \negPC .
        \end{align*}
        This property follows the extractability notion for polynomial commitments of Chiesa et al.~\cite{Chiesa2020-marlin}.
    \end{enumerate}

    \fakeparagraphnovskip{Completeness}
    Let $((\sInstance, \scomList), (\sWitness, \srandList)) \in \sRelationCP$, i.e., $(\sInstance, \sWitness) \in \sRelation$ and $(\spolySingle, \sRandomness_i)$ opens $\scom_i$ for every $i \in [\ell]$, where \spolySingle is the polynomial with coefficient vector $\scommittedWitness_i$, the $i$-th committed sub-witness of \sWitness.
    We follow the steps of the protocol in \Cref{fig:compiler:artemis:meta}.
    In Step~\ref{step:zkp}, the prover holds $(\sInstance, (\sWitness, \sMask)) \in \sRelationMu$, and $\aipObj^\prime$ is complete for \sRelationMu by Lemma~\ref{lemma:aipext}.
    Completeness of \zkSARK for \aipExt implies that \zkVerifier accepts $\tau$.
    For honestly generated messages, $\msgProjI(\pMsgs)$ equals the committed sub-witnesses of \sWitness, so by \Cref{eq:proof:inside} the honest message of round $\nRounds+2$ is
    \begin{equation*}
        \sEvalValue = \sMask + \textstyle\sum_{i=1}^{\ell} \alpha^i \cdot \spolySingle(\sEvalPoint) = \mathbf{w}^*(\sEvalPoint).
    \end{equation*}
    In Step~\ref{step:zkpreveal}, the prover reveals coordinate $\hIdx = 1$ of round $\hRound = \nRounds+2$, which is \sEvalValue, and $V^{(\hRound,\hIdx)}_\mathsf{Reveal}$ accepts by reveal completeness (\Cref{def:vr-snark}).
    In Step~\ref{step:accept}, the verifier computes
    \begin{equation*}
        \begin{aligned}
        \scom &= \commitmentCombination \\
              &\stackrel{(a)}{=} \PCCommit\bigl(\ck,\ \spolyMask + \textstyle\sum_{i=1}^\ell \alpha^i \cdot \spolySingle,\ \pcDbound, \\
              &\qquad\qquad\qquad \sRandomness_\sMask + \textstyle\sum_{i=1}^\ell \alpha^i \cdot \sRandomness_i\bigr) \\
              &= \PCCommit(\ck, \mathbf{w}^*, \pcDbound, \sRandomness^*),
        \end{aligned}
    \end{equation*}
    where (a) holds by properties~(i) and~(ii).
    The polynomial $\mathbf{w}^*$ has degree at most \pcDbound, because \spolyMask is constant and every \spolySingle has degree at most \pcDbound.
    The prover computes \piCommitment for the polynomial $\mathbf{w}^*$, the randomness $\sRandomness^*$, the point \sEvalPoint, and the value $\sEvalValue = \mathbf{w}^*(\sEvalPoint)$.
    Completeness of \PCScheme implies that \PCCheck accepts, so $\mathcal{V}$ accepts.

    \fakeparagraphnovskip{Knowledge Soundness}
    Fix an instance \sInstance.
    Let \adv be a PPT adversary that receives $(\textsf{pk}, \textsf{vk})$, outputs commitments \scomList, and then interacts with $\mathcal{V}(\textsf{vk}, (\sInstance, \scomList))$ in the role of the prover.
    We construct an expected polynomial-time extractor \extCompiler that receives the state of \adv, including its random coins (cf.~\Cref{def:zk-sarks}).
    As in \Cref{def:message-ext}, \extCompiler first runs one complete interaction with \adv in the role of $\mathcal{V}$, the \emph{main transcript} $T_0$, then rewinds \adv, and finally outputs $(\tilde{\sWitness}, \tilde{\sRandomness}_1, \ldots, \tilde{\sRandomness}_\ell)$.
    We show that
    \begin{equation}
        \label{eq:proof:ks}
        \begin{aligned}
        \Pr\bigl[\, &\mathcal{V} \text{ accepts } T_0 \;\land \\
        &\bigl((\sInstance, \scomList), (\tilde{\sWitness}, \tilde{\sRandomness}_1, \ldots, \tilde{\sRandomness}_\ell)\bigr) \notin \sRelationCP \,\bigr] \le \kappa
        \end{aligned}
    \end{equation}
    for a negligible knowledge error $\kappa$ that we compute at the end of the proof.

    The proof has three parts.
    First, we show that message extractability and reveal soundness ensure that an accepting transcript reveals the evaluation of an \emph{inside polynomial} that \adv fixes before it receives $(\alpha, \sEvalPoint)$ (\Cref{claim:inside}).
    Second, we use rewinding to extract an opening for the instance commitments.
    \extCompiler rewinds \adv to obtain $\ell+1$ further accepting transcripts with distinct $\alpha$, opens their combined commitments, and recovers an opening of every $\scom_i$ (\Cref{claim:rewinding,claim:openings}).
    Third, by Lemma~\ref{lemma:schwartz_zippel}, we show that the opened polynomials equal the inside polynomials (\Cref{claim:sz}), and the extractor \extAIP of the witness-carrying \gls{aip} turns the extracted messages into a valid witness (\Cref{claim:witness}).

    \smallskip\noindent\emph{Notation.}
    An execution of \adv with $\mathcal{V}$ consists of $\scom_\sMask$ (Step~\ref{step:mask}), the transcript $\tau$ of \zkSARK with the exposed challenges \aipChals and $(\alpha, \sEvalPoint)$ (Step~\ref{step:zkp}), the revealed value \sEvalValue with the transcript of the reveal protocol (Step~\ref{step:zkpreveal}), and \piCommitment (Step~\ref{step:link}).
    Let \stA be the state of \adv after Step~\ref{step:mask}.
    The state \stA determines \scomList and $\scom_\sMask$.
    From \stA on, \adv is a prover for \zkSARK on \aipExt and \sInstance.
    Let \extSARK be the extractor of this prover and let $f_1, \ldots, f_{\nRounds+2}$ be the functions of \Cref{def:message-ext}.
    For challenges \aipChals, we define
    \begin{equation*}
        \begin{aligned}
        \hat{\pMsg}_s &:= f_s(\stA, \zkc_1, \ldots, \zkc_{s-1}) \quad \text{for } s \in [\nRounds+1], \\
        (\hat{\scommittedWitness}_1, \ldots, \hat{\scommittedWitness}_\ell) &:= \msgProjI(\hat{\pMsg}_1, \ldots, \hat{\pMsg}_\nRounds), \\
        \polyIn(Y, X) &:= \hat{\pMsg}_{\nRounds+1} + \textstyle\sum_{i=1}^{\ell} Y^i \cdot \hat{\spolySingle}(X),
        \end{aligned}
    \end{equation*}
    where $\hat{\spolySingle}$ is the polynomial with coefficient vector $\hat{\scommittedWitness}_i$.
    The inside polynomial \polyIn is a function of \stA and \aipChals only.
    Let \stAlpha be the state of \adv immediately before \zkVerifier sends $(\alpha, \sEvalPoint)$.
    The state \stAlpha determines \stA and \aipChals, and therefore \polyIn.
    A \emph{continuation} from \stAlpha samples $(\alpha, \sEvalPoint) \sample \sfield^2$ and completes the execution in the role of $\mathcal{V}$ with fresh coins.
    The protocol is public-coin, so a continuation from \stAlpha has the distribution of the remainder of an execution with $\mathcal{V}$ that reaches \stAlpha.
    Let $q(\stAlpha)$ be the probability that a continuation from \stAlpha is accepting.

    \smallskip\noindent\emph{The extractor \extCompiler.}
    \begin{enumerate}[label=\textbf{E\arabic*.},ref=E\arabic*,leftmargin=20pt,topsep=2pt]
        \item \label{ks:step:main} \textbf{Main transcript.}
        Run \adv to obtain \scomList and $\scom_\sMask$, and let \stA be its state.
        Run \extSARK on \adv from \stA.
        The main transcript $\tau_0$ of \extSARK contains the challenges \aipChals and $(\alpha_0, \sEvalPoint_0)$; record the state \stAlpha of \adv immediately before \extSARK sends $(\alpha_0, \sEvalPoint_0)$.
        \extSARK outputs $\tilde{\pMsg}_1, \ldots, \tilde{\pMsg}_{\nRounds+2}$.
        Complete Steps~\ref{step:zkpreveal}--\ref{step:accept} after $\tau_0$ in the role of $\mathcal{V}$ to obtain $T_0$.
        Abort if $\mathcal{V}$ rejects $T_0$.

        \item \label{ks:step:fork} \textbf{Rewinding.}
        Run the tree finder of Attema et al.~\cite[Lemma~5]{Attema2021-ack} on the challenge $\alpha$:
        sample $\alpha$ from $\sfield \setminus \{\alpha_0\}$ without replacement, rewind \adv to \stAlpha, and repeat the remainder of $T_0$ with $\alpha$ in place of $\alpha_0$ and with all other verifier coins of $T_0$ unchanged, in particular $\sEvalPoint_0$.
        Stop as soon as $\ell+1$ of these transcripts, $T_1, \ldots, T_{\ell+1}$, are accepting, and abort if $\sfield \setminus \{\alpha_0\}$ is exhausted earlier.
        For $j \in [\ell+1]$, let $\alpha_j$, $\sEvalValue_j$, and $\piCommitment^{(j)}$ be the challenge, the revealed value, and the evaluation proof of $T_j$, and let $\scom^{(j)} := \scom_\sMask + \sum_{i=1}^{\ell} \alpha_j^i \cdot \scom_i$.

        \item \label{ks:step:pc} \textbf{Opening the combined commitments.}
        For $j \in [\ell+1]$, let $\advPC^{(j)}$ be the adversary against the extractability of \PCScheme that receives \ck, runs Steps~\ref{ks:step:main} and~\ref{ks:step:fork}, and outputs $(\scom^{(j)}, \sEvalPoint_0, \sEvalValue_j, \piCommitment^{(j)})$.
        Run the extractor \extPC of $\advPC^{(j)}$ on the coins of Steps~\ref{ks:step:main} and~\ref{ks:step:fork} to obtain $(\mathbf{p}^{(j)}, \sRandomness^{(j)})$.

        \item \label{ks:step:solve} \textbf{Opening the instance commitments.}
        Let $M \in \sfield^{(\ell+1) \times (\ell+1)}$ be the Vandermonde matrix with entries $M_{j,i} := \alpha_j^i$ for $j \in [\ell+1]$ and $i \in \{0, \ldots, \ell\}$.
        $M$ is invertible, because the $\alpha_j$ are pairwise distinct.
        Compute
        \begin{equation*}
            \begin{aligned}
            (\mathbf{p}_0, \ldots, \mathbf{p}_\ell)^\top &:= M^{-1} (\mathbf{p}^{(1)}, \ldots, \mathbf{p}^{(\ell+1)})^\top, \\
            (\tilde{\sRandomness}_0, \ldots, \tilde{\sRandomness}_\ell)^\top &:= M^{-1} (\sRandomness^{(1)}, \ldots, \sRandomness^{(\ell+1)})^\top .
            \end{aligned}
        \end{equation*}

        \item \label{ks:step:aip} \textbf{Witness.}
        Compute $\tilde{\sWitness} := \extAIP(\sInstance, \tilde{\pMsg}_1, \ldots, \tilde{\pMsg}_\nRounds)$ and output $(\tilde{\sWitness}, \tilde{\sRandomness}_1, \ldots, \tilde{\sRandomness}_\ell)$.
    \end{enumerate}

    \smallskip\noindent\emph{Analysis.}
    We call an accepting transcript with challenges $(\alpha, \sEvalPoint)$ and revealed value \sEvalValue \emph{consistent} if $\sEvalValue = \polyIn(\alpha, \sEvalPoint)$.

    \begin{claim}[Inside Polynomial]
    \label{claim:inside}
    An execution of \adv with $\mathcal{V}$ is accepting and not consistent with probability at most $2\negSARK + \negReveal$.
    \end{claim}
    \begin{proof}
    Consider the experiment of reveal soundness (\Cref{def:vr-snark}) for \adv from \stA and the coordinate $(\hRound, \hIdx) = (\nRounds+2, 1)$.
    The main transcript $\tau$ and the reveal protocol of this experiment have the distribution of Steps~\ref{step:zkp} and~\ref{step:zkpreveal} of an execution with $\mathcal{V}$, and \extSARK outputs $\tilde{\pMsg}_1, \ldots, \tilde{\pMsg}_{\nRounds+2}$.
    Assume that the execution is accepting, so \zkVerifier accepts $\tau$ and $V^{(\hRound,\hIdx)}_\mathsf{Reveal}$ accepts \sEvalValue.
    Extraction (\Cref{def:message-ext}) implies that $\dCircuit^\prime$ accepts the extracted messages with the challenges of $\tau$, except with probability \negSARK.
    The extracted messages then satisfy \Cref{eq:proof:inside}.
    Reveal soundness implies $\sEvalValue = \tilde{\pMsg}_{\nRounds+2}$, except with probability \negReveal.
    Prefix determination implies $\tilde{\pMsg}_s = \hat{\pMsg}_s$ for all $s \in [\nRounds+1]$, except with probability \negSARK, so the right-hand side of \Cref{eq:proof:inside} equals $\polyIn(\alpha, \sEvalPoint)$.
    \end{proof}

    \begin{claim}[Rewinding]
    \label{claim:rewinding}
    The probability that $\mathcal{V}$ accepts $T_0$ and Step~\ref{ks:step:fork} aborts is at most $(\ell+1)/|\sfield|$.
    Steps~\ref{ks:step:main} and~\ref{ks:step:fork} run at most $\ell+2$ executions of \adv in expectation.
    \end{claim}
    \begin{proof}
    Let $H$ be the binary matrix with one column for every $\alpha \in \sfield$ and one row for every choice of the remaining coins of an execution, i.e., the coins of \adv, of the setup, and of $\mathcal{V}$.
    An entry of $H$ is $1$ if and only if the execution with these coins is accepting.
    The main transcript $T_0$ is a uniformly random entry of $H$, and Step~\ref{ks:step:fork} samples further entries of the row of $T_0$ without replacement until this row contains $\ell+2$ known entries equal to $1$.
    Steps~\ref{ks:step:main} and~\ref{ks:step:fork} are therefore the tree-finding algorithm of Attema et al.~\cite[Lemma~5]{Attema2021-ack} for a single challenge with $k_1 = \ell+2$ and $N = |\sfield|$.\footnote{Statement numbers refer to the full version of~\cite{Attema2021-ack}.}
    Let $\varepsilon$ be the probability that $\mathcal{V}$ accepts $T_0$.
    The algorithm queries at most $k_1$ entries of $H$ in expectation and succeeds with probability at least $\varepsilon - (k_1 - 1)/N$.
    It succeeds only if $\mathcal{V}$ accepts $T_0$, so $\mathcal{V}$ accepts $T_0$ and the algorithm aborts with probability at most $(k_1 - 1)/N = (\ell+1)/|\sfield|$.
    \end{proof}

    \begin{claim}[Instance Commitment Openings]
    \label{claim:openings}
    If Step~\ref{ks:step:fork} succeeds, then $(\mathbf{p}_0, \tilde{\sRandomness}_0)$ opens $\scom_\sMask$ and $(\mathbf{p}_i, \tilde{\sRandomness}_i)$ opens $\scom_i$ for every $i \in [\ell]$, except with probability $(\ell+1)\negPC$.
    \end{claim}
    \begin{proof}
    Every $T_j$ is accepting, so \PCCheck accepts $\piCommitment^{(j)}$ for the commitment $\scom^{(j)}$, the point $\sEvalPoint_0$, and the value $\sEvalValue_j$.
    Property~(iv) for $\advPC^{(j)}$ implies that $(\mathbf{p}^{(j)}, \sRandomness^{(j)})$ opens $\scom^{(j)}$, except with probability \negPC for each $j \in [\ell+1]$.
    Let $\vec{\scom} := (\scom_\sMask, \scom_1, \ldots, \scom_\ell)^\top$, so that $(\scom^{(1)}, \ldots, \scom^{(\ell+1)})^\top = M \vec{\scom}$.
    For every $i \in \{0, \ldots, \ell\}$,
    \begin{equation*}
        \begin{aligned}
        \PCCommit(\ck, \mathbf{p}_i, \pcDbound, \tilde{\sRandomness}_i)
        &\stackrel{(a)}{=} \textstyle\sum_{j} (M^{-1})_{i,j} \cdot \scom^{(j)} \\
        &= (M^{-1} M \vec{\scom})_i = \vec{\scom}_i ,
        \end{aligned}
    \end{equation*}
    where (a) holds by properties~(i) and~(ii), and index $0$ refers to $\scom_\sMask$.
    Every $\mathbf{p}_i$ has degree at most \pcDbound, because it is a linear combination of the $\mathbf{p}^{(j)}$.
    \end{proof}

    \begin{claim}[Inside Equals Outside]
    \label{claim:sz}
    If Steps~\ref{ks:step:fork}--\ref{ks:step:solve} succeed, then $\mathbf{p}_0 = \hat{\pMsg}_{\nRounds+1}$ and $\mathbf{p}_i = \hat{\spolySingle}$ for every $i \in [\ell]$, except with probability $(\ell + \pcDbound + 1)/|\sfield| + 2\negSARK + \negReveal + \negPC + \negBind$.
    \end{claim}
    \begin{proof}
    Consider the polynomial
    \begin{equation*}
        \polyBiv(Y, X) := \textstyle\sum_{i=0}^{\ell} Y^i \cdot \mathbf{p}_i(X) - \polyIn(Y, X)
    \end{equation*}
    of degree at most $\ell$ in $Y$ and at most \pcDbound in $X$.
    The claim states that $\polyBiv = 0$.
    The transcripts $T_0, \ldots, T_{\ell+1}$ determine \polyBiv, so we evaluate \polyBiv at the challenges of a transcript that is sampled after \polyBiv is fixed.
    To this end, we extend the experiment by a step that does not change the output of \extCompiler:
    after Step~\ref{ks:step:solve}, run at most $|\sfield|$ further continuations from \stAlpha and stop at the first accepting one, $T_*$, with challenges $(\alpha_*, \sEvalPoint_*)$, revealed value $\sEvalValue_*$, evaluation proof $\piCommitment^{(*)}$, and combined commitment $\scom^{(*)} := \scom_\sMask + \sum_{i=1}^{\ell} \alpha_*^i \cdot \scom_i$.
    Let $(\mathbf{p}^{(*)}, \sRandomness^{(*)})$ be the output of the extractor \extPC of the adversary that runs the extended experiment and outputs $(\scom^{(*)}, \sEvalPoint_*, \sEvalValue_*, \piCommitment^{(*)})$.

    Fix \stAlpha and let $q := q(\stAlpha)$.
    The remainder of $T_0$ after \stAlpha is a continuation, so $\mathcal{V}$ accepts $T_0$ with probability $q$.
    The experiment therefore reaches the additional step with probability at most $q$, and the additional step then runs at most $1/q$ continuations in expectation, where the experiment runs a continuation only if all previous ones are rejecting.
    Summing over the continuations, the additional step runs a continuation with a given property with probability at most the probability that a single continuation from \stAlpha has this property.
    We apply this bound twice.
    First, by \Cref{claim:inside}, a continuation is accepting and not consistent with probability at most $2\negSARK + \negReveal$ on average over \stAlpha, so $T_*$ is consistent except with this probability.
    Second, \polyBiv is fixed before the additional step, and every continuation samples $(\alpha, \sEvalPoint)$ uniformly and independently of \polyBiv.
    By Lemma~\ref{lemma:schwartz_zippel}, $\polyBiv \neq 0$ and $\polyBiv(\alpha_*, \sEvalPoint_*) = 0$ hold together with probability at most $(\ell + \pcDbound)/|\sfield|$.
    The additional step finds no accepting continuation with probability $(1-q)^{|\sfield|}$, and $q (1-q)^{|\sfield|} \le 1/|\sfield|$.

    It remains to show $\polyBiv(\alpha_*, \sEvalPoint_*) = 0$.
    Property~(iv) implies that $(\mathbf{p}^{(*)}, \sRandomness^{(*)})$ opens $\scom^{(*)}$ and $\mathbf{p}^{(*)}(\sEvalPoint_*) = \sEvalValue_*$, except with probability \negPC.
    \Cref{claim:openings} and properties~(i) and~(ii) imply that $\bigl(\sum_{i=0}^{\ell} \alpha_*^i \cdot \mathbf{p}_i,\ \sum_{i=0}^{\ell} \alpha_*^i \cdot \tilde{\sRandomness}_i\bigr)$ opens $\scom^{(*)}$ as well.
    Binding implies $\mathbf{p}^{(*)} = \sum_{i=0}^{\ell} \alpha_*^i \cdot \mathbf{p}_i$, except with probability \negBind.
    Since $T_*$ is consistent,
    \begin{equation*}
        \polyBiv(\alpha_*, \sEvalPoint_*) = \mathbf{p}^{(*)}(\sEvalPoint_*) - \polyIn(\alpha_*, \sEvalPoint_*) = \sEvalValue_* - \sEvalValue_* = 0. \qedhere
    \end{equation*}
    \end{proof}

    \begin{claim}[Witness]
    \label{claim:witness}
    Except with probability $2\negSARK + \negAIP$, if $\mathcal{V}$ accepts $T_0$, then $(\sInstance, \tilde{\sWitness}) \in \sRelation$ and the committed sub-witnesses of $\tilde{\sWitness}$ are $\hat{\scommittedWitness}_1, \ldots, \hat{\scommittedWitness}_\ell$.
    \end{claim}
    \begin{proof}
    Consider the prover $\hat{\sprover}$ for \aipObj that sends $\hat{\pMsg}_s = f_s(\stA, \zkc_1, \ldots, \zkc_{s-1})$ in round $s \in [\nRounds]$.
    \zkSARK is challenge-exposing, so the challenges \aipChals of $\tau_0$ are uniform, and $\hat{\pMsg}_1, \ldots, \hat{\pMsg}_\nRounds$ with \aipChals is the transcript of an execution of $\hat{\sprover}$ with the honest verifier of \aipObj.
    \Cref{def:witness-carrying} holds for every prover, in particular for $\hat{\sprover}$:
    except with probability \negAIP, if \dCircuit accepts this transcript, then $\hat{\sWitness} := \extAIP(\sInstance, \hat{\pMsg}_1, \ldots, \hat{\pMsg}_\nRounds)$ satisfies $(\sInstance, \hat{\sWitness}) \in \sRelation$, and the committed sub-witnesses of $\hat{\sWitness}$ are $\msgProjI(\hat{\pMsg}_1, \ldots, \hat{\pMsg}_\nRounds) = (\hat{\scommittedWitness}_1, \ldots, \hat{\scommittedWitness}_\ell)$.
    We relate $\hat{\sWitness}$ to the output of \extCompiler.
    If $\mathcal{V}$ accepts $T_0$, then \zkVerifier accepts $\tau_0$, and extraction implies that $\dCircuit^\prime$, and hence \dCircuit, accepts $\tilde{\pMsg}_1, \ldots, \tilde{\pMsg}_\nRounds$ with \aipChals, except with probability \negSARK.
    Prefix determination implies $\tilde{\pMsg}_s = \hat{\pMsg}_s$ for all $s \in [\nRounds]$, except with probability \negSARK.
    In this case, \dCircuit accepts the transcript of $\hat{\sprover}$ and $\tilde{\sWitness} = \hat{\sWitness}$.
    \end{proof}

    \smallskip\noindent\emph{Knowledge error.}
    Assume that $\mathcal{V}$ accepts $T_0$ and that none of the failure events of \Cref{claim:rewinding,claim:openings,claim:sz,claim:witness} occurs.
    \extCompiler then does not abort.
    By \Cref{claim:witness}, $(\sInstance, \tilde{\sWitness}) \in \sRelation$ and the $i$-th committed sub-witness of $\tilde{\sWitness}$ is $\hat{\scommittedWitness}_i$.
    By \Cref{claim:openings,claim:sz}, $(\hat{\spolySingle}, \tilde{\sRandomness}_i)$ opens $\scom_i$ for every $i \in [\ell]$.
    Hence $\bigl((\sInstance, \scomList), (\tilde{\sWitness}, \tilde{\sRandomness}_1, \ldots, \tilde{\sRandomness}_\ell)\bigr) \in \sRelationCP$.
    A union bound over the failure events yields \Cref{eq:proof:ks} with
    \begin{equation*}
        \begin{aligned}
        \kappa \;\le\;\; &\frac{2\ell + \pcDbound + 2}{|\sfield|} + 4\negSARK + \negReveal \\
        &+ (\ell+2)\,\negPC + \negBind + \negAIP ,
        \end{aligned}
    \end{equation*}
    which is negligible in \secpar, because $\ell$ and \pcDbound are at most $\poly$.
    The first term is statistical.

    \smallskip\noindent\emph{Running time.}
    \extCompiler runs \extSARK once, at most $\ell+1$ rewound executions in expectation (\Cref{claim:rewinding}), \extPC $\ell+1$ times, and \extAIP once, and it solves one Vandermonde system over polynomials of degree at most \pcDbound with $O(\ell^2 \pcDbound)$ field operations.
    \extCompiler therefore runs in expected polynomial time.
    The additional step in the proof of \Cref{claim:sz} runs at most one continuation in expectation, so the adversaries against extractability and binding of \PCScheme run in expected polynomial time as well.
    Truncating an adversary that succeeds with probability $p$ after $2/p$ times its expected running time yields a strict polynomial-time adversary that succeeds with probability at least $p/2$ by Markov's inequality.
   As a result, the extractability and binding properties of \PCScheme against PPT adversaries ensure a bound on \negPC and \negBind.

    \fakeparagraphnovskip{Zero-knowledge}
    We show that \zkSARKCP is zero-knowledge for \sRelationCP in the sense of \Cref{def:zk-sarks}: a simulator \simu that receives the coins $\vec{\theta}$ of $\mathcal{V}$ outputs a transcript that is indistinguishable from the transcript of an execution of $\mathcal{P}$ with $\mathcal{V}$.
    The coins $\vec{\theta}$ are the coins of \zkVerifier and $V^{(\hRound,\hIdx)}_\mathsf{Reveal}$, and they determine $(\alpha, \sEvalPoint)$, because \zkSARK is challenge-exposing.
    In the real experiment, $(\textsf{pk}, \textsf{vk})$ is the output of \zkSetup, $\adv_1(\textsf{pk}, \textsf{vk})$ outputs $((\sInstance, \scomList), (\sWitness, \srandList)) \in \sRelationCP$ with a state $\mathsf{st}$, and $\adv_2$ receives $\mathsf{st}$ and the transcript $(\scom_\sMask, \tau, \sEvalValue, \tau_\mathsf{R}, \piCommitment)$ of an execution of $\mathcal{P}$ with $\mathcal{V}$ on uniform coins $\vec{\theta}$, where $\tau_\mathsf{R}$ is the transcript of the reveal protocol.
    In the ideal experiment, the setup of \simu outputs $(\textsf{pk}, \textsf{vk})$ with a trapdoor \td, and $\adv_2$ receives $\mathsf{st}$ and the output of \simu on \td, $(\sInstance, \scomList)$ and $\vec{\theta}$.
    We use hiding of \PCScheme in the following form.
    \begin{enumerate}[label=(\roman*),leftmargin=18pt,topsep=2pt]
        \setcounter{enumi}{4}
        \item \PCScheme is hiding with error \negHide: a PPT adversary that receives \ck and outputs a point $z$ and polynomials $\spolynomial, \spolynomial^\prime$ of degree at most \pcDbound with $\spolynomial(z) = \spolynomial^\prime(z) = v$ distinguishes
        $\bigl(\PCCommit(\ck, \spolynomial, \pcDbound, \sRandomness),\ \PCProve(\ck, \spolynomial, \pcDbound, z, v, \sRandomness)\bigr)$
        from the same pair for $\spolynomial^\prime$ with advantage at most \negHide, where \sRandomness is uniform.
        The hiding notion of Chiesa et al.~\cite{Chiesa2020-marlin} implies this property, because their simulator produces both pairs from $(z, v)$ alone.
    \end{enumerate}

    \smallskip\noindent\emph{The simulator \simu.}
    The transcript of \zkSARKCP contains the transcript $\tau$ of \zkSARK and the transcript $\tau_\mathsf{R}$ of its reveal protocol for the value \sEvalValue.
    \zkSARK is soundly-revealing for \aipExt, so \Cref{def:vr-snark} provides a simulator \simSARK that outputs both transcripts for a given \sEvalValue (reveal zero-knowledge for the coordinate $(\hRound, \hIdx)$).
    The setup of \simu runs \zkSetup of \Cref{fig:compiler:artemis:meta} with the setup of \simSARK in place of Step~\ref{step:setup:compile}, and outputs $(\textsf{pk}, \textsf{vk})$ and the trapdoor \td{} of \simSARK.
    The commitment key \ck is an input of \zkSetup, so \simu holds no trapdoor for \PCScheme.
    On input \td, $(\sInstance, \scomList)$ and $\vec{\theta}$, \simu reads $(\alpha, \sEvalPoint)$ from $\vec{\theta}$ and proceeds as follows.
    \begin{enumerate}[label=\textbf{S\arabic*.},ref=S\arabic*,leftmargin=20pt,topsep=2pt]
        \item \label{zk:step:sample} Sample $\sEvalValue, \sRandomness_\sEvalValue \sample \sfield$ and compute $\scom_\sEvalValue \leftarrow \PCCommit(\ck, \polyRho, \pcDbound, \sRandomness_\sEvalValue)$, where \polyRho is the constant polynomial \sEvalValue.
        \item \label{zk:step:mask} Set $\scom_\sMask := \scom_\sEvalValue - \sum_{i=1}^{\ell} \alpha^i \cdot \scom_i$.
        \item \label{zk:step:sark} Compute $(\tau, \tau_\mathsf{R}) \leftarrow \simSARK.\textsf{Sim}(\td, \sInstance, \vec{\theta}, \sEvalValue)$.
        \item \label{zk:step:link} Compute $\piCommitment \leftarrow \PCProve(\ck, \polyRho, \pcDbound, \sEvalPoint, \sEvalValue, \sRandomness_\sEvalValue)$ and output $(\scom_\sMask, \tau, \sEvalValue, \tau_\mathsf{R}, \piCommitment)$.
    \end{enumerate}
    The verifier computes $\scom = \commitmentCombination = \scom_\sEvalValue$ in Step~\ref{step:accept}, so the honestly computed \piCommitment of Step~\ref{zk:step:link} is an evaluation proof for \scom.
    Step~\ref{zk:step:mask} uses $\alpha$ to compute the first message of the protocol, which restricts \simu to an honest verifier.

    \smallskip\noindent\emph{Indistinguishability.}
    Let $\mathsf{Hyb}$ be the real experiment in which the setup of \simSARK generates $(\ipk, \ivk)$ and $\simSARK.\textsf{Sim}(\td, \sInstance, \vec{\theta}, \sEvalValue)$ generates $(\tau, \tau_\mathsf{R})$, for the value $\sEvalValue = \sMask + \sum_{i=1}^{\ell} \alpha^i \cdot \spolySingle(\sEvalPoint)$ of the honest prover.

    \begin{claim}
    \label{claim:zk:reveal}
    $\adv$ distinguishes the real experiment from $\mathsf{Hyb}$ with advantage at most \negRZK, the advantage against reveal zero-knowledge of \zkSARK for \aipExt.
    \end{claim}
    \begin{proof}
    We construct an adversary $(\mathcal{B}_1, \mathcal{B}_2)$ against reveal zero-knowledge for \aipExt and $(\hRound, \hIdx)$.
    $\mathcal{B}_1$ receives $(\ipk, \ivk)$, assembles $(\textsf{pk}, \textsf{vk})$ with \ck, runs $\adv_1$, samples $\sMask, \sRandomness_\sMask \sample \sfield$, and outputs \sInstance, the witness $(\sWitness, \sMask)$, which is valid for \aipExt by Lemma~\ref{lemma:aipext}, and the state $(\mathsf{st}, \sWitness, \srandList, \sMask, \sRandomness_\sMask)$.
    $\mathcal{B}_2$ receives $(\tau, \sEvalValue, \tau_\mathsf{R})$, reads $(\alpha, \sEvalPoint)$ from $\tau$, computes $\scom_\sMask$ and \piCommitment as $\mathcal{P}$ does in Steps~\ref{step:mask} and~\ref{step:link}, and runs $\adv_2$ on the resulting transcript.
    The honest prover of $\aipObj^\prime$ sends $\sEvalValue = \sMask + \sum_{i=1}^{\ell} \alpha^i \cdot \spolySingle(\sEvalPoint)$ at the coordinate $(\hRound, \hIdx)$ (\Cref{eq:proof:inside}).
    The experiment $\mathsf{Real}_{\mathcal{B}}$ is therefore the real experiment of \adv, and $\mathsf{Ideal}_{\mathcal{B}}$ is $\mathsf{Hyb}$.
    \end{proof}

    \begin{claim}
    \label{claim:zk:hiding}
    $\adv$ distinguishes $\mathsf{Hyb}$ from the ideal experiment with advantage at most \negHide.
    \end{claim}
    \begin{proof}
    In $\mathsf{Hyb}$, the transcript depends on \sMask and $\sRandomness_\sMask$ only through \sEvalValue, $\scom_\sMask$ and \piCommitment.
    The honest verifier samples $(\alpha, \sEvalPoint)$ independently of $(\sMask, \sRandomness_\sMask)$, so $\sEvalValue = \sMask + \sum_{i=1}^{\ell} \alpha^i \cdot \spolySingle(\sEvalPoint)$ and $\sRandomness^* = \sRandomness_\sMask + \sum_{i=1}^{\ell} \alpha^i \cdot \sRandomness_i$ are uniform, independent of each other, and independent of $\mathsf{st}$ and $\vec{\theta}$.
    Properties~(i) and~(ii) imply $\scom_i = \PCCommit(\ck, \spolySingle, \pcDbound, \sRandomness_i)$ and
    \begin{equation*}
        \scom_\sMask = \PCCommit(\ck, \mathbf{w}^*, \pcDbound, \sRandomness^*) - \textstyle\sum_{i=1}^{\ell} \alpha^i \cdot \scom_i .
    \end{equation*}
    $\mathsf{Hyb}$ therefore samples $\sEvalValue, \sRandomness^* \sample \sfield$ and computes the transcript as \simu does in Steps~\ref{zk:step:mask}--\ref{zk:step:link}, with the pair
    $\bigl(\PCCommit(\ck, \mathbf{w}^*, \pcDbound, \sRandomness^*),\ \PCProve(\ck, \mathbf{w}^*, \pcDbound, \sEvalPoint, \sEvalValue, \sRandomness^*)\bigr)$
    in place of $(\scom_\sEvalValue, \piCommitment)$, where $\mathbf{w}^* = \sEvalValue + \sum_{i=1}^{\ell} \alpha^i \cdot \bigl(\spolySingle - \spolySingle(\sEvalPoint)\bigr)$.
    The ideal experiment computes the same pair for the polynomial \polyRho.
    Both polynomials have degree at most \pcDbound and evaluate to \sEvalValue at \sEvalPoint.
    The adversary against property~(v) receives \ck, runs the setup of \simu and $\adv_1$, samples $\vec{\theta}$ and \sEvalValue, outputs $(\sEvalPoint, \mathbf{w}^*, \polyRho)$, completes the transcript from the pair that it receives, and runs $\adv_2$.
    \end{proof}

    \noindent
    By \Cref{claim:zk:reveal,claim:zk:hiding}, \adv distinguishes the real experiment from the ideal experiment with advantage at most $\negRZK + \negHide$, which is negligible in \secpar.
\end{proof}

\subsection{Polynomial Blinding}
\label{apx:blinding}
To achieve zero-knowledge, many \gls{piop} schemes use blinding to ensure the polynomial evaluations do not reveal information about the witness polynomials.
In the Plonk protocol~\cite{Gabizon2019-plonk}, blinding is performed by masking the witness polynomial $w(X)$ with a degree-$p$ blinding polynomial $b(X)$ together with the vanishing $Z_D$ over a domain $D$:
$$
w^\prime(X) = w(X) + b(X) \cdot Z_D(X).
$$
This method, however, can be inefficient in practice, as it increases the degree of the polynomial.
Applying the blinding polynomial results in a polynomial of degree $n + p$, which must be padded to the next power of 2.
To avoid this overhead, the Halo2 protocol instead appends
$p+1$ random blinding values to the witness polynomial before padding it to the next power of 2~\cite{halo2book}.
This achieves a similar zero-knowledge guarantee as the masking polynomial approach,
while avoiding additional padding---provided the original witness polynomial has enough unused slots to accommodate the blinding values.

\begin{lemma}[Halo2 Blinding~\cite{halo2blindingblog,Pearson-Plonkup2022}]
    \label{thm:blinding}
    Let $\omega \in \sfield$ be an $n$-th primitive root of unity, with $n$ a power of two, forming the domain
    $
    D = \{ \omega^0, \omega^1, \ldots, \omega^{n-1} \}$,
    and let $\{L_i\}_{i=0}^{n-1}$ be the Lagrange basis over $D$
    where $L_i(\omega^i) = 1$ and $L_i(\omega^j) = 0$ for all other elements $\omega^j \in D$.
    For a column of private inputs $\{ w_1, \ldots, w_{d_w} \}$, let $w$ be the (unblinded) witness polynomial, computed by interpolating the private inputs over $D$:
    \ifdefined\isnotextended
    $
    w(X) = \sum_{i=1}^{d_w} w_i \cdot L_{i-1}(X).
    $
    \else
    $$
    w(X) = \sum_{i=1}^{d_w} w_i \cdot L_{i-1}(X).
    $$
    \fi
    Let $p+1$ blinding factors $\{ b_1, \ldots, b_{p+1} \}$ be chosen uniformly at random from $\sfield$, and let
    \ifdefined\isnotextended
    $
    w^\prime(X) = w(X) + \sum_{i=1}^{p+1} b_i \cdot L_{d_w+i-1}(X).
    $
    \else
    $$
    w^\prime(X) = w(X) + \sum_{i=1}^{p+1} b_i \cdot L_{d_w+i-1}(X).
    $$
    \fi
    Then the vector of openings $\bigl(w^\prime(q_1), w^\prime(q_2), \ldots, w^\prime(q_{p})\bigr)$ reveals no information about the witness polynomial $w(X)$,
    where $\{ q_1, q_2, \ldots, q_{p} \} \subset \sfield \setminus D$ are a set of challenge points.

    \end{lemma}
We refer to~\cite[Proposition 2]{Pearson-Plonkup2022} for a proof of this lemma.
The intuition behind the proof is that the blinding term $\sum_{i=1}^{p+1} b_i \cdot L_{d_w+i-1}(X)$ has degree at most $n-1$ and vanishes on the $n-p-1$ points of $D^{\prime\prime} = D \setminus \{\omega^{d_w}, \ldots, \omega^{d_w+p}\}$.
It therefore equals $b(X) \cdot Z_{D^{\prime\prime}}(X)$ for a polynomial $b(X)$ of degree at most $p$, whose $p+1$ coefficients are uniformly random because the map from $(b_1, \ldots, b_{p+1})$ to $b(X)$ is a linear bijection.
As a result, $w^\prime(X)$ agrees with $w(X)$ on $D^\prime = \{ \omega^0, \ldots, \omega^{d_w-1} \} \subseteq D^{\prime\prime}$, and its evaluations at $p$ points outside $D$, where $Z_{D^{\prime\prime}}$ is nonzero, are masked by $p$ evaluations of the uniformly random degree-$p$ polynomial $b(X)$.

\bibliographystyle{IEEEtran}
\bibliography{references,additional_references,extended}

\else

\title{Artemis: Efficient Commit-and-Prove SNARKs}

\ifdefined\isnotanon
\author{
\IEEEauthorblockN{Hidde Lycklama\textsuperscript{1*}, Alexander Viand\textsuperscript{2*}, Nikolay Avramov\textsuperscript{1}, Nicolas K\"uchler\textsuperscript{3}, Anwar Hithnawi\textsuperscript{1}}
\IEEEauthorblockA{\textsuperscript{1}\textit{University of Toronto} \quad \textsuperscript{2}\textit{University of Cambridge} \quad \textsuperscript{3}\textit{ETH Zurich}}
\thanks{* These authors contributed equally to this work.}
}
\else
\author{\textit{Paper \#1581, 13 pages + references \& appendix}}
\fi

\newcommand*{\affmark}[1][*]{\textsuperscript{#1}}

\date{}

\maketitle

\begin{abstract}

Zero-knowledge proofs (ZKPs) are powerful cryptographic tools for ensuring integrity while preserving privacy in open, potentially adversarial environments.
A recurring pattern in systems that employ ZKPs is the need to prove statements about values that were fixed earlier via a cryptographic commitment, for example, to prove possession of an issued credential, audit rollup execution against committed state, or bind an inference proof to a pre-committed machine learning model.
In many systems, these statements are supported by simply
recomputing the commitment inside the proof, resulting in substantial overhead.
Recent \gls{cpsnark} constructions attempt to reduce this cost using linking proofs that establish consistency between external and internal proof commitments.
In practice, however, these approaches often suffer from limited interoperability, flexibility, and performance.

We present \oursystem, a new \gls{cpsnark} construction that
introduces an efficient argument that directly links witness values inside the \acrshort{snark} to values bound by an external commitment, avoiding rigid commitment-level linking proofs.
As a result, \oursystem supports a broad class of proof systems, including plausibly post-quantum constructions, and can be instantiated without pairings or a trusted setup.
We implement and evaluate \oursystem, showing that it is efficient
across a wide range of proof systems and commitment
schemes.

\end{abstract}
\glsresetall

\subsecspacingtop
\section{Introduction}\lsec{intro}
\subsecspacingbot

\Glspl{zkp} are powerful cryptographic tools for enforcing privacy-preserving integrity in
open and adversarial systems.
They enable mutually distrustful parties to verify correctness properties of data, computation, and protocol execution without revealing sensitive inputs, internal state, or auxiliary information.
Advances over the past decade have brought \glspl{zkp} into the realm of practice, and today \glspl{zkp} are deployed across a broad range of real-world applications, including digital credential systems such as Google Wallet~\cite{googlewalletid}, voting~\cite{Benaloh2024-electionguard}, decentralized consensus and finance~\cite{zksync-rollup,ZcashSpec2016}, and privacy-preserving authentication~\cite{zkpassport2026}.

In many of these deployments, the goal is not only to show that a statement holds, but also to prove properties about values that were fixed, recorded, or committed prior to proof generation.
This commit-and-prove pattern appears in a wide range of applications, including anonymous credential schemes~\cite{Rosenberg2023-jt}, e-voting protocols~\cite{Adida2008-Helios}, verifiable encryption~\cite{Lee2019-sr}, decentralized auditing~\cite{Waiwitlikhit2024-ys}, zk-rollups~\cite{zksync-rollup,scroll2024whitepaper,plonky2,plonky3}, and zkML pipelines~\cite{Ganescu24-EZKL,EZKL_Github,Kang2022-zkml,Chen2024-ng}.
It also underlies accumulator-based constructions such as anonymous payments~\cite{ZcashSpec2016}, financial compliance mechanisms (e.g., AML), and anonymous broadcast protocols~\cite{Benarroch21-CPSet}.
For some proof systems, such as Sigma protocols and Bulletproofs~\cite{Bunz2018-mg}, expressing statements about values contained in commitments is an inherent part of the protocol.
In contrast, the most widely deployed modern (zk-)\glspl{snark}, such as Groth16~\cite{Groth16}, Halo2~\cite{halo2book}, and many STARKs, are designed to prove algebraic relations over constraint polynomials involving the witness.
Consequently, the commitments appearing in these proofs are tied to the proof system's internal representation of the computation.
This prevents them from serving as (reusable) commitments defined prior to proof generation.
As a result, statements about committed values must be supported by an explicit construction.

\fakeparagraph{Commitment Consistency in Practice}
A common approach used by many legacy \gls{snark} systems is that the prover recomputes the commitment within the \gls{snark} and constrains it to match the public commitment.
This strategy is simple and general, but expensive, since commitment schemes typically involve many operations that translate into large numbers of constraints.
This has led to the adoption of field-friendly commitment schemes and hash functions that are specifically designed to be cheaper to recompute inside a \gls{snark}, such as Sinsemilla~\cite{ZcashSpec2016} and Poseidon~\cite{Grassi2021-poseidon}.
However, despite their zk-friendly design, as committed data grows in modern applications, recomputation-based approaches grow increasingly inefficient and ultimately overwhelmingly dominate prover costs.
In addition, hash-based commitments offer limited algebraic structure, which restricts opportunities for optimization, batching, and amortization, as seen in systems such as Zcash~\cite{ZcashSpec2016} and Filecoin~\cite{FilecoinSpec2017}.
Beyond performance, many zk-friendly hash functions are relatively recent and have undergone limited cryptanalytic scrutiny, leaving greater uncertainty about their security margins~\cite{PoseidonInitiative,Grassi2025-PoseidonSecurity}.

To avoid this
cost, recent systems replace recomputation with ad hoc constructions
that tie the external commitment directly to the
proof~\cite{Rinberg2026-flashprove,Datta2025-veritas,ezkl_commit_blog,Frolov2026-zkcinema,Frolov2026-icefish}.
However, ensuring that these constructions preserve both soundness and zero-knowledge remains challenging.
For example, Anthropic's recent FlashProve system for zkML requires periodic refreshes of the model commitment to prevent information leakage~\cite{Rinberg2026-flashprove}.
Other systems take a less formal approach to defining their construction, which can obscure soundness~\cite{Frolov2026-icefish} and zero-knowledge issues~\cite{Datta2025-veritas,ezkl_commit_blog,Frolov2026-zkcinema,Frolov2026-icefish} arising from subtle gaps in their formalization.

\begin{figure}
\centering

\begin{subfigure}[t]{0.3\columnwidth}
\centering
\begin{tikzpicture}[
  >=latex,
  every node/.style={inner sep=1pt},
  msg/.style={midway, fill=white, inner sep=3pt},
  title/.style={font=\bfseries\large},
]
  \node (P) at (0,2.0) {$\mathcal{P}(x,w)$};
  \node (V) at (0,0) {$\mathcal{V}(x)$};

  \draw[->] (P) -- node[msg, right] {$\pi$} (V);

\end{tikzpicture}
\caption{zk-SNARK}
\end{subfigure}
\hfill
\begin{subfigure}[t]{0.67\columnwidth}
\centering

\begin{tikzpicture}[
  >=latex,
  every node/.style={inner sep=3pt},
  msg/.style={midway, fill=white, inner sep=3pt},
  title/.style={font=\bfseries\large},
]
  \node (P) at (0,2.0) {$\mathcal{P}\bigl((x,c),(w,r)\bigr)$};
  \node (V) at (0,0) {$\mathcal{V}(x,c)$};

  \node[left=10mm of P, align=center] (Com) {$\textsf{Com}(w,r)$};

  \draw[->] (P) -- node[msg, right] {$\pi$} (V);
  \draw[->, dashed] (Com) -- node[msg, above] {$c,r$} (P);

  \coordinate (start)  at ($(P.south)+(-30mm,0)$);
  \coordinate (bend)   at ($(V.west)+(-8mm,0)$);
  \coordinate (corner) at (start |- bend);

    \draw[dashed, ->]
  (start) |- node[msg, left,pos=0.21] {$c$} (V.west);

\end{tikzpicture}

\caption{CP-SNARK with one (external) commitment $c$ and decommitment $r$}
\end{subfigure}

\caption{zk-SNARK and CP-SNARK for an instance $\sInstance$ and corresponding witness $\sWitness$.}
\label{fig:system}
\end{figure}

\fakeparagraph{\acrlong{cpsnark} Constructions}
At the same time, there is a line of work taking a more formal approach, defining this consistency verification pattern as a \gls{cpsnark}.
A number of \gls{cpsnark} compilers,
i.e., transformations that extend a SNARK to efficiently support commit-and-prove functionality,
have been proposed for modern \glspl{snark}~\cite{Campanelli2019-as,Campanelli2021-yd,Aranha2021-ts}.
These compilers avoid recomputation altogether by
using the commitment structure present in modern proof systems.
Although these constructions are elegant, they are rarely adopted, because applying them in practice and at scale presents several challenges.
Existing \gls{cpsnark} compilers rely on strong assumptions about how witness values are internally represented and committed to by the proof system, e.g., that the committed values appear contiguously
and in order within a single internal commitment.
Modern SNARK implementations routinely violate these assumptions, as they aggressively reorder, interleave, and encode witness values to optimize constraint efficiency and prover performance.
Consequently, compilers must recover the relevant witness values at the commitment level, either by modifying the proof system itself or by repeatedly invoking linking proofs, introducing significant overhead.
In addition, these linking proofs frequently require additional security assumptions, such as a trusted setup, and are limited in the types of commitments they can support, only allowing links between commitments of the same type~\cite{Campanelli2019-as,Campanelli2021-yd,Aranha2021-ts}.
At a fundamental level, these limitations arise because existing compilers enforce consistency between commitments, strongly tying
them to the commitment structure of the underlying proof system.
As a result, substantial performance remains untapped, both for accelerating existing applications and for enabling commitment-heavy workloads that are currently impractical.
Our goal in this work is to bridge this gap by providing a construction that retains the generality of recomputation-based approaches while achieving performance suitable for real-world deployments.

\fakeparagraph{Contributions}
We present \oursystem, a \gls{cpsnark} construction that makes commitment consistency practical at scale.
The core contribution of \oursystem is a new construction for enforcing commitment consistency that avoids linking internal and external commitments.
Instead, \oursystem introduces an \emph{interactive} algebraic equality argument that directly relates witness values inside the SNARK to values bound by external commitments.
This design removes dependence on the proof system's internal commitment structure and witness layout, yielding a layout-agnostic interface that supports a broad class of proof systems and commitment schemes.

We formalize Artemis as a meta-compiler for any black-box proof system that supports arithmetic over verifier challenges and sound revelation of witness values to the verifier.
We capture this class of proof systems with the notion of \emph{interactive arithmetizations} and show that it covers both R1CS and Plonkish constraint systems.
The external commitment can be any homomorphic polynomial commitment over the field of the proof system, such as KZG or Pedersen vector commitments, regardless
of the scheme used inside the proof.
As a result, Artemis is, to the best of our knowledge, the first \gls{cpsnark} compiler that is independent of the proof system’s internal commitments, and that can be instantiated without requiring pairings or a trusted setup.

Our construction requires performing polynomial evaluations within the arithmetization, which can already be efficiently realized using standard addition and multiplication gates.
We further introduce an optimized construction for Plonkish arithmetizations based on a custom gate that implements Horner’s method, improving concrete efficiency.
We evaluate Artemis on two commitment-consistency-heavy workloads:
a zkML pipeline across models ranging from MNIST to \mbox{GPT-2}~\cite{Radford2019-gpt2},
and the TLSNotary system across a range of different-sized websites.
Our results demonstrate that \oursystem is efficient
across a wide range of proof systems and commitment
schemes, including pairing-based and transparent proof systems.

\subsecspacingtop
\section{Background}
\subsecspacingbot
\lsec{background}

We review the main cryptographic notions used throughout the paper.
We discuss interactive proofs, Commit-and-Prove SNARKs, and polynomial commitments.
\ifdefined\isnotextended
Due to space constraints, we defer additional definitions to~\Cref{apx:definitions}.
\else
We defer additional definitions to~\Cref{apx:definitions}.
\fi

\fakeparagraph{Notation}
We use the standard notation for bitstrings $\{ 0,1\}^*$,
groups (\sgenerator generates \sgroup) and fields $\sfield_p$ with order $p$.
We use bracket notation to denote ranges, e.g., $[n] = \{1, \ldots, n\}$, and polynomials are denoted by \textbf{bold} symbols.

\begin{definition}[Witness Relation]
A witness relation $\sRelation$ is a set of pairs $(\sInstance, \sWitness)$, where
$\sInstance$ is an instance, and $\sWitness$ is a witness.
Given a security parameter $\lambda \in \mathbb{N}$, we write $\sRelation_\lambda$ for the restriction of $\sRelation$ to those pairs with $|\sInstance| \leq \lambda$.
\end{definition}
Unless stated otherwise, the asymptotic security notions below are quantified over $\secpar$-compatible relations $\sRelation_\secpar$, and we usually write $\sRelation$ for simplicity.

\subsection{Proofs}
An interactive proof for a witness relation $\sRelation$ is a protocol between a prover $\sprover$ and a verifier $\sverifier$ by which $\sprover$ convinces $\sverifier$ that it knows a witness $\sWitness$ with $(\sInstance, \sWitness) \in \sRelation$ for a public instance $\sInstance$.
We write $\langle \sprover(\ipk, \sInstance, \sWitness), \sverifier(\ivk, \sInstance) \rangle$ for an execution of the protocol and $\tau$ for its transcript, and we say that $\sverifier$ accepts $\tau$ if it outputs $1$.
An interactive proof is an information-theoretic object: it is complete if the
honest prover always convinces $\sverifier$, and sound if any prover can convince $\sverifier$ only with negligible probability when no witness exists for an instance.

\fakeparagraph{Proof Systems}
A proof system turns an interactive proof into a succinct argument whose
soundness holds against polynomial-time provers.
A proof system $\zkSARK$ for a relation class $\sRelationClass$ consists of:
\begin{algos}
    \item $\zkSetup(\sRelation \in \sRelationClass) \rightarrow (\ipk,\ivk)$: on input a relation \sRelation, outputs a proving key and a
    verification key.
    \item $\langle \sprover(\ipk, \sInstance, \sWitness), \sverifier(\ivk,
    \sInstance) \rangle \rightarrow \tau$: a public-coin protocol to prove that $(\sInstance, \sWitness) \in \sRelation$, at the end of
    which $\sverifier$ accepts or rejects.
\end{algos}
\noindent
We say that $\zkSARK$ is a \gls{zksark} for $\sRelation$ if
the compiled argument satisfies completeness, knowledge-soundness, succinctness and zero-knowledge.
\ifdefined\isnotextended
Formal definitions of these properties are provided in~\cref*{def:zk-sarks} in~\cref*{apx:definitions}.
\else
Formal definitions of these properties are provided in~\cref{def:zk-sarks} in~\cref{apx:definitions}.
\fi
We can apply the Fiat-Shamir transform to convert a \gls{zksark} into a \gls{zksnark}.

\fakeparagraph{\acrlongpl{cpsnark}}
\glspl{cpsnark} are \glspl{snark} where the instance contains one or more commitments to parts of the witness~\cite{Campanelli2019-as,Campanelli2021-yd,Aranha2021-ts}.
Intuitively, a CP-SNARK proves both that the underlying relation holds and that the commitments in the instance are consistent with a subset of the witness used in the proof.
In particular, the instance contains a set of commitments, i.e., $(\sInstance, \scomList)$,
to subsets of the witness $\sWitness$ with decommitments $\srandList$ where $\scom_i = \sComCommit(\ck, \scommittedWitness_i, \sRandomness_i)$ with $\scommittedWitness_i$ a subset of the witness $\sWitness$.

\begin{definition}[\glspl{cpsnark}~\cite{Campanelli2019-as}]
\label{def:cpsnark}
Let $\sRelation$ be a relation over $\mathcal{D}_{\sInstance} \times \mathcal{D}_{\sWitness}$, where the witness domain decomposes as
$\mathcal{D}_{\sWitness} = \mathcal{D}_1 \times \ldots \times \mathcal{D}_\ell \times \mathcal{D}_{\sWitnessExtra}$
for some $\ell \geq 1$, represented by an index set $\sIndexSetCom$ that is part of the description of $\sRelation$.
Let $\scommittedWitness_1, \ldots, \scommittedWitness_\ell$ and $\sWitness_{\sWitnessExtra}$ denote the corresponding sub-witnesses following this decomposition.
Let $\text{Com} = (\sComSetup, \sComCommit, \sComVerify)$ be a commitment scheme (as per~\cref{def:commitments}) whose message space $\scomMessageSpace$ satisfies $\mathcal{D}_i \subseteq \scomMessageSpace$ for all $i \in [\ell]$.
A \acrlong{cpsnark} for relation $\sRelation$ and commitment scheme $\text{Com}$ is a \gls{snark} for the relation $\sRelationCP$, represented by the pair $(\ck, \sRelation)$ with $\ck \in \sComSetup(1^\secpar)$ and defined as:

\begin{equation*}
    \sRelationCP = \left\{
        \begin{array}{c}

          \left(
    (\sInstance, \scomList), (\sWitness, \srandList)
   \right)
    : \\
    \begin{array}{l}
    (\sInstance, \sWitness) \in \sRelation \\
        \bigwedge_{j \in [\ell]} \sComVerify(\ck, \scom_j, \scommittedWitness_j, \scomRandomness_j) = 1
    \end{array}
    \end{array}
    \right\}
\end{equation*}

\noindent
A \gls{cpsnark} satisfies completeness, knowledge soundness, and succinctness properties analogous to those of a SNARK.
A zk-\gls{cpsnark} additionally satisfies zero-knowledge.
We refer to Campanelli et al.~\cite{Campanelli2019-as} for formal definitions of these properties.
\glsunset{cpsark}
Similar to \glspl{snark}, \glspl{cpsnark} are typically constructed by applying the Fiat-Shamir transform to their interactive variant, the \gls{cpsark}.
\end{definition}

\subsection{Polynomial Commitments}
\lsec{pc}
A polynomial commitment scheme allows a prover to commit to a polynomial while retaining the ability to later prove evaluations of that polynomial at specific points.
Polynomial commitments are a fundamental building block for succinct proof systems.

\begin{definition}[Polynomial Commitments~\cite{Kate2010-px}]
\ldef{pc}
A polynomial commitment scheme consists of a quintuple
$(\PCSetup, \PCCommit, \PCVerify, \PCProve, \PCCheck)$, defined as follows:
\begin{algos}
    \item $\PCSetup(1^\lambda, D) \rightarrow \ck$: on input a security parameter $\lambda$ and a maximum degree bound $D$, outputs a commitment key $\ck$, which includes a description of a finite field $\sfield \in \mathcal{F}$.
    \item $\PCCommit(\ck, \spolynomial, \pcDbound, \sRandomness) \rightarrow \scom$: commits to a polynomial $\spolynomial$ of degree at most $\pcDbound \le D$, using randomness $\sRandomness$.
    \item $\PCVerify(\ck, \scom, \pcDbound, \spolynomial, \sRandomness) \rightarrow \{0,1\}$: checks whether $\scom$ is a valid commitment to $\spolynomial$ under randomness $\sRandomness$ and degree bound $\pcDbound$.
    \item $\PCProve(\ck, \spolynomial, \pcDbound, x, y, \sRandomness) \rightarrow \pi$: produces a proof $\pi$ that the committed polynomial satisfies $\spolynomial(x) = y$.
    \item $\PCCheck(\ck, \scom, \pcDbound, x, y, \pi) \rightarrow \{0,1\}$: verifies that $\pi$ is a valid proof that $\scom$ commits to a polynomial of degree at most $\pcDbound$ that evaluates to $y$ at point $x$.
\end{algos}

\noindent
A polynomial commitment scheme is secure if it satisfies completeness, binding, and hiding.
Some schemes additionally provide \emph{extractability}, meaning there exists an extractor that can recover the committed polynomial from any valid evaluation proof, given full access to the adversary's state.
We refer to~\cite{Chiesa2020-marlin} for formal definitions of these properties.
Some polynomial commitment schemes additionally support homomorphic operations on commitments (cf.~\Cref{def:homomorphic_commitment} in~\Cref{apx:definitions}), a property that our construction relies on to achieve zero-knowledge.

\end{definition}


\subsecspacingtop
\section{Related Work}
\subsecspacingbot
\lsec{relatedwork}

Many deployed ZK applications prove statements about data that has been previously committed.
This pattern appears in ZK Rollups, which generate validity proofs over batches of committed transactions~\cite{zksync-rollup,scroll2024whitepaper,plonky2,plonky3}, in \gls{zkml} pipelines that bind proofs to committed model parameters or training data~\cite{Feng2021-cl,Kang2022-zkml,Waiwitlikhit2024-ys,ezkl_commit_blog}, and in accumulator based systems used for anonymous payments and storage proofs~\cite{ZcashSpec2016,FilecoinSpec2017}.
In these settings, the most common approach is based on effectively “(re-)computing” the commitment inside the \gls{snark}.
However, commitments and SNARKs generally rely on different algebraic structures; therefore, one needs to emulate operations, such as elliptic curve computations, using a large number of arithmetic circuit operations.
To reduce this gap, several systems adopt ZK-friendly primitives such as Zcash's Jubjub curve and the Sinsemilla hash~\cite{halo2book,ZcashSpec2016},
often in combination with algebraic hash functions like Poseidon~\cite{Grassi2021-poseidon}.
These primitives are substantially cheaper to encode than conventional hash functions and commitments.
Nevertheless, for large committed inputs, recomputation remains expensive, motivating approaches that avoid embedding the full commitment computation inside the circuit.

Several \gls{snark} applications propose ad hoc constructions, but lack the formal treatment needed to prove soundness and fail to achieve zero-knowledge.
For example, the construction used in VerITAS~\cite{Datta2025-veritas} directly includes the externally committed polynomial in the \gls{snark}'s permutation argument.
This fails to achieve the claimed zero-knowledge property, as each proof reveals evaluation points from the external commitment through the permutation argument.
This is the same issue identified by Anthropic's FlashProve~\cite{Rinberg2026-flashprove}, which addresses it by regularly re-generating the commitment.
However, this solution does not transfer well to the image provenance setting.
In concurrent work, zk-Cinema~\cite{Frolov2026-zkcinema} uses a construction
that tries to use a similar technique based on polynomial evaluation inside and outside the \gls{snark}, but also fails to achieve zero-knowledge.
Although such leakage may be tolerable in some applications, reusing these constructions in other settings can have more serious consequences.
For instance, Icefish~\cite{Frolov2026-icefish} adopts zk-Cinema's construction in a setting where a biobank publishes a dataset commitment containing human genomic data for use across many studies.
The information leaked across sufficiently many studies can therefore accumulate to permit complete reconstruction of the underlying dataset which is highly sensitive.
Together with any homomorphic polynomial commitment scheme, our meta-compiler transforms the proof systems used in these works into a sound \gls{cpsnark} that achieves zero-knowledge.

\fakeparagraph{Commit-and-Prove SNARKs}
The notion of (zero-knowledge) ``Commit-and-Prove SNARKs'' (CP-SNARKs) has been formalized by multiple works.
Early CP-SNARKs~\cite{Groth2010-short,Costello2015-Gepetto,Lipmaa2016-cpsnark,Veeningen2017-adaptive} rely on commitments that depend on the proof's relation or reference string, which makes their commitments incompatible across proof systems~\cite{Campanelli2019-as}.
We adopt the formalization of Campanelli et al.~\cite{Campanelli2019-as} who propose a modular approach to constructing them in LegoSNARK by adapting the Groth16~\cite{Groth16} \acrshort{zksnark} to a (zk-)\gls{cpsnark} through a linking proof.
Subsequent work proposes CP-SNARK constructions for a variety of proof systems.
Chen et al.~\cite{Chen2023-hyperplonk} show how to convert sumcheck-based SNARKs into CP-SNARKs, although their construction does not yield a zk-CP-SNARK.
ECLIPSE~\cite{Aranha2021-ts} introduces a compiler that transforms IOP-style (zk-)SNARKs instantiated with Pedersen-style commitments into (zk-)CP-SNARKs relying on amortized Sigma protocols.
This transformation results in a proof size sublinear in the
number of commitments and size of the committed witness.
However, the verifier’s computational overhead is linear in the committed witness size, which significantly impacts verifier efficiency when a large fraction of the witness is committed.

Most relevant to our work are the linking proofs presented in Lunar~\cite{Campanelli2021-yd,Campanelli2021-asiacrypt} and LegoSNARK~\cite{Campanelli2019-as}, whose construction is efficient enough for our setting involving large commitments.
Their linking proofs offer proof-size overhead independent of the committed witness size, but restrict the construction to pairing-based commitments, which, in practice, require a trusted setup for all currently known pairing-based schemes~\cite{Kate2010-px}.
Jivanyan et al.~\cite{Jivanyan2024-zkfft} propose a Halo2 extension that supports vector commitments as inputs via a modified inner product argument.
Although this approach avoids pairings, its modified inner-product argument incurs prover overhead quadratic in the number of committed elements.
More broadly, these systems all use a similar approach to linking external and internal commitments, which tightly couples the linking mechanism to the proof system's internal commitment structure and inherently limits their applicability.


\subsecspacingtop
\section{Design}
\subsecspacingbot
\lsec{design}

Despite the existence of well-formalized efficient \gls{cpsnark} compilers for a range of modern proof systems, practical deployments frequently still rely on either inefficient in-circuit recomputation or introduce ad hoc constructions that deliver on performance but fail to achieve the full set of desired guarantees.
As discussed in \Cref{sec:relatedwork}, this is because existing \gls{cpsnark} compilers have serious limitations in practice due to enforcing consistency at the commitment level,
which strongly ties them to the internal commitment structure of the underlying proof system.
Our goal is to bridge this gap with a construction that offers the guarantees of existing compilers, yet remains efficient across a wide range of proof systems and commitment schemes, removing the need for practitioners to resort to ad hoc constructions.
Concretely, we seek a design that:
(i)~is \emph{layout-agnostic}: arithmetizations reorder, interleave, and
duplicate witness values, so the values bound by a single external commitment
are typically scattered across many internal commitments, and the cost of our construction should depend only on the size of the committed witness rather than its layout;
(ii)~admits a \emph{succinct verifier} whose additional work is sublinear in the size of
the committed witness;
(iii)~introduces \emph{no additional setup or hardness assumptions} beyond those
of the underlying proof system and commitment scheme, in particular no pairings
or trusted setup where neither is otherwise required;
(iv)~is \emph{interoperable}, i.e., the external commitment scheme is
independent of the proof system's internal commitments, including their encoding
(e.g., coefficient or evaluation form and evaluation domain), thereby also
supporting hash-based, plausibly post-quantum proof systems; and
(v)~is \emph{provably sound and zero-knowledge}, even when an external
commitment is reused across many proofs.
To maximize generality, we further treat the proof system as a black box,
relying only on capabilities that most modern SNARKs already provide.

\fakeparagraph{Our Approach}
Rather than linking commitments, we enforce consistency directly on witness
values, by reframing it as an algebraic statement that is cheap to prove given
random challenges from the verifier.
Most commitment schemes used in practice are polynomial commitment schemes
\footnote{The Pedersen commitments used in earlier \gls{cpsnark} constructions can be viewed as polynomial commitments without additional hardness assumptions, with \PCProve implemented via a standard inner-product argument~\cite{Bowe2019-ug}. We refer to \Cref{def:pedersen_vector} in \Cref{apx:definitions} for details.}.
These schemes allow the prover to establish that a committed polynomial evaluates to a claimed value at a given point, enabling an efficient polynomial equality check.
Specifically, we compare the evaluation of this \emph{outside} polynomial with that of an \emph{inside} polynomial defined by the proof's witness values.
By the DeMillo--Lipton--Schwartz--Zippel lemma (\Cref{lemma:schwartz_zippel}), two distinct polynomials of bounded degree agree at a uniformly random evaluation point only with small probability.
This check is sound, however, only if both polynomials are fixed before the prover learns the verifier's evaluation point.
Otherwise, the prover could choose an alternative witness that shifts the inside polynomial to match the outside polynomial at that point.
Therefore, the prover must commit to the witness values that define the inside polynomial before receiving the challenge.
Many proof systems heavily rely on verifier challenges to improve prover efficiency, such as lookup and permutation arguments~\cite{halo2book}, and support expressing constraints involving these challenges.
This capability was recently formalized for R1CS-based systems.
In the following, we extend this to a generic notion of interactive arithmetization that we then use in our generic construction.

\subsection{Interactive Arithmetization}
The standard computation model used in \glspl{snark}, i.e., an arithmetic circuit over a finite field, does not support computation over verifier challenges.
Such computations are naturally modeled as interactive proofs, which capture the requirement that the prover fix intermediate values before receiving the verifier challenges on which subsequent computations depend.
Recently, Ozdemir et al. formalized this approach for R1CS in the context of volatile memory checking proofs~\cite{Ozdemir2025-aip}.
We adopt their notion of an \emph{Algebraic} Interactive Proof (AIP), which we briefly restate in the following.
We then define the concept of an interactive arithmetization that generalizes their R1CS arithmetization and then show that it applies to many other \glspl{snark} by introducing the I-Plonkish \gls{aip} arithmetization.

Formally, an \gls{aip} is a public-coin interactive proof for some relation \sRelation in which the prover sends vectors of field elements, and the verifier's final test is an arithmetic circuit over \sfield.
Like interactive proofs, \glspl{aip} can be complete and sound for some witness relation \sRelation,
but \glspl{aip} do not need to be succinct or zero-knowledge.
As a result, it is straightforward to transform a sound \gls{aip} into a knowledge-sound \gls{aip} by sending the witness \sWitness in the first round.
\begin{definition}[\acrlong{aip}~\cite{Ozdemir2025-aip}]
\label{def:aip}
	An \gls{aip} is a $\nRounds$-round public-coin interactive proof between a prover \sprover and a verifier \sverifier in which all messages are vectors of elements of a field \sfield.
	At the end of the protocol, \sverifier performs a final test defined by an algebraic circuit \dCircuit over \sfield.
\end{definition}
\noindent
An \gls{aip} specifies an abstract protocol, whereas a proof system operates on a concrete encoding of that protocol.
An \emph{interactive arithmetization} defines a class of such encodings.

\begin{definition}[Interactive Arithmetization]
\label{def:arithmetization}
An interactive arithmetization is a set $\arithK$
    of encodings in a specified encoding format.
    Each instance $\aip \in \arithK$ encodes an \gls{aip} together with its honest prover.
\end{definition}
\noindent
Our definition is designed to generalize any arithmetization for which an equivalent \gls{aip} exists.
A proof system $\zkSARK$ for an interactive arithmetization $\arithK$ is a proof system whose setup takes an instance, i.e., $\zkSetup(\aip) \rightarrow (\ipk, \ivk)$ for $\aip \in \arithK$.
We say $\zkSARK$ is a \gls{zksark} for $\arithK$ if, for every $\aip \in \arithK$ that encodes a complete and knowledge-sound \gls{aip} for a relation $\sRelation$, $\zkSARK$ is a \gls{zksark} for $\sRelation$ (cf.~\Cref{def:zk-sarks}).
For example, Ozdemir et al.\ define I-R1CS, an interactive arithmetization based on R1CS, and compile its instances using Mirage+, a variant of Groth16 that supports verifier challenges~\cite{Ozdemir2025-aip}.
I-R1CS fits our definition of an interactive arithmetization.
We show that Plonkish arithmetization also fits this framework.
We write $\aip \in \arithK$ for an instance and, when no confusion arises, also use $\aip$ to denote the \gls{aip} it encodes.


\begin{figure}[t]
    \centering
    \footnotesize
\begin{tikzpicture}[
    x=1cm, y=1cm,
    msg/.style={-{Stealth[length=4pt]}, semithick},
    lbl/.style={inner sep=1.5pt, font=\footnotesize},
    gate/.style={ipgate, thick, rounded corners=1.5pt},
    copy/.style={ipcopy, thick, densely dashed},
    look/.style={iplook, very thick, densely dotted},
]
    \definecolor{ipgate}{RGB}{0,114,178}
    \definecolor{ipcopy}{RGB}{213,94,0}
    \definecolor{iplook}{RGB}{0,140,100}
    \definecolor{ipfixed}{gray}{0.82}
    \definecolor{ipinst}{gray}{0.93}

    \def\xl{0.5}
    \def\xr{3.5}

    \node[lbl, anchor=west] (P) at (0, 0) {$\mathcal{P}(\sInstance, \sWitness)$};
    \node[lbl, anchor=east] (V) at (4.0, 0) {$\mathcal{V}(\sInstance)$};
    \draw (P.south west) -- (P.south east);
    \draw (V.south west) -- (V.south east);

    \draw[msg] (\xl, -0.75) -- node[lbl, above] {$\pMsg_1 = A_1$} (\xr, -0.75);
    \draw[msg] (\xr, -1.25) -- node[lbl, above] {$\zkc_1$} (\xl, -1.25);
    \node[lbl] at (2.0, -1.55) {$\vdots$};
    \draw[msg] (\xl, -2.25) -- node[lbl, above] {$\pMsg_{\nRounds} = A_{\nRounds}$} (\xr, -2.25);

    \begin{scope}[shift={(4.8, -0.6)}]
        \def\cw{0.4}    %
        \def\ch{0.34}   %
        \def\gh{1.7}    %
        \foreach \r in {0,...,4} {
            \foreach \x in {0, 0.4}              { \draw[fill=ipfixed] (\x, {-\r*\ch}) rectangle ++(\cw, -\ch); }
            \foreach \x in {0.88, 1.28, 1.76, 2.56} { \draw (\x, {-\r*\ch}) rectangle ++(\cw, -\ch); }
            \draw[fill=ipinst] (3.04, {-\r*\ch}) rectangle ++(\cw, -\ch);
        }
        \node[lbl] at (2.36, -0.17) {$\cdots$};

        \node[lbl] at (0.4,  0.2) {$F$};
        \node[lbl] at (1.28, 0.2) {$A_1$};
        \node[lbl] at (1.96, 0.2) {$A_2$};
        \node[lbl] at (2.76, 0.2) {$A_{\nRounds}$};
        \node[lbl] at (3.24, 0.2) {$P$};
        \node[lbl] at (1.72, 0.6) {grid $Z$ with $P = \sInstance$};

        \draw[gate] (-0.05, -0.37) rectangle (3.49, -0.65);
        \draw[copy] (1.48, -0.85) -- (3.24, -1.19);
        \fill[ipcopy] (1.48, -0.85) circle (1.6pt) (3.24, -1.19) circle (1.6pt);
        \draw[look] (0.4, 0) rectangle (0.8, -\gh);
        \draw[look, -{Stealth[length=3.5pt]}] (1.82, -1.53) -- (0.85, -1.53);
        \fill[iplook] (1.96, -1.53) circle (1.6pt);
    \end{scope}

    \def\dy{-2.8}
    \node[lbl, anchor=west] (d0) at (0.9, \dy) {$\dCircuit_{\aip} :=$};
    \draw[gate] ($(d0.east)+(0.12, 0.13)$) rectangle ++(0.36, -0.26);
    \node[lbl, anchor=west] (d1) at ($(d0.east)+(0.55, 0)$) {gates $\,\land$};
    \draw[copy] ($(d1.east)+(0.16, 0)$) -- ++(0.34, 0);
    \fill[ipcopy] ($(d1.east)+(0.16, 0)$) circle (1.6pt) ++(0.34, 0) circle (1.6pt);
    \node[lbl, anchor=west] (d2) at ($(d1.east)+(0.6, 0)$) {copies $\,\land$};
    \draw[look, -{Stealth[length=3.5pt]}] ($(d2.east)+(0.5, 0)$) -- ++(-0.38, 0);
    \node[lbl, anchor=west] (d3) at ($(d2.east)+(0.55, 0)$) {lookups \ on $Z$};
\end{tikzpicture}
    \figcaptionvspace
    \caption{Our I-Plonkish \gls{aip} arithmetization. The decision circuit $\dCircuit_{\aip}$ checks the grid's gate constraints, copies and lookups directly by encoding them as field operations.}
    \label{fig:iplonkish}
\end{figure}

\fakeparagraph{Plonkish Arithmetization}
For I-R1CS, it is straightforward to construct the equivalent \gls{aip}, because we can have the decision circuit encode the verification of the rank-1 constraints specified by the matrices $A$, $B$, and $C$.
Plonkish proof systems, by contrast, are typically defined for relations expressed through polynomial identities~\cite{Chen2023-hyperplonk,halo2book}.
Although we could encode these identities and their associated checks directly in the \gls{aip}'s decision circuit, we instead separate the underlying computation, including its grid, gate constraints, copy constraints, and lookups, from the operations introduced by the proof system.
This separation significantly simplifies the exposition of our arithmetization by eliminating auxiliary polynomial encodings used by the proof system.

In our formulation, which we call I-Plonkish, each advice phase corresponds to one round of the \gls{aip}.
The prover sends vectors of field elements representing the advice columns for that phase, and the verifier returns the challenges available to subsequent phases.
The decision circuit then checks the gate constraints, copy constraints, and lookups using the advice, instance, and fixed values provided by the prover.
Gate constraints are defined as field polynomials, but copy constraints and lookups require no additional polynomial encoding.
Instead, the circuit checks copy constraints using equality tests based on Fermat's little theorem and verifies lookup membership as a disjunction of equality tests over the table rows.
\Cref{fig:iplonkish} shows the resulting \gls{aip}.

\begin{definition}[I-Plonkish]
	\label{def:iplonkish:concise}
	An \emph{I-Plonkish} instance is a tuple
	\begin{equation*}
		\begin{aligned}
			\aip = (\sfield;\ & \nRounds, n, n_f, n_p, n_a^{(1)}, \ldots, n_a^{(\nRounds)}, \nChal_{\zkc_1},\ldots,\nChal_{\zkc_\nRounds} \in \mathbb{N};\\
			& F \in \sfield^{n_f \times n}, B, P_\sigma, \sigma, T).
		\end{aligned}
	\end{equation*}
	It defines a \nRounds-round \gls{aip} over \sfield with instance $\sInstance = P \in \sfield^{n_p \times n}$, where in round $i \in [\nRounds]$, $\mathcal{P}$ sends the advice columns $A_i \in \sfield^{n_a^{(i)} \times n}$, computed from $(\sInstance, \sWitness, \zkc_1, \ldots, \zkc_{i-1})$, and receives uniformly random challenges $\zkc_i \in \sfield^{\nChal_{\zkc_i}}$.
	The last response must be empty: $\nChal_{\zkc_\nRounds} = 0$.
	$\mathcal{V}$'s test is defined on the challenges $\zkc = (\zkc_1, \ldots, \zkc_\nRounds)$ and the grid $Z = (F \mid A \mid P)$ of $n$ rows, where $A = (A_1 \mid \ldots \mid A_\nRounds)$.
	$\mathcal{V}$ accepts iff:
	\begin{itemize}[leftmargin=10pt,topsep=4pt]
		\setlength\itemsep{4pt}
		\item \textbf{(Gates)} $\polyGate(j, \zkc) = 0$ for all $\polyGate \in B$ and $j \in [n]$, where $B$ is a set of polynomials over \sfield in the cells of $Z$ at a finite set of row offsets and, optionally, in the challenges \zkc, of degree at most $d_B$ in the cells, and $\polyGate(j, \zkc)$ is the value of $\polyGate$ when the offsets are relative to row $j$, modulo $n$;
		\item \textbf{(Copies)} $Z[z] = Z[\sigma(z)]$ for all $z \in P_\sigma$, where $\sigma$ is a bijection on the set of cells $P_\sigma$;
		\item \textbf{(Lookups)} $Z[L, j] \in \{ Z[R, j^\prime] \}_{j^\prime \in [n]}$ for all $(L, R) \in T$ and $j \in [n]$, where $T$ is a set of pairs of equal-length tuples of columns of $Z$, $R$ contains only fixed and advice columns, and $Z[L, j]$ is the tuple of the cells of the columns $L$ in row $j$.
	\end{itemize}
\end{definition}
\noindent
Together with I-R1CS, these formulations cover the main arithmetization paradigms used in \gls{aip} constructions, which leads us to believe that this approach extends to any of the arithmetizations commonly used in practice, including AIR~\cite{airrap} and CCS~\cite{Setty2023-ccs}.

\subsection{Witness Consistency}
\lsec{design:oursystem}

We express the consistency check as an algebraic statement consisting of two polynomials, one inside the arithmetization over the committed witness $\sWitness_i$, and one outside committed-to by the commitment $\scom_i$.
We construct the polynomial by interpreting each committed witness $\sWitness_i$ as a
coefficient vector
\begin{equation*}
	\pol{w}_i(X) := \sum_{j=1}^{|\sWitness_i|} \sWitness_i^{(j)} \cdot X^{j-1},
\end{equation*}
where $\sWitness_i^{(j)}$ denotes the $j$-th element of the $i$-th committed witness.
We extend the proof's interactive arithmetization to compute an evaluation of $\pol{w}_i$ using a random challenge point $\sEvalPoint$ chosen by the verifier.
For the external polynomial, this is straightforward using $\scom_i$ and an opening proof \PCProve.
The verifier checks that both evaluations result in the same value, which implies that the witnesses in the \gls{snark} indeed agree with the committed values (with high probability) due to the DeMillo–Lipton–Schwartz–Zippel Lemma (\Cref{lemma:schwartz_zippel}).

\fakeparagraph{Zero Knowledge}
Evaluating the witness polynomial leaks a point $\pol{w}_i(\sEvalPoint)$ to the verifier, which can be reconstructed if too many evaluations are known.
To preserve zero knowledge, we apply a masking technique that hides the underlying witness values.
Specifically, we define $\sEvalValue = \spolySingle(\sEvalPoint) + \sMask$, where $\sMask \sample \sfield_p$ is a uniformly random masking value chosen by the prover.
However, without further changes, a malicious prover could easily cheat by sending a commitment $\scom_{\sMask'}$ to an arbitrary polynomial $\spolyMask'$, rather than $\spolyMask$.
For example, setting $\spolyMask' = \spolyMask + \spolySingle^\prime - \spolySingle$ would allow for an arbitrary witness $\spolySingle^\prime$ to be used inside the SNARK.
Therefore, we require an additional verifier challenge $\alpha$ and compute the masked polynomial as a linear combination
$$\pol{w}^{\mu}_i(X) := \spolyMask + \pol{w}_i(X) \cdot \alpha.$$
Since the distribution of $\sMask$ is uniformly random, opening this extended witness value does not compromise the zero-knowledge property of our construction.
We give a formal argument in the security proof in Appendix~\ref{apx:compiler}.

\fakeparagraph{Aggregating Multiple Commitments}
A key advantage of our approach ---particularly in comparison to prior \gls{cpsnark} compilers--- is the ability to aggregate all $\spolySingle$ and $\scom_{i}$, thereby reducing the number of commitment checks to a single polynomial evaluation.
To achieve this, we compute a linear combination with the additional challenge $\alpha$; specifically, we set:
\begin{equation*}
    \pol{w}(X) = \spolyMask + \sum_{i=1}^\ell \spolySingle(X) \cdot \alpha^{i}
\end{equation*}
where $\spolyMask$ is the degree-0 polynomial defined by $\sMask$.
The prover computes a single polynomial evaluation proof for the combined commitment.
The verifier then verifies the proof with respect to the commitment $\scom_{\sMask} + \textstyle\sum_i \scom_{i} \cdot \alpha^i$.
We show in our proof that the knowledge soundness error this introduces is negligible.
Note that this is distinct from the usual batch opening that some polynomial commitments support (e.g., employed in Plonk with KZG commitments~\cite{Gabizon2019-plonk}).

\fakeparagraph{\oursystem Meta-Compiler}
Existing \gls{cpsnark} constructions augment a proof system \zkSARK with explicit linking proofs that show consistency between a collection of internal proof-system commitments and the external commitments $\scom_1,\ldots,\scom_\ell$.
We express the consistency check as an algebraic statement, which lets us avoid many of the proof-system-specific dependencies between the commitments.
As a result, \oursystem treats the proof system \zkSARK as a black box, and is compatible with most modern \glspl{snark} and any external homomorphic polynomial commitment scheme.
We therefore present \oursystem as a meta-compiler that transforms a \gls{zksark} proof system \zkSARK for an interactive arithmetization $\arithK$ to a proof system for a \gls{cpsark} \zkSARKCP for $\arithK$.
We can then turn the \gls{cpsark} into a \gls{cpsnark} using the standard Fiat-Shamir transform.

In order for our meta-compiler to work as a generic transformation, we need to capture some properties of real-world proof systems formally. %
First, the proof system should be challenge-exposing (\Cref{def:chsnark}), i.e., its verifier sends the \gls{aip}'s verifier challenges, which the compiler needs to obtain the evaluation point for the opening proof.
This is trivially satisfied for all known \glspl{snark}.
Second, the proof system should be message-extractable (\Cref{def:message-ext}), i.e., an accepting proof implies knowledge of the \gls{aip}'s prover messages and not only of a witness, which the compiler needs to relate the revealed evaluation to the committed polynomial.
Most proof systems satisfy this property, because their knowledge-soundness proofs already extract the prover messages from the commitments in the proof and only then derive a witness from them (e.g., \extSARK in~\cite[App.~A.4]{Ozdemir2024-vh}).
Third, the proof system should support sound revelation (\Cref{def:vr-snark}), i.e., the prover can reveal a single value of an \gls{aip} prover message to the verifier without leaking the remaining values, which the verifier needs to compare the evaluation result.\footnote{
The value cannot be made public through the instance because it depends on a verifier challenge, which would make the statement challenge-dependent and introduce a circularity that complicates soundness.
}
Most proof systems provide this capability at the cost of one additional opening, because they already commit to the prover messages with a scheme that supports openings: \gls{piop}-based proof systems allow querying a committed polynomial, and Pedersen-based proof systems support opening with a $\Sigma$-protocol.


\begin{figure}[t]
    \centering
    \footnotesize
\begin{tikzpicture}[
    x=1cm, y=1cm,
    msg/.style={-{Stealth[length=4pt]}, semithick},
    lbl/.style={inner sep=1.5pt, font=\footnotesize},
    side/.style={inner sep=1pt, font=\footnotesize, anchor=west},
]
    \definecolor{aipnew}{RGB}{222,235,247}   %
    \definecolor{aipnewline}{RGB}{33,113,181}

    \def\xl{1.7}
    \def\xr{4.7}
    \def\xs{5.25}

    \fill[aipnew, rounded corners=2pt] (0, -1.8) rectangle (8.3, -4.2);
    \node[lbl, anchor=north west, text=aipnewline] at (0.05, -1.85) {\textbf{\oursystem}};

    \node[lbl, anchor=west] (P) at (0, 0) {$\mathcal{P}(\sInstance, \sWitness, \sMask)$};
    \node[lbl, anchor=east] (V) at (8.3, 0) {$\mathcal{V}(\sInstance)$};
    \draw (P.south west) -- (P.south east);
    \draw (V.south west) -- (V.south east);

    \draw[msg] (\xl, -0.7) -- node[lbl, above] {$\pMsg_i$} (\xr, -0.7);
    \draw[msg] (\xr, -1.2) -- node[lbl, above] {$\zkc_i$} (\xl, -1.2);
    \draw[rounded corners=3pt, -{Stealth[length=3pt]}]
        (\xl-0.1, -1.2) -- (\xl-0.4, -1.2) -- (\xl-0.4, -0.7) -- (\xl-0.1, -0.7);
    \node[lbl] at ({(\xl+\xr)/2}, -1.48) {$\vdots$};
    \node[side] at (\xs, -0.95) {rounds $1,\ldots,\nRounds$ of \aipObj};

    \draw[msg, aipnewline] (\xl, -2.35) -- node[lbl, above, text=black] {$\sMask$} (\xr, -2.35);
    \draw[msg, aipnewline] (\xr, -2.85) -- node[lbl, above, text=black] {$\alpha, \sEvalPoint$} (\xl, -2.85);
    \node[side] at (\xs, -2.6) {round $\nRounds+1$};

    \draw[msg, aipnewline] (\xl, -3.35) -- node[lbl, above, text=black] {$\sEvalValue$} (\xr, -3.35);
    \node[side] at (\xs, -3.35) {round $\nRounds+2$};

    \node[lbl, anchor=east] at (8.2, -3.88)
        {$\dCircuit^\prime := \dCircuit \,\land\, \bigl[\sEvalValue = \sMask + \textstyle\sum_{i=1}^{\ell} \alpha^i \cdot \spolySingle(\sEvalPoint)\bigr]$};
\end{tikzpicture}
    \figcaptionvspace
    \caption{The \oursystem \gls{aip} extension $\aipObj^\prime$ of an \gls{aip} \aipObj. Rounds $1,\ldots,\nRounds$ are those of \aipObj. \oursystem appends two rounds (shaded): the prover sends the mask \sMask and receives the challenges $(\alpha, \sEvalPoint)$, then sends the masked evaluation \sEvalValue. The extended decision circuit $\dCircuit^\prime$ additionally checks \sEvalValue against the committed sub-witnesses $\scommittedWitness_1,\ldots,\scommittedWitness_\ell$, where \spolySingle is the polynomial with coefficient vector $\scommittedWitness_i$.}
    \label{fig:aip:ext}
\end{figure}

\fakeparagraph{\ourcompiler Protocol}
We construct \ourcompiler (cf.~\Cref{fig:compiler:artemis:meta}), a meta-compiler that transforms a \gls{zksark} into a zk-\gls{cpsark}, which guarantees that each instance commitment $\scom_i$ opens to a subset $\scommittedWitness_i$ of $\sWitness$ with decommitment~$\sRandomness_i$ (cf.~\Cref{def:cpsnark}).
The meta-compiler takes as input a \gls{zksark} compiler \zkSARK for an interactive arithmetization \arithK that is challenge-exposing (\Cref{def:chsnark}), message-extractable (\Cref{def:message-ext}), and soundly-revealing (\Cref{def:vr-snark}), and outputs a zk-\gls{cpsark} compiler for \arithK, where \PCScheme is the homomorphic polynomial commitment scheme of the external commitments.
\ourcompiler constructs $\aipObj^\prime$ by extending the \gls{aip} corresponding to an instance $\aip \in \arithK$ to compute a polynomial evaluation, as \Cref{fig:aip:ext} illustrates.
After \aip's last round, the prover sends the mask \sMask, the verifier answers with the challenges $(\alpha, \sEvalPoint)$, and the prover sends the evaluation result \sEvalValue.
Sending \sMask as an \gls{aip} message is permissible because zero-knowledge is provided by the compiler rather than required of the underlying \gls{aip}.
The extended decision circuit accepts only if $\sEvalValue = \sMask + \sum_{i=1}^{\ell} \alpha^i \cdot \spolySingle(\sEvalPoint)$, where \spolySingle is the polynomial whose coefficients are the $i$-th committed part of the witness.
We then call setup on \aipExt, the arithmetization of the extended \gls{aip} $\aipObj^\prime$.

In the protocol,
the prover first commits to the constant polynomial \spolyMask under \PCScheme and sends $\scom_\sMask$ to the verifier.
Prover and verifier then execute \zkSARK on \aipExt, in which the verifier sends $\alpha$ and \sEvalPoint.
Afterwards, the prover sends \sEvalValue and uses the sound-revelation property of \zkSARK to show that \sEvalValue is the evaluation result in the transcript of \zkSARK.
The protocol then proceeds to show that the external commitments $\scom_1,\ldots,\scom_\ell$ evaluate to the same value $\sEvalValue$, using the same challenges $\alpha$ and $\sEvalPoint$.
Towards this, the prover computes the opening proof using the masked linear combination of the witness polynomials and the commitment randomnesses,
$\piCommitment \gets \PCProve(\ck, \spolyMask + \sum_{i=1}^{\ell} \spolySingle \cdot \alpha^i, \pcDbound, \sEvalPoint, \sEvalValue, \sRandomness_\sMask + \sum_{i=1}^{\ell} \sRandomness_i \cdot \alpha^i)$.
The verifier computes the masked linear combination of the commitments $\scom \gets \commitmentCombination$ which is possible due to their homomorphic nature, and uses it to verify the opening proof.
Due to the DeMillo–Lipton–Schwartz–Zippel Lemma, this check will complete (with high probability) only if the witnesses in the \gls{snark} indeed agree with the committed values.
We prove that \ourcompiler outputs a zk-\gls{cpsark} compiler.

\begin{restatable}[\artemisMain yields a (zk-)\gls{cpsark}]{theorem}{thmArtemisMain}
\label{thm:artemis:aip}
    Let \PCScheme be an additively homomorphic polynomial commitment scheme over \sfield with binding and extractability (\Cref{def:pc,def:homomorphic_commitment}),
    and let $\zkSARK$ be a \gls{zksark} compiler for an interactive arithmetization \arithK for \sRelation that is challenge-exposing (\Cref{def:chsnark}), message-extractable (\Cref{def:message-ext}), and soundly-revealing (\Cref{def:vr-snark}) for every instance of \arithK.
    Let $\aip \in \arithK$ encode a complete and knowledge-sound \gls{aip} for \sRelation, and
    let $\zkSARKCP = \artemisMainOut$ be the output of the meta-compiler described in~\Cref{fig:compiler:artemis:meta}.
    \zkSARKCP is a zk-\gls{cpsark} compiler for \sRelation and \PCScheme (\Cref{def:cpsnark}).
    If \PCScheme is hiding, then this protocol is honest-verifier zero-knowledge.
\end{restatable}

\noindent\emph{Proof intuition.}
For knowledge soundness, the extractor uses message extractability of \zkSARK to obtain the witness messages, and rewinds the prover on $\ell+1$ distinct $\alpha$ to extract openings of the combined commitments, from which it recovers an opening of every $\scom_i$ by solving a Vandermonde system.
The DeMillo–Lipton–Schwartz–Zippel Lemma then implies that these openings equal the committed sub-witnesses inside \zkSARK.
For zero-knowledge, a challenge is that the simulator does not have access to the trapdoor of the instance commitments, because the commitment key is part of the relation.
Instead, the simulator samples a uniform \sEvalValue, commits to the constant polynomial \sEvalValue as $\scom_\sEvalValue$, and uses the homomorphic property to equivocate the mask commitment as $\scom_\sMask := \scom_\sEvalValue - \sum_{i=1}^{\ell} \alpha^i \cdot \scom_i$.
As a result, the combined commitment the verifier computes equals $\scom_\sEvalValue$ and the simulator can compute an honest opening proof without the witness.
We provide the full proof in \Cref{apx:compiler:ext}.

\ifdefined\isnotextended
\else

\begin{figure}[t]
    \centering

\begin{tikzpicture}
    \def\graywidth{1} %
    \def\cols{1} %
    \def\rows{5} %
    \definecolor{restcolor}{gray}{0.7}
    \definecolor{constantcolor}{gray}{0.96}
    \definecolor{fixedcolor}{gray}{0.8}

    \foreach \x in {1,...,\cols} {
        \foreach \y in {0,...,\rows} {
            \draw (\graywidth+\x, -\y) rectangle ++(1, -1);
        }
    }

    \foreach \y in {0,...,\rows} {
        \draw[fill=constantcolor] (\graywidth, -\y) rectangle ++(1, -1);
    }
        \foreach \y in {0,...,\rows} {
        \draw[fill=fixedcolor] (\graywidth+\cols + 1, -\y) rectangle ++(1, -1);
    }

    \draw[fill=restcolor] (0, 0) rectangle ++(\graywidth, -\rows-1);

    \node at (\graywidth/2, -\rows/2 - 0.25) {$\sRelation$};
    \node at (\graywidth/2, -\rows/2 - 0.75) {$\cdots$};

    \node at (\graywidth + 1.5, 0.5) {$\sEvalValue$};
    
    \node at (\graywidth + 0.5, 0.5) {$w$};
    \node at (\graywidth + 2.5, 0.5) {$F_{n_f + 1}$};

    \foreach \y in {0,...,4} {
        \pgfmathsetmacro\r{\y}
        \node at (\graywidth + 0.5, -\y - 0.5) {$w^{({\pgfmathprintnumber{\r}})}$};
    }

     \foreach \y in {0,...,4} {
        \pgfmathsetmacro\r{\y}
        \node at (\graywidth + 1.5, -\y - 0.5) {$\rho_{\pgfmathprintnumber{\r}}$};
    }
    \node at (\graywidth + 1.5, -5.5) {$0$};

    \foreach \y in {0,...,\rows} {
        \node at ( - 0.5, -\y - 0.5) {$\pgfmathparse{int(\y+1)}\pgfmathresult$};
    }

    \foreach \y in {0,...,4} {
        \node at ( \graywidth + 2.5, -\y - 0.5) {$1$};
    }
    \node at (\graywidth + 2.5, -5.5) {$0$};

\draw[red, line width=5pt]
      (\graywidth, -1) -- (\graywidth+\cols+5.5, -1)
      -- (\graywidth+\cols+5.5, -2)
      -- (\graywidth+2, -2)
      -- (\graywidth+2, -3)
      -- (\graywidth+1, -3)
      -- (\graywidth+1, -2)
      -- (\graywidth, -2)
      -- cycle;
    
    \node at (\graywidth + 4.75, -1.5) {$\polyHornerGate(X, \ldots, \Calpha, \Cbeta)$};

    \draw [decorate,decoration={brace,amplitude=6pt,mirror,raise=4pt},yshift=0pt]
        (\graywidth + 0.1, -6) -- (\graywidth+1.9, -6) node [black,midway,yshift=-0.6cm] {\footnotesize Advice};
    
    \draw [decorate,decoration={brace,amplitude=6pt,mirror,raise=4pt},yshift=0pt]
        (\graywidth+2.1, -6) -- (\graywidth+2.9, -6) node [black,midway,yshift=-0.6cm] {\footnotesize Fixed};

\end{tikzpicture}
    \figcaptionvspace
    \caption{Simplified visualization of a Plonkish grid with our extensions for a single commitment.}
    \label{fig:grid:simple}
\end{figure}

\fi


\begin{figure*}
\centering
\begin{tcolorbox}[]
\artemisMain is parameterized by a homomorphic polynomial commitment scheme $\PCScheme$ (\Cref{def:pc,def:homomorphic_commitment}).

\protocolparagraph{Input}
A \gls{zksark} compiler $\zkSARK = (\zkSetup, \langle \zkProver, \zkVerifier \rangle)$ for an interactive arithmetization \arithK with reveal rounds \revealRounds (\Cref{def:arithmetization,def:chsnark,def:message-ext,def:vr-snark}).

\protocolparagraph{Output}
The commit-and-prove compiler $\zkSARKCP = (\zkSetup, \artemisLink)$ for \arithK, defined as follows.

\begin{list}{}{\setlength{\leftmargin}{1.2em}\setlength{\topsep}{0pt}\setlength{\partopsep}{0pt}}\item[]
\protocolparagraph{\zkSetup\textnormal{$(\aip, \ck, \sIndexSetCom)$}}
On input an instance $\aip \in \arithK$, a commitment key \ck, and the index set \sIndexSetCom of the committed sub-witnesses:
\begin{enumerate}
    \item \label{step:setup:aip} Let \aipObj be the \gls{aip} defined by \aip, with \nRounds rounds, prover messages \pMsgs and decision circuit \dCircuit.
    \item \label{step:setup:ext} Extend \aipObj to $\aipObj^\prime$ by adding the following (cf.~\Cref{fig:aip:ext}):
    \begin{enumerate}
        \item in round $\nRounds+1$, the prover sends \sMask, and the verifier answers with $(\alpha, \sEvalPoint) \sample \sfield^2$;
        \item in round $\nRounds+2$, the prover sends \sEvalValue;
        \item the decision circuit is $\dCircuit^\prime := \dCircuit \land \bigl[ \sEvalValue = \sMask + \sum_{i=1}^{\ell} \alpha^i \cdot \spolySingle(\sEvalPoint) \bigr]$, where \spolySingle is the polynomial with coefficient vector $\scommittedWitness_i$ for $(\scommittedWitness_1,\ldots,\scommittedWitness_\ell) := \msgProjI(\pMsgs)$ where \msgProj is the message projection map of $\phi$.
    \end{enumerate}
    \item \label{step:setup:arith} Let \aipExt be the \arithK-arithmetization of $\aipObj^\prime$.
    Let $(\hRound, \hIdx) := (\nRounds+2, 1)$, the round and the coordinate of \sEvalValue in \aipExt.
    \item \label{step:setup:compile} Call $(\ipk, \ivk) \leftarrow \zkSARK.\zkSetup(\aipExt)$.
    \item Output $\textsf{pk} := (\ipk, \ck, \sIndexSetCom, \hRound, \hIdx)$ and $\textsf{vk} := (\ivk, \ck, \sIndexSetCom, \hRound, \hIdx)$.
\end{enumerate}

\protocolparagraph{Protocol \textnormal{$\langle \mathcal{P}^{[\PCScheme]}(\textsf{pk}, (\sInstance, \scomList), (\sWitness, \srandList)), \mathcal{V}^{[\PCScheme]}(\textsf{vk}, (\sInstance, \scomList)) \rangle$}}
\begin{enumerate}
    \item \label{step:mask} $\mathcal{P}$ samples
    $\sProtocolRandomness, \sRandomness_\sProtocolRandomness \sample
    \sfield$ and sends $\scom_\sProtocolRandomness \leftarrow
    \PCCommit(\ck, \spolyMask, \pcDbound, \sRandomness_\sProtocolRandomness)$ to
    $\mathcal{V}$, where \spolyMask is the constant polynomial \sMask.
    \item \label{step:zkp} $\mathcal{P}$ and $\mathcal{V}$ run the protocol
    $
    \langle \zkProver(\ipk, \sInstance, (\sWitness, \sProtocolRandomness)),\;
        \zkVerifier(\ivk, \sInstance) \rangle
    $.
    The resulting transcript $\tau$ exposes the challenges $\alpha$ and \sEvalPoint (\Cref{def:chsnark}).
    \item \label{step:zkpreveal} $\mathcal{P}$ sends \sEvalValue to $\mathcal{V}$, and $\mathcal{P}$ and $\mathcal{V}$ run the reveal protocol of \zkSARK for coordinate \hIdx of round \hRound (\Cref{def:vr-snark}):
    $$
    \langle P^{(\hRound,\hIdx)}_\mathsf{Reveal}(\ipk, \sInstance, (\sWitness, \sMask), \tau, \sEvalValue),\
    V^{(\hRound,\hIdx)}_\mathsf{Reveal}(\ivk, \sInstance, \tau, \sEvalValue) \rangle.
    $$
    \item \label{step:link} Let
    $\mathbf{w}^* := \spolyMask +
        \sum_{i=1}^{\ell} \alpha^i \cdot \spolySingle$
    and
    $\sRandomness^* := \sRandomness_\sProtocolRandomness +
        \sum_{i=1}^{\ell} \alpha^i \cdot \sRandomness_i$,
    where $\spolySingle$ is the polynomial with coefficient vector $\scommittedWitness_i$, the $i$-th committed sub-witness of \sWitness following \sIndexSetCom.
    $\mathcal{P}$ sends
    $\piCommitment \leftarrow
    \PCProve(\ck, \mathbf{w}^*, \pcDbound, \sEvalPoint, \sEvalValue,
    \sRandomness^*)$ to $\mathcal{V}$.
    \item \label{step:accept} $\mathcal{V}$ computes $\scom \leftarrow \commitmentCombination$ and accepts if and only if $\zkVerifier$ accepts in Step~\ref{step:zkp}, $V^{(\hRound,\hIdx)}_\mathsf{Reveal}$ accepts in Step~\ref{step:zkpreveal}, and
    $\PCCheck(\ck, \scom, \pcDbound, \sEvalPoint, \sEvalValue, \piCommitment) = 1$.
\end{enumerate}
\end{list}
\end{tcolorbox}
\caption{The \oursystem meta-compiler \artemisMain.}
\label{fig:compiler:artemis:meta}
\end{figure*}


\begin{figure*}[t]
    \centering

\begin{tikzpicture}
    \def\graywidth{4} %
    \def\cols{1} %
    \def\rows{5} %
    \definecolor{restcolor}{gray}{0.7}
    \definecolor{constantcolor}{gray}{0.96}
    \definecolor{fixedcolor}{gray}{0.8}

    \foreach \y in {0,...,\rows} {
        \draw (\graywidth, -\y) rectangle ++(1, -1);
    }
    \draw (\graywidth+1, 0) rectangle ++(1, -\rows-1);
    \foreach \y in {0,...,\rows} {
        \draw (\graywidth+2, -\y) rectangle ++(1, -1);
    }

    \foreach \y in {0,...,\rows} {
        \draw (\graywidth+3, -\y) rectangle ++(1, -1);
    }

    \foreach \y in {0,...,\rows} {
        \draw[fill=fixedcolor] (\graywidth+\cols+3, -\y) rectangle ++(1, -1);
    }

    \draw[fill=restcolor] (0, 0) rectangle ++(\graywidth, -\rows-1);

    \node at (\graywidth/2, -\rows/2 - 0.25) {$\sRelation$};
    \node at (\graywidth/2, -\rows/2 - 0.75) {$\cdots$};

    \node at (\graywidth + 3.5, 0.5) {$\sEvalValue$};

    \node at (\graywidth + 0.5, 0.5) {$w_1$};
    \node at (\graywidth + 1.5, 0.5) {$\ldots$};
    \node at (\graywidth + 1.5, -\rows/2 - 0.5) {$\ldots$};
    \node at (\graywidth + 2.5, 0.5) {$w_{\numberCols}$};

    \node at (\graywidth + 4.5, 0.5) {$F_{n_f + 1}$};

    \foreach \y in {0,...,4} {
        \pgfmathsetmacro\r{2*\y}
        \node at (\graywidth + 0.5, -\y - 0.5) {$w_1^{({\pgfmathprintnumber{\r}})}$};
    }
    \foreach \y in {0,...,3} {
        \pgfmathsetmacro\r{2*(4 - \y)}
        \node at (\graywidth + 2.5, -\y - 0.5) {$w_{\ell}^{(d-{\pgfmathprintnumber{\r}})}$};
    }
    \node at (\graywidth + 2.5, - 4 - 0.5) {$w_{\ell}^{(d)}$};

     \foreach \y in {0,...,4} {
        \pgfmathsetmacro\r{\y}
        \node at (\graywidth + 3.5, -\y - 0.5) {$\rho_{\pgfmathprintnumber{\r}}$};
    }
    \node at (\graywidth + 3.5, -5.5) {$0$};

    \foreach \y in {0,...,\rows} {
        \node at ( - 0.5, -\y - 0.5) {$\pgfmathparse{int(\y+1)}\pgfmathresult$};
    }

    \foreach \y in {0,...,4} {
        \node at ( \graywidth + 4.5, -\y - 0.5) {$1$};
    }
    \node at ( \graywidth + 4.5, -5.5) {$0$};

    \draw[red, line width=5pt]
      (\graywidth, -1) -- (\graywidth+\cols+7.5, -1)
      -- (\graywidth+\cols+7.5, -2)
      -- (\graywidth+4, -2)
      -- (\graywidth+4, -3)
      -- (\graywidth+3, -3)
      -- (\graywidth+3, -2)
      -- (\graywidth, -2)
      -- cycle;
    \node at (\graywidth + 6.75, -1.5) {$\polyHornerGate(X, \ldots, \Calpha, \Cbeta)$};

    \draw [decorate,decoration={brace,amplitude=6pt,mirror,raise=4pt},yshift=0pt]
        (\graywidth + 0.1, -6) -- (\graywidth+3.9, -6) node [black,midway,yshift=-0.6cm] {\footnotesize Advice};

    \draw [decorate,decoration={brace,amplitude=6pt,mirror,raise=4pt},yshift=0pt]
        (\graywidth+4.1, -6) -- (\graywidth+4.9, -6) node [black,midway,yshift=-0.6cm] {\footnotesize Fixed};

\end{tikzpicture}

    \figcaptionvspace
    \caption{
        Visualization of a Plonkish grid with our extensions for $\ell$ commitments to $d+1$ values each, stored in $\numberCols$ advice columns $w_1, \ldots, w_{\numberCols}$.
    }
    \label{fig:grid}
\end{figure*}

\subsecspacingtop
\subsection{Efficiently Instantiating the Arithmetization}
\subsecspacingbot
\lsec{arithmetization}

Our meta-compiler is defined to be generic by augmenting the \gls{aip} with a polynomial evaluation over the relevant witness values.
While the arithmetizations we consider can express any \gls{aip}, including arbitrary extensions, this gives us no guarantees on whether this transformation is efficient.
In the following we show an optimized approach for I-Plonkish (cf.~\cref{def:iplonkish:concise}) used by many proof systems such as Halo2~\cite{halo2book}.
\ifdefined\isnotextended
In~\Cref{fig:grid}, we visualize the required additions to the Plonkish grid.
\else
In~\Cref{fig:grid:simple,fig:grid}, we visualize the required additions to the Plonkish grid.
\fi
Note that this is not to scale: in practice, grids will have many more rows, and the vast majority of the grid will be dedicated to the original relation \sRelation rather than our additions.

\fakeparagraph{Strawman Approach}
A straightforward approach to arithmetizing a polynomial evaluation would be to express it as the inner product of the witness polynomial and the powers of the evaluation point $1,\sEvalPoint,\ldots,\sEvalPoint^{d}$.
This requires adding at least three additional advice columns (one each for the witness values, powers of $\sEvalPoint$, and intermediate results).
More generally, if there are too many witness values to fit into a single column,
this approach requires $2\ceil{\frac{d+1}{n}}+1$ columns, where $d+1$ is the number of witness values and $n$ is the number of rows.
In settings where the committed witness values already occupy a substantial portion of the grid, the inclusion of the powers of $\sEvalPoint$ effectively doubles the overhead of the polynomial evaluation.
This remains noticeable even if the cost of the main proof is significant:
for example, for a proof of inference for the GPT-2 model, our method described below leads to a 22\% reduction in overall prover time.

\fakeparagraph{Horner's Method}
As the additional constraint that we need to add is essentially an evaluation of a polynomial at a specific point, we can utilize an arithmetization based on Horner's method~\cite{Horner1819}.
In order to illustrate this, we first consider a simplified setting, with a single commitment $\scom$ to witness polynomial $\mathbf{w}$ with coefficients \scommittedWitness (i.e., $\ell = 1$).
\ifdefined\isnotextended
For this simplified setting, we will also assume that the size $d+1$ of the witness is $n-1$, one less than the number of rows $n$ of the Plonkish grid, ensuring that the last row can hold the base case of the recurrence below.
\else
For this simplified setting, which we visualize in \Cref{fig:grid:simple}, we will also assume that the size $d+1$ of the witness is $n-1$, one less than the number of rows $n$ of the Plonkish grid, ensuring that the last row can hold the base case of the recurrence below.
\fi
We denote the individual elements $\scommittedWitness$ as $\scommittedWitness^{(0)}, \ldots, \scommittedWitness^{(d)}$.
Note that we specifically use zero-based indexing here as this is more natural when considering the elements as coefficients of $\pol{\scommittedWitness}$.
According to Horner's method, we can then compute
$$
\scommittedWitness^{(0)} + \scommittedWitness^{(1)}\sEvalPoint + \scommittedWitness^{(2)}\sEvalPoint^2 + \scommittedWitness^{(3)}\sEvalPoint^3 + \cdots
$$
as
$$
 \scommittedWitness^{(0)}
 + \sEvalPoint \bigg(\scommittedWitness^{(1)}
 + \sEvalPoint \Big(\scommittedWitness^{(2)}
 + \sEvalPoint \big(\scommittedWitness^{(3)}
 + \cdots
 \big) \Big) \bigg).
$$
This latter form enables a convenient recursive computation that, in order to compute the partial evaluation down to degree $j$, only requires access to the $j$-th coefficient, $\sEvalPoint$, and the \mbox{$j+1$-th} partial evaluation.
We denote the partial evaluation for the $j$-th degree as $\sEvalValue_j$.
Then, we have the recurrence relation
$$\sEvalValue_{j} = \scommittedWitness^{(j)} + \sEvalPoint \cdot \sEvalValue_{j+1}  \text{\ with \ } \sEvalValue_{d+1} = 0.$$
To express this in the Plonkish grid, we extend the grid with a set of additional columns:
a selector fixed column $n_f + 1$,
advice column \idxWit to store $\scommittedWitness^{(j)}$, and
advice column \idxRho to store $\sEvalValue_{j}$.
The prover sends column \idxRho in a later round than column \idxWit, so that it can compute $\sEvalValue_{j}$ from the challenge $\sEvalPoint$, which gates refer to through the challenge variable \Cbeta.
We add copy constraints to ensure that the copies of the witness values correspond to their original occurrences in the arithmetization of \sRelation.
Finally, we add a custom gate constraint:
\begin{equation*}
\begin{split}
    \polyHornerGate(r, \Cbeta) \coloneqq\ & F[n_f+1, r] \cdot \big( A[\idxWit, r] \\
    & + \Cbeta \cdot A[\idxRho, r+1] - A[\idxRho, r] \big),
\end{split}
\end{equation*}
where $Z[c, r]$ denotes the cell in column $c$ and row $r \in [n]$ of $Z \in \{F, A\}$ (cf.~\Cref{def:iplonkish:concise}), and row $r$ holds $\scommittedWitness^{(r-1)}$ and $\sEvalValue_{r-1}$.
The selector $F[n_f+1, r]$ is $1$ for $r \in [n-1]$ and $0$ in the last row, which holds $\sEvalValue_{d+1} = 0$, ensuring the gate never wraps around.

The custom gate spans two rows, with $A[\idxRho, r+1]$ referring to $\sEvalValue_{r}$ in the next row.
We provide a full description of the gate constraint and transformation of the arithmetization instance \hornerGate in Appendix~\ref{apx:compiler:inst}.

\fakeparagraph{Supporting Larger Commitments}
So far, we have assumed that the size $d+1$ of the commitment $\scommittedWitness$ is $n-1$, where $n$ is the number of rows in the Plonkish grid.
Where $d+1$ is smaller, we can trivially pad $\scommittedWitness$  with zeros.
However, if $d+1$ is larger than $n-1$,
we need to split $\scommittedWitness$ across multiple advice columns.
A straightforward approach might add a separate advice column for the intermediate value $\sEvalValue^\prime$ for each witness column, as well as multiple custom gates and selector columns.
However, we can avoid this overhead by combining Horner's method with a (generalized) even-odd decomposition approach.
Specifically, we use the common observation that
$$
\scommittedWitness^{(0)}
+ \scommittedWitness^{(1)}\sEvalPoint
+ \scommittedWitness^{(2)}\sEvalPoint^2
+ \scommittedWitness^{(3)}\sEvalPoint^3
+ \scommittedWitness^{(4)}\sEvalPoint^4
+ \scommittedWitness^{(5)}\sEvalPoint^5
+ \cdots
$$
can be rewritten as
\begin{equation*}
    \begin{split}
        &\scommittedWitness^{(0)}
        +  \scommittedWitness^{(2)}\sEvalPoint^{2}
        +  \scommittedWitness^{(4)}\sEvalPoint^{2^2}
        +  \cdots \\
        + \sEvalPoint  & \left( \scommittedWitness^{(1)}
        + \scommittedWitness^{(3)}\sEvalPoint^2
        + \scommittedWitness^{(5)}\sEvalPoint^{2^2}
        + \ldots  \right) \\
    \end{split}
\end{equation*}
which can be interpreted as a combination of two polynomials in $X^2$.
Combining this with Horner's method, we arrive at
 $$
 \left(\scommittedWitness^{(0)} + \sEvalPoint \scommittedWitness^{(1)} \right)
 +  \sEvalPoint^2 \Bigg( \left(\scommittedWitness^{(2)} + \sEvalPoint \scommittedWitness^{(3)} \right)
  +   \sEvalPoint^2 \Big( \cdots \Big) \bigg) \Bigg).
$$
which gives rise to the recurrence
$$\sEvalValue_{j} = \left( \scommittedWitness^{(2j)} + \sEvalPoint \cdot \scommittedWitness^{(2j+1)}   \right) + \sEvalPoint^2 \sEvalValue_{j+1} $$
for $j \in [0, n_\text{horner}-1]$, where $n_\text{horner} \coloneqq \ceil{\frac{d+1}{2}} < n$ is the number of active rows and $\sEvalValue_{n_\text{horner}} = 0$.
This is why we split the witness into columns not based on sequential chunks, but based on even and odd terms. %
We can easily adapt our custom gate to compute this new formula by introducing the extra witness column.
This approach generalizes to any number of columns:
instead of splitting the polynomial into even and odd components, we split it modulo $\ceil{\frac{d+1}{n-1}}$.

\fakeparagraph{Supporting Multiple Commitments}
Finally, we consider the case with $\ell$ commitments, beginning with the naive approach, then show how this can be extended to an efficient solution for a large number of small commitments, before introducing our optimization for multiple large commitments.
Similar to the naive approach to supporting larger commitments, we can resolve this by adding an advice column (for $\sEvalValue_i$) for each witness column.
This introduces two advice columns per commitment,
however, in cases where all commitments are small, this is highly inefficient, as the vast majority of each column will be unused.
Instead, if all commitments are sufficiently small, we can more efficiently ``stack'' multiple commitments into a single column, and make use of the same additional advice columns (and the same custom gate) by simply setting $\sEvalValue_{i,j+1}$ to zero whenever a new commitment starts.
However, when each commitment might be larger than we can accommodate in a single column, we cannot apply this technique.
Instead, our optimization relies on aggregating multiple commitments.
The key insight here is that we can use essentially the same optimized technique we used to handle multiple columns per witness to also handle multiple witnesses.
For this, we introduce the additional challenge variable \Calpha for $\alpha$, and in our custom gate, replace each occurrence of $\scommittedWitness^{(j)}_i$ with
$$ \textstyle\sum_{i=1}^\ell \alpha^i \scommittedWitness_i^{(j)}. $$
We visualize our additions to the grid in \Cref{fig:grid}.
This approach assumes that each commitment uses the same number of advice columns; for smaller commitments we pad with additional columns to match the required size.
While this method is effective when dealing with large commitments, it may be less efficient in particularly unbalanced scenarios, e.g., when a few commitments are significantly larger than the others.
In such cases, a hybrid approach that combines aggregation with the naive per-column method may offer better efficiency for larger commitments while retaining the benefits of aggregation for the rest.

\fakeparagraph{Masking}
Finally, we consider $\sMask$, which defines the degree-0 polynomial $\spolyMask(X) \coloneqq \sMask$ that masks the result of the polynomial evaluation.
For the vast majority of arithmetizations of \sRelation, there will be suitable empty cells and existing custom gates (e.g., addition or inner products) that we can reuse, in which case we only need to add a single copy constraint to link the computed value of $\sEvalValue_0$ with its copy in the addition.
In the rare cases where it is not possible to integrate this addition into the existing grid, we can add another advice column for $\mu$ that is zero except for one cell, and extend our custom gate to include this column.

\begin{remark}[I-R1CS]
\label{rem:ir1cs}
    An efficient extension transform also exists for I-R1CS.
    The transform expresses Horner's method as rank-1 constraints over the challenge variables $\alpha$ and \sEvalPoint, which adds $O(\ell \cdot \pcDbound)$ constraints.
\end{remark}


\providecommand{\cpFill}{\ifdefined\cellcolor\cellcolor{gray!15}\fi}
\providecommand{\cpFillLegend}{\ifdefined\cellcolor\colorbox{gray!15}{\strut\itshape shaded}\else\textit{grey text}\fi}
\providecommand{\cpUnsupported}[1]{{\color{black!65}\footnotesize\itshape #1}}
\providecommand{\cpUnsupportedShort}{{\color{black!65}$\bot$}}
\providecommand{\cpNotImpl}[1]{{\color{gray}\footnotesize\itshape #1}}
\providecommand{\cpNotImplShort}{{\color{black!65}$\bot$}}
\providecommand{\cpNotRun}{{\color{gray}--}}
\begin{table*}[t]
\centering
\small
\setlength{\tabcolsep}{2.5pt}
\begin{tabular}{ll@{\hspace{6.0pt}}rrrrrr@{\hspace{6.0pt}}rrrrrr}
\toprule
 &  & \multicolumn{6}{c}{Pairing-based: prover overhead} & \multicolumn{6}{c}{Transparent: prover overhead} \\
\cmidrule(lr){3-8} \cmidrule(lr){9-14}
 &  & No Com. & Artemis & Lego & Lunar & Poseidon & SHA-256 & No Com. & Artemis & Lego & Lunar & Poseidon & SHA-256 \\
\midrule
\multicolumn{2}{l}{\textit{zkML}} & \multicolumn{6}{c}{\textit{KZG (Halo2)}} & \multicolumn{6}{c}{\textit{IPA (Halo2)}} \\
\cmidrule(lr){3-14}
MNIST & {\scriptsize 8k} & 2.89 & 1.07$\times$ & \textbf{1.05$\times$} & 1.30$\times$ & 18.5$\times$ & 26.6$\times$ & 3.80 & \textbf{1.07$\times$} & \multicolumn{2}{c}{\cpFill} & 19.1$\times$ & 26.0$\times$ \\
R18 & {\scriptsize 280k} & 55.3 & 1.04$\times$ & \textbf{1.03$\times$} & 1.24$\times$ & 3.79$\times$ & 20.9$\times$ & 73.6 & \textbf{1.03$\times$} & \multicolumn{2}{c}{\cpFill} & 3.62$\times$ & 18.3$\times$ \\
Mob. & {\scriptsize 3.5M} & 1071 & 1.02$\times$ & \textbf{1.01$\times$} & 1.20$\times$ & 3.29$\times$ & \multicolumn{1}{c}{\cpNotRun} & 1286 & \textbf{1.03$\times$} & \multicolumn{2}{c}{\cpFill} & 3.23$\times$ & \multicolumn{1}{c}{\cpNotRun} \\
VGG & {\scriptsize 15M} & 459 & 1.05$\times$ & \textbf{1.04$\times$} & 1.15$\times$ & 11.7$\times$ & \multicolumn{1}{c}{\cpNotRun} & 538 & \textbf{1.14$\times$} & \multicolumn{2}{c}{\cpFill} & 11.6$\times$ & \multicolumn{1}{c}{\cpNotRun} \\
Diff. & {\scriptsize 19.5M} & 4433 & 1.02$\times$ & \textbf{1.02$\times$} & 1.14$\times$ & 3.90$\times$ & \multicolumn{1}{c}{\cpNotRun} & 5111 & \textbf{1.05$\times$} & \multicolumn{2}{c}{\cpFill} & 3.20$\times$ & \multicolumn{1}{c}{\cpNotRun} \\
GPT-2 & {\scriptsize 81M} & 3964 & 1.08$\times$ & \textbf{1.04$\times$} & 1.15$\times$ & \multicolumn{1}{c}{\cpNotRun} & \multicolumn{1}{c}{\cpNotRun} & 4658 & \textbf{1.14$\times$} & \multicolumn{2}{c}{\cpFill\multirow{-6}{*}{\cpUnsupportedShort}} & \multicolumn{1}{c}{\cpNotRun} & \multicolumn{1}{c}{\cpNotRun} \\
\midrule
\multicolumn{2}{l}{\textit{TLSNotary}} & \multicolumn{6}{c}{\textit{KZG (Barretenberg UltraHonk)}} & \multicolumn{6}{c}{\textit{WHIR (ProveKit)}} \\
\cmidrule(lr){3-14}
Bank & {\scriptsize 520\,B} & 0.019 & 1.14$\times$ & \textbf{1.08$\times$} & \multicolumn{1}{c}{\cpNotImplShort} & 7.38$\times$ & 17.8$\times$ & 0.040 & 1.21$\times$ & \multicolumn{2}{c}{\cpFill} & \textbf{1.11$\times$} & 6.29$\times$ \\
Spotify & {\scriptsize 2.3\,KiB} & 0.044 & 1.24$\times$ & \textbf{1.06$\times$} & \multicolumn{1}{c}{\cpNotImplShort} & 8.62$\times$ & 24.6$\times$ & 0.042 & 1.50$\times$ & \multicolumn{2}{c}{\cpFill} & \textbf{1.31$\times$} & 21.9$\times$ \\
Twitter & {\scriptsize 3.9\,KiB} & 0.066 & 1.18$\times$ & \textbf{1.04$\times$} & \multicolumn{1}{c}{\cpNotImplShort} & 9.21$\times$ & 22.8$\times$ & 0.045 & \textbf{1.64$\times$} & \multicolumn{2}{c}{\cpFill} & 1.70$\times$ & 35.8$\times$ \\
Discord & {\scriptsize 9.8\,KiB} & 0.12 & 1.21$\times$ & \textbf{1.04$\times$} & \multicolumn{1}{c}{\cpNotImplShort} & 10.9$\times$ & 40.7$\times$ & 0.061 & \textbf{1.70$\times$} & \multicolumn{2}{c}{\cpFill} & 2.10$\times$ & 53.8$\times$ \\
Uber & {\scriptsize 16\,KiB} & 0.16 & 1.19$\times$ & \textbf{1.03$\times$} & \multicolumn{1}{c}{\cpNotImplShort} & 13.1$\times$ & 56.7$\times$ & 0.084 & \textbf{1.62$\times$} & \multicolumn{2}{c}{\cpFill} & 2.41$\times$ & 48.2$\times$ \\
ID.me & {\scriptsize 64\,KiB} & 0.47 & 1.21$\times$ & \textbf{1.02$\times$} & \multicolumn{1}{c}{\cpNotImplShort} & 17.1$\times$ & 151$\times$ & 0.21 & \textbf{1.46$\times$} & \multicolumn{2}{c}{\cpFill\multirow{-6}{*}{\cpUnsupportedShort}} & 3.43$\times$ & 69.1$\times$ \\
\bottomrule
\end{tabular}

\caption{Prover overhead for zkML and TLSNotary, for pairing-based and transparent proof systems: Prover time is reported in seconds for No Com., and for every other scheme as a multiple of No Com. $\bot$: Lunar and LegoSNARK require a pairing-based (KZG) proof system, and Lunar also does not support the multilinear commitments of Barretenberg UltraHonk; \cpNotRun{}: out of memory.}
\label{tab:eval:prover:rel}
\end{table*}

\subsecspacingtop
  \section{Evaluation}
\subsecspacingbot
\lsec{eval}

In this section, we evaluate the performance of \oursystem and compare it against state-of-the-art approaches, namely Lunar~\cite{Campanelli2021-asiacrypt}, LegoSNARK~\cite{Campanelli2019-as}, and recomputing the commitment inside the \gls{snark} with Poseidon~\cite{Grassi2021-poseidon} or SHA-256~\cite{NIST2015-sha2}.
We evaluate across two different arithmetizations, three proof systems and three commitment types.
Our results demonstrate that \oursystem is efficient
across a wide range of proof systems and commitment
schemes, including pairing-based and transparent proof systems.

\fakeparagraph{Application: \gls{zkml}}
To demonstrate the impact of our \gls{cpsnark} construction in realistic settings, we consider an emerging application domain that is particularly well-suited for evaluating \glspl{cpsnark}: ``zkML'' proofs of correct model inference.
Modern ML models involve millions or even billions of parameters, leading to large commitments and heterogeneous access patterns that challenge both the scalability and flexibility of existing \gls{cpsnark} constructions.
In \gls{zkml}, proving correctness of inference alone is insufficient: without enforcing consistency, a prover could perform the computation using a different or modified model. Meaningful proofs must therefore guarantee that the computation was carried out with respect to a specific, externally committed model.
In practice, enforcing this form of consistency between the internal witness and large external model commitments has emerged as a dominant performance bottleneck in \gls{zkml}~\cite{Kang2022-zkml}.
As a result, \gls{zkml} workloads provide a representative stress test for the scalability limits of existing CP-SNARK constructions and offer a realistic setting in which to evaluate whether commitment consistency can be enforced efficiently at scale.

\fakeparagraph{Application: TLSNotary}
A second application domain involves proving properties of authenticated web data using TLSNotary.
TLSNotary authenticates TLS data with a notary party, which blindly signs authenticated commitments to the transcript without learning the underlying plaintext.
The prover can subsequently use these commitments to prove properties of the received data to other parties while keeping sensitive information private.
Since TLSNotary's current approach (recomputing SHA-256 in-circuit) makes commitment consistency checks very expensive, they separately commit to selected parts of the website, increasing storage overhead and reducing flexibility of later disclosures~\cite{tlsnotary_commit_strategy}.
With \oursystem, we can achieve better efficiency while committing to the entire webpage.

\subsection{Implementation}
\subsecspacingbot
\lsec{implementation}

In addition to implementing \oursystem, we provide the first (to the best of our knowledge) complete implementation of Lunar's \cplink construction~\cite{Campanelli2021-yd}.
We implement all techniques in Rust.
For \gls{zkml}, we extend the Halo2 library~\cite{halo2book}, which includes implementations for KZG-based~\cite{Kate2010-px} and IPA-based~\cite{Bowe2019-ug} zero-knowledge proofs.
We instantiate the underlying group with the pairing-friendly BN254 curve for KZG-based proofs and the Pallas curve for the IPA-based proofs.
As our \gls{zkml} workload, we adopt the Halo2 arithmetizations of Chen et al.'s ZKML framework~\cite{Chen2024-ng}, and extend them with the advice columns required to support the different \gls{cpsnark} constructions.
For TLSNotary, we use the workloads expressed in Noir~\cite{noir}, using both the pairing-based Barretenberg backend~\cite{barretenberg} and World's ProveKit backend~\cite{provekit}, which combines Spartan~\cite{Setty2020-spartan} with the hash-based WHIR polynomial commitment scheme~\cite{whir}.
For ProveKit, we add the Horner constraints to the compiled I-R1CS.
For Barretenberg, which does not natively expose verifier challenges in its implementation, we extend its C++ prover with a small patch that exports the committed wire columns.
Our prover reads these columns and their commitments directly and uses them in a Horner evaluation written in Noir, which Barretenberg compiles to Plonkish gates.
As external commitment, we use a hiding, additively homomorphic polynomial commitment of Pedersen form, with an opening proof matched to the proof system. In the pairing-based settings (Halo2-KZG, Barretenberg), the generators are the SRS powers, enabling KZG opening proofs. In the transparent settings (Halo2-IPA, ProveKit), the Pedersen vector commitment (cf.~\Cref{def:pedersen_vector}) is opened with a zero-knowledge inner-product argument.
We make all our implementations and benchmarking configurations available as open-source\footnote{\ifdefined\isnotanon
\href{\githuburl}{\githuburl}
\else
\href{\githuburlanon}{\githuburlanon}
\fi}.
Below, we discuss the implementation of each of the approaches we evaluate in more detail:

\begin{itemize}[topsep=0.5em,label={},leftmargin=0pt,align=left,labelsep=0pt,labelwidth=!]
\setlength\itemsep{8pt}
    \item \gls{no_com}: Our baseline does not perform any kind of commitment checks and utilizes the ZKML arithmetizations developed by Chen et al.~\cite{Chen2024-ng}. For TLSNotary, we use the example Noir circuit~\cite{tlsn_examples_zk} for age verification and set the commitment size to match the server response size for each webpage as specified in the repository~\cite{tlsn_extension}.

    \item \gls{sha256}:
    For \gls{zkml}, we generalize the Table16 SHA-256 chip of the Halo2 gadget library~\cite{halo2book}, which is fixed to the Pallas base field, to work over any field, allowing it to run for BN254 used by the KZG commitments.
    For TLSNotary, we hash the response with Noir's SHA-256 implementation~\cite{noir_sha256}.

    \item \gls{poseidon}: For \gls{zkml}, we use a Poseidon~\cite{Grassi2021-poseidon} gadget provided by the Halo2 library~\cite{halo2book}. For TLSNotary, we use the same instantiation from Noir's Poseidon library~\cite{noir_poseidon} on both backends.

    \item \gls{zkfft}: We implement Jivanyan et al.'s Halo2-based \gls{cpsnark} for Pedersen commitments~\cite{Jivanyan2024-zkfft} based on the code provided by the authors~\cite{zkfft_github}.
The protocol introduces an inner product whose size grows quadratically with the committed witness, leading to prohibitive overhead even for small models (152s proving and 8s verification on MNIST); we therefore omit further evaluation.

    \item \gls{legosnark}:
    While initially proposed in the context of R1CS and Groth16~\cite{Campanelli2019-as},
    we extend the construction to the Plonkish setting for KZG commitments.
    We implement LegoSNARK's linking proof \cplink~\cite{Campanelli2019-as} in Rust in the Halo2 library and on top of Barretenberg’s UltraHonk.
    For Barretenberg, we run the \cplink prover on the wire columns of a real proof.

    \item \gls{cp_link}: We implement Lunar's \gls{cpsnark} construction~\cite{Campanelli2021-yd} for Halo2's Plonkish arithmetization.
    Specifically, we implement \cplinkone and \cplinktwo from \cite{Campanelli2021-yd}.
    We use Halo2's underlying finite field library \texttt{ff}.
    \cplink relies heavily on division of vanishing polynomials on a subset of the evaluation domain, which is not directly supported by Halo2's polynomial implementation.
    Therefore, we extend this implementation with support for efficient FFT-based polynomial division to ensure competitive performance of \cplink.

    \item \gls{poly}: For the \gls{zkml} implementation of \oursystem (cf.~\Cref{sec:design:oursystem}), we use Halo2's standard implementation of polynomial commitments and implement the arithmetization of polynomial evaluation using Horner's method (cf.~\Cref{sec:arithmetization}) as a gadget in the Halo2 library.
    We leverage the Halo2 Challenge API, which allows the definition of custom gates parameterized by verifier challenges~\cite{halo2_challenge_api}.
    For TLSNotary, we extend the existing implementation with the Horner constraints over the committed response
    in the arithmetization of the underlying proof system: Plonkish for Barretenberg and R1CS for ProveKit.

\end{itemize}
\begin{table*}[t]
\centering
\setlength{\tabcolsep}{2.0pt}
\resizebox{\textwidth}{!}{%
\begin{tabular}{ll*{6}{c}*{6}{c}*{6}{c}}
\toprule
 &  & \multicolumn{6}{c}{Proof Size (kB)} & \multicolumn{6}{c}{Verifier Time (ms)} & \multicolumn{6}{c}{Prover Memory (GB)} \\
\cmidrule(lr){3-8} \cmidrule(lr){9-14} \cmidrule(lr){15-20}
\multicolumn{2}{l}{\textit{zkML}} & MNST & R18 & Mob & VGG & Diff & GPT & MNST & R18 & Mob & VGG & Diff & GPT & MNST & R18 & Mob & VGG & Diff & GPT \\
\midrule
\multirow{6}{*}{\rotatebox{90}{KZG}} & No Com. & 9 & 14 & 18 & 16 & 32 & 15 & 6 & 8 & 9 & 9 & 15 & 9 & 0.48 & 10.47 & 151 & 71.81 & 515 & 528 \\
 & Artemis & 10 & 15 & 18 & 17 & 33 & 15 & 6 & 8 & 9 & 12 & 16 & 11 & 0.52 & 10.86 & 155 & 80.52 & 538 & 599 \\
 & Lego & 10 & 15 & 18 & 16 & 33 & 15 & 7 & 9 & 10 & 11 & 16 & 10 & 0.51 & 10.59 & 152 & 74.99 & 522 & 549 \\
 & Lunar & 12 & 18 & 22 & 20 & 39 & 18 & 23 & 32 & 40 & 49 & 66 & 40 & 0.48 & 10.48 & 150 & 71.19 & 518 & 526 \\
 & Poseidon & 12 & 16 & 21 & 15 & 25 & - & 7 & 8 & 10 & 9 & 42 & - & 9.16 & 18.06 & 362 & 547 & 1433 & - \\
 & SHA-256 & 17 & 100 & - & - & - & - & 9 & 38 & - & - & - & - & 13.00 & 144 & - & - & - & - \\
\midrule
\multirow{6}{*}{\rotatebox{90}{IPA}} & No Com. & 10 & 16 & 19 & 17 & 34 & 17 & 25 & 279 & 1.7k & 699 & 3.1k & 8.5k & 0.34 & 7.62 & 132 & 65.23 & 469 & 486 \\
 & Artemis & 11 & 16 & 20 & 20 & 36 & 19 & 26 & 283 & 1.7k & 3.3k & 10.1k & 35.2k & 0.38 & 7.93 & 135 & 68.04 & 478 & 505 \\
 & Lego & - & - & - & - & - & - & - & - & - & - & - & - & - & - & - & - & - & - \\
 & Lunar & - & - & - & - & - & - & - & - & - & - & - & - & - & - & - & - & - & - \\
 & Poseidon & 14 & 18 & 23 & 17 & 23 & - & 283 & 490 & 3.7k & 6.4k & 11.8k & - & 6.97 & 15.79 & 318 & 476 & 1260 & - \\
 & SHA-256 & 18 & 101 & - & - & - & - & 282 & 543 & - & - & - & - & 9.93 & 128 & - & - & - & - \\
\midrule
\multicolumn{2}{l}{\textit{TLSNotary}} & Bank & Spotify & Twitter & Discord & Uber & ID.me & Bank & Spotify & Twitter & Discord & Uber & ID.me & Bank & Spotify & Twitter & Discord & Uber & ID.me \\
\midrule
\multirow{6}{*}{\rotatebox{90}{KZG}} & No Com. & 15 & 15 & 15 & 15 & 15 & 15 & 6.4 & 6.5 & 6.4 & 6.4 & 6.4 & 6.4 & 0.05 & 0.05 & 0.05 & 0.06 & 0.08 & 0.18 \\
 & Artemis & 15 & 15 & 15 & 15 & 15 & 15 & 8.3 & 8.3 & 8.3 & 8.3 & 8.3 & 8.2 & 0.05 & 0.05 & 0.05 & 0.07 & 0.09 & 0.21 \\
 & Lego & 15 & 15 & 15 & 15 & 15 & 15 & 7.6 & 7.6 & 7.6 & 7.6 & 7.6 & 7.6 & 0.05 & 0.05 & 0.05 & 0.06 & 0.08 & 0.18 \\
 & Lunar & - & - & - & - & - & - & - & - & - & - & - & - & - & - & - & - & - & - \\
 & Poseidon & 15 & 15 & 15 & 15 & 15 & 15 & 6.4 & 6.4 & 6.4 & 6.4 & 6.4 & 6.4 & 0.07 & 0.12 & 0.17 & 0.34 & 0.52 & 1.88 \\
 & SHA-256 & 15 & 15 & 15 & 15 & 15 & 15 & 6.6 & 6.6 & 6.6 & 6.6 & 6.6 & 6.6 & 0.13 & 0.31 & 0.47 & 1.19 & 2.55 & 39.03 \\
\midrule
\multirow{6}{*}{\rotatebox{90}{WHIR}} & No Com. & 336 & 338 & 338 & 357 & 390 & 446 & 14 & 15 & 16 & 19 & 21 & 38 & 0.02 & 0.02 & 0.02 & 0.03 & 0.04 & 0.15 \\
 & Artemis & 338 & 338 & 339 & 358 & 392 & 445 & 18 & 21 & 22 & 27 & 30 & 54 & 0.02 & 0.02 & 0.02 & 0.03 & 0.05 & 0.15 \\
 & Lego & - & - & - & - & - & - & - & - & - & - & - & - & - & - & - & - & - & - \\
 & Lunar & - & - & - & - & - & - & - & - & - & - & - & - & - & - & - & - & - & - \\
 & Poseidon & 339 & 337 & 358 & 391 & 433 & 498 & 17 & 24 & 32 & 55 & 85 & 317 & 0.02 & 0.04 & 0.05 & 0.11 & 0.17 & 0.60 \\
 & SHA-256 & 446 & 521 & 538 & 591 & 591 & 623 & 61 & 215 & 296 & 692 & 990 & 3.8k & 0.19 & 0.64 & 1.16 & 2.43 & 2.92 & 19.45 \\
\bottomrule
\end{tabular}}
\caption{Proof size, verifier time, and prover memory for pairing-based and transparent proof systems. zkML (Halo2 with KZG or IPA): \gls{mnist} (MNST), \gls{resnet18} (R18), \gls{mobilenet} (Mob), \gls{vgg} (VGG), \gls{diffusion} (Diff), and \gls{gpt2} (GPT).}
\label{tab:results}
\end{table*}

Note that for very large models (e.g., \gls{diffusion} or \gls{gpt2}),
adding the \gls{sha256} and \gls{poseidon} gadgets requires more rows than available with curve BN254,
as a result we switch to curve BLS12-381 for these baselines.
Lunar (and zkFFT) assume that the external commitment(s) fit into the proof's evaluation domain.
However, this is not the case for several of the larger models in our evaluation.
Therefore, we split the external commitment into column-sized chunks where necessary, side-stepping the overhead of extending the evaluation domain for the entire proof.
Lunar can only show shift-and-link on contiguous, in-order column-wise chunks of the witness.
While the witness layout used by Chen et al.'s ZKML~\cite{Chen2024-ng} is mostly contiguous, it uses a row-wise layout which substantially increases the number of shifting proofs required.
Therefore, we transpose the witness layout in order to make Lunar's approach viable.

\fakeparagraph{Experimental Setup}
We evaluate the prover time, prover memory, verifier time, and proof size.
For \gls{zkml}, we evaluate inference proofs with a commitment to the model for six different models, ranging from MNIST to GPT-2.
For TLSNotary, we evaluate the same age check application across webpages of different sizes.
We describe the benchmarks we consider in our evaluation and the machine configuration used for each in Appendix~\ref{apx:expsetup}.

\subsection{Results}
\label{sec:eval:results}
In the following, we present the results of our experimental evaluation.
In~\Cref{tab:eval:prover:rel} we report prover times, while we report prover memory, proof sizes, and verifier times in~\Cref{tab:results}.

\fakeparagraph{Prover Overhead}
We begin by discussing prover overhead (cf.~\Cref{tab:eval:prover:rel}), which is by far the most important metric when considering the practicality of both \gls{zkml} and TLSNotary.
For \gls{zkml}, recomputing a \gls{sha256} commitment inside the \gls{snark} results in prohibitive overhead, with prover costs ranging from $18\times$ to $27\times$ the baseline for \gls{mnist} and \gls{resnet18}.
For the remaining models, the \gls{sha256} experiments exhausted the memory available on the corresponding instance types.
\gls{poseidon} also incurs a significant overhead that depends on the number of parameters relative to the model's computation (FLOPs),
ranging from $3.3\times$--$18.5\times$ for KZG, and from $3.2\times$--$19.1\times$ for IPA compared to the baseline (\gls{no_com}).
Note that for \gls{gpt2}, \gls{poseidon} was unable to complete successfully because of memory requirements beyond the 4TB on the largest AWS instance available to us.
By reusing the \gls{snark}'s internal witness commitment, \gls{cp_link} reduces prover costs to $1.14\times$--$1.3\times$ the baseline for some models, but supports only pairing-based proof systems.
Our results show that generality comes at little additional cost.
While \gls{legosnark} is the fastest construction, incurring just $5\%$ overhead, \gls{poly} follows closely at $8\%$.
In the transparent setting, which neither \gls{cp_link} nor \gls{legosnark} supports, \oursystem' overhead remains similarly low, at most $14\%$.

TLSNotary proofs involve a much smaller computation than \gls{zkml} inference, so the commitment check makes up a larger share of the prover time.
Recomputing \gls{sha256} inside the \gls{snark} remains impractical for full-page proofs, increasing prover time by $17.8\times$--$151\times$ with Barretenberg and $6.3\times$--$69.1\times$ with ProveKit, up to $71$\,s for the largest page.
\gls{poseidon} is considerably cheaper, at $7.4\times$--$17.1\times$ with Barretenberg and $1.1\times$--$3.4\times$ with ProveKit, but relies on an arithmetization-oriented hash function that has received far less cryptanalysis than \gls{sha256}.
With Barretenberg, both \oursystem and \gls{legosnark} outperform \gls{poseidon} on every page, with overheads of at most $24\%$ and $8\%$, respectively, while \gls{cp_link} does not support UltraHonk.
In the transparent setting, the small workload highlights the fixed cost of \oursystem' additional opening proof: overhead grows to $21\%$--$70\%$, and \gls{poseidon} is faster for the two smallest pages.
In absolute terms, however, \oursystem adds at most $0.10$\,s to the prover time, and outperforms \gls{poseidon} on all pages of 3.9\,KiB and larger, by up to $2.3\times$ for ID.me.
These results demonstrate the generality of \oursystem and the concrete efficiency of \oursystem' adapted arithmetization, which incorporates a polynomial evaluation gate based on Horner's method.

\ifdefined\isnotextended
\else
\fi

We present prover memory costs in \Cref{tab:results}.
In general, larger models require significantly more memory to prove, even in the \gls{no_com} baseline.
For \gls{zkml}, \gls{sha256} and \gls{poseidon} introduce a prohibitively high memory overhead which, as mentioned above, prevents \gls{sha256} from running beyond \gls{resnet18} and \gls{poseidon} from running \gls{gpt2} with IPA.
The other constructions incur significantly smaller memory overheads compared to \gls{no_com}: $1.02\times$--$1.13\times$ for \gls{poly}, $1.01\times$--$1.06\times$ for \gls{legosnark}, and less than $1\%$ for \gls{cp_link}.
\gls{poly} requires slightly more memory than \gls{cp_link} ($1.03\times$--$1.14\times$) and \gls{legosnark} ($1.02\times$--$1.09\times$) because it adds another set of columns for the polynomial evaluation.
For TLSNotary, we see a similar pattern for the recomputation approaches: for the largest setting, \gls{sha256} increases prover memory by $217\times$ with Barretenberg and $130\times$ with ProveKit, and \gls{poseidon} by $10\times$ and $4\times$, whereas \oursystem needs at most $0.21$\,GB, compared to $0.18$\,GB for \gls{no_com} and \gls{legosnark}.

\fakeparagraph{Verifier}
We present the verifier times in~\Cref{tab:results}.
KZG-based proof systems provide a verification time constant in the size of the witness.
However, even for the baseline (\gls{no_com}) the verifier times for different benchmarks still vary, because different model output sizes result in different proof instance sizes.
Similarly, merely adding the commitments to the instance increases the KZG verifier time.
However, the vast majority of the differences in verifier time between the different approaches are due to the additional checks that (some of) the constructions introduce.
In contrast, verifier times for IPA scale with the size of the witness so we expect slower verification times in general.

For KZG, \gls{poseidon} shows a negligible increase in verification time as it only adapts the arithmetization of the relation and not the \gls{snark}, resulting in a tiny increase in verification time due to the addition of the commitment to the public input.
The exception is \gls{poseidon} for \gls{diffusion}, which increases verification time by $2.8\times$ because this configuration uses the BLS12-381 curve to fit the larger model, as mentioned above.
\gls{cp_link} (which only supports KZG) increases the verification time compared to \gls{no_com} significantly by $3.8\times$--$5.4\times$,
as it requires a linear number of additional pairing operations to verify the \cplink proofs.
\gls{legosnark} incurs a small overhead of $1.1\times$--$1.2\times$.
\gls{poly} results in an overhead similar to \gls{legosnark} ($1.0\times$--$1.3\times$) by requiring only two additional pairing operations when using KZG.
For IPA, \gls{poly}' verification overhead is negligible for the three smaller models, but reaches $3.3\times$--$4.7\times$ for \gls{vgg}, \gls{diffusion}, and \gls{gpt2}.
This overhead remains lower than for \gls{poseidon}, for which we observe a considerable increase in verifier time ($1.8\times$--$11\times$), due to the large size of its arithmetization, which impacts IPA verification times.
For TLSNotary, the commitment check adds at most $1.8$\,ms to the verifier time of \oursystem with Barretenberg ($1.2$\,ms for \gls{legosnark}) and at most $16$\,ms with ProveKit, compared to up to $0.28$\,s for \gls{poseidon} and $3.8$\,s for \gls{sha256} with ProveKit.

\fakeparagraph{Proof Size}
We report proof sizes in~\Cref{tab:results}.
In general, proof sizes are small (a few dozen kB at most) for the baseline (\gls{no_com}) across all \gls{zkml} models and for TLSNotary with Barretenberg. The WHIR proofs are larger, $336$--$446$\,kB, due to their hash-based polynomial commitments.
Furthermore, the overhead of adding commitment verification is generally low ($\leq 1.2\times$ for \gls{poly} and \gls{legosnark}) with the exception of \gls{sha256} ($1.8\times$--$7.1\times$ for \gls{zkml} and $1.3\times$--$1.7\times$ with WHIR), \gls{poseidon} (up to $1.4\times$), and \gls{cp_link}, where the linking proofs noticeably increase proof size ($1.2\times$--$1.3\times$), particularly as the number of witness chunks grows.
These results demonstrate that \oursystem has similar performance to the fastest \gls{cpsnark} construction,
while being compatible with a wide range of proof systems and commitment
schemes.


\section*{Open Science}

\ifdefined\isnotanon
Our code is available at
\href{\githuburl}{\githuburl}.
\else
Our code is available (in anonymized form) at
\href{\githuburlanon}{\githuburlanon}.
\fi
Our evaluation is structured using an open-source experimentation framework, which allows straightforward reproducibility of results by defining experiments in a configuration file (suite) and executing them on a set of machines.
We provide the necessary commands to reproduce all results and tables shown in the paper.
\ifdefined\isnotextended
In addition, the repository contains a PDF file with additional proof details.
\fi

\ifdefined\isnotanon
\section*{Acknowledgements}
We would like to thank Christian Knabenhans for his insightful feedback.
This work was partially supported by an unrestricted gift from Google and by the Natural Sciences and Engineering Research Council of Canada (NSERC) through a Discovery Grant.

\else
\fi

\bibliographystyle{IEEEtran}
\interlinepenalty=10000 %
\bibliography{references,additional_references,extended}

\appendices
\crefalias{section}{appendix}%
\crefname{appendix}{Appendix}{Appendices}%
\Crefname{appendix}{Appendix}{Appendices}%


\section{Definitions}
\label{apx:definitions}

\begin{definition}[Commitment Scheme]
    \ldef{commitments}
    A non-interactive commitment scheme consists of a message space $\scomMessageSpace$, randomness space $\scomRandomnessSpace$,
    a commitment space $\scomCommitmentSpace$
    and a tuple of polynomial-time algorithms $(\sComSetup,\sComCommit,\sComVerify)$
    defined as follows:
    \begin{algos}
        \item $\sComSetup(1^\lambda) \rightarrow \ck$: Given a security parameter $\lambda$, it outputs a commitment key $\ck$.

        \item $\sComCommit(\ck, m, r) \rightarrow c$: Given a commitment key $\ck$, a message $m \in \scomMessageSpace$ and randomness $r \in \scomRandomnessSpace$, it outputs a commitment $c$.

        \item $\sComVerify(\ck, c, m, r) \rightarrow \{0,1\}$: Given a commitment key $\ck$, a commitment $c$, a decommitment $r$, and a message $m$, it outputs $1$ if the commitment is valid, otherwise $0$.
    \end{algos}
    A non-interactive commitment scheme has the following properties:
    \begin{algos}
        \item \fakeparagraph{Correctness}
        For all security parameters $\lambda$, for all $m$ and for all $\ck$ output by $\sComSetup(1^\lambda)$, if $c = \sComCommit(\ck, m, r)$, then $\sComVerify(\ck, c, m, r) = 1$.

        \item \fakeparagraph{Binding}
        For all polynomial-time adversaries $\mathcal{A}$, the probability
        \begin{equation*}
            \begin{split}
                \Pr\bigl[\sComVerify(\ck, c, m_1, r_1) &= 1 \land \\
                \sComVerify(\ck, c, m_2, r_2) &= 1 \land m_1 \neq m_2 : \\
                \ck \leftarrow \sComSetup(1^\lambda), (c, r_1, r_2, & m_1, m_2) \leftarrow \mathcal{A}(\ck) \bigr]
            \end{split}
        \end{equation*}
        is negligible.

        \item \fakeparagraph{Hiding}
        For all polynomial-time adversaries $\mathcal{A}$ and all messages $m_1, m_2 \in \scomMessageSpace$, the advantage
        \begin{equation*}
            \bigl|\Pr[\mathcal{A}(\ck, c_1) = 1] - \Pr[\mathcal{A}(\ck, c_2) = 1]\bigr|
        \end{equation*}
        is negligible, where $\ck \leftarrow \sComSetup(1^\lambda)$, $r \sample \scomRandomnessSpace$, and $c_b \leftarrow \sComCommit(\ck, m_b, r)$ for $b \in \{1, 2\}$.
    \end{algos}
\end{definition}

\begin{definition}[Homomorphic Commitment Scheme~\cite{Bunz2018-mg}]
    \ldef{homomorphic_commitment}
    A homomorphic commitment scheme is a non-interactive commitment scheme such that
    \scomMessageSpace, \scomRandomnessSpace and \scomCommitmentSpace are all abelian groups and
    for all $m_1, m_2 \in \scomMessageSpace$ and $r_1, r_2 \in \scomRandomnessSpace$, we have
    \begin{equation*}
        \begin{split}
            &\sComCommit(\ck, m_1 + m_2, r_1 + r_2) = \\
            &\sComCommit(\ck, m_1, r_1) + \sComCommit(\ck, m_2, r_2).
        \end{split}
    \end{equation*}
    If, in addition, \scomMessageSpace, \scomRandomnessSpace and \scomCommitmentSpace are $\sfield$-vector spaces for a prime field \sfield (e.g., all three have prime exponent $|\sfield|$), the scheme is \emph{$\sfield$-linear}: for all $\lambda_1,\ldots,\lambda_n \in \sfield$, messages $m_j$ and randomness $r_j$,
    \begin{equation*}
        \begin{split}
            &\sComCommit\bigl(\ck, \textstyle\sum_{j} \lambda_j m_j, \textstyle\sum_{j} \lambda_j r_j\bigr) = \\
            &\textstyle\sum_{j} \lambda_j \cdot \sComCommit(\ck, m_j, r_j).
        \end{split}
    \end{equation*}
    For a homomorphic polynomial commitment scheme (\Cref{def:pc}), the message space is $\sfield^{\le \pcDbound}[X]$ and $\sfield$-linearity follows from additive homomorphism, since the groups have prime exponent $|\sfield|$. Both KZG commitments~\cite{Kate2010-px} and Pedersen commitments with inner-product openings (\Cref{def:pedersen_vector}) are $\sfield$-linear.
\end{definition}

\begin{definition}[Pedersen Vector Commitments]
    \ldef{pedersen_vector}
    Let $\sgroup$ be a cyclic group of prime order $p$.
    Let $\sgenerator_1,\ldots,\sgenerator_n,\sgenerator_{n+1} \in \sgroup$ be generators such that all discrete logarithm relations between them are unknown.
    Let the message space be $\sfield_p^n$ and the randomness space be $\sfield_p$.
    The Pedersen vector commitment scheme is defined as follows:
    \begin{algos}
        \item $\sPedSetup(1^\lambda, n) \rightarrow \ck$: Sample generators $\sgenerator_1,\ldots,\sgenerator_n,\sgenerator_{n+1} \in \sgroup$. Output $\ck = (\sgenerator_1,\ldots,\sgenerator_n,\sgenerator_{n+1})$.

        \item $\sPedCommit(\ck, \vec{m}, r) \rightarrow c$: Given $\vec{m} = (m_1,\ldots,m_n) \in \sfield_p^n$ and $r \in \sfield_p$, output
        $
            c = \sum_{i=1}^n m_i \cdot \sgenerator_i + r \cdot \sgenerator_{n+1}.
        $

        \item $\sPedVerify(\ck, c, \vec{m}, r) \rightarrow \{0,1\}$: Output $1$ if
        $
            c = \sum_{i=1}^n m_i \cdot \sgenerator_i + r \cdot \sgenerator_{n+1},
        $
        otherwise output $0$.
    \end{algos}

    \noindent Pedersen commitments are perfectly hiding and computationally binding under the discrete logarithm assumption.
    Moreover, they are homomorphic.
\end{definition}

Pedersen vector commitments can be interpreted as a special case of polynomial commitments (cf.~\Cref{def:pc}).
A Pedersen commitment $c = \sPedCommit(\ck,\vec{g}, r)$ corresponds to an inner-product-based commitment to the polynomial
$\spolySingle(X) = \sum_{j=1}^{n} g_j X^{j-1}$, where the vector $\vec{g}$ represents the polynomial coefficients.
Given a challenge point chosen by the verifier, the opening proof can be realized using the zero-knowledge variant of the inner-product argument~\cite{Bowe2019-ug}.

\begin{definition}
    \label{def:zk-sarks}
    Let \sRelation be a witness relation and let $\aip \in \arithK$ encode an \gls{aip} for \sRelation.
    A \gls{zksark} $\zkSARK = (\zkSetup, \langle \zkProver, \zkVerifier \rangle)$ for \sRelation is a public-coin interactive argument with the following properties, where $\tau$ denotes the transcript of an execution:
\begin{algos}
    \item \fakeparagraph{Completeness}
    For every true statement for the relation \sRelation
    an honest prover with a valid witness always convinces the
    verifier, i.e.,
    $  \forall (\sInstance, \sWitness) \in \sRelation{} :$
    \begin{equation*}
        \condprob{
            \zkVerifier \text{ accepts } \tau
        }{
            \begin{gathered}
            (\ipk, \ivk) \gets \zkSetup(\aip)\\
            \tau \gets \langle \zkProver(\ipk, \sInstance, \sWitness), \\
            \qquad \zkVerifier(\ivk, \sInstance) \rangle
            \end{gathered}
        } = 1
    \end{equation*}

    \item \fakeparagraph{Knowledge Soundness} For every PPT prover \malProver, there exists an expected polynomial-time extractor $\mathcal{E}$ that gets full access to the state of \malProver (including its random coins and inputs).
    \malProver receives the keys and outputs an instance $\tilde{\sInstance}$.
    $\mathcal{E}$ then interacts with \malProver in the role of $\zkVerifier(\ivk, \tilde{\sInstance})$, possibly rewinding it, and finally outputs a witness $\sWitness'$.
    Whenever \malProver convinces the verifier in the main transcript, the extractor computes a witness with high probability:
    $\forall \malProver{}\ \exists \mathcal{E} : $
    \begin{equation*}
    \begin{gathered}
        \condprob{
            \begin{gathered}
                \zkVerifier \text{ accepts } \tau\\
                {}\land{}\\
                (\tilde{\sInstance}, \sWitness') \notin \sRelation
            \end{gathered}
        }{
            \begin{gathered}
            (\ipk, \ivk) \gets \zkSetup(\aip)\\
            (\tilde{\sInstance}, \tau; \sWitness') \gets \\
            \qquad \langle \malProver, \mathcal{E} \rangle(\ipk, \ivk)
            \end{gathered}
        } \\
        = \negl
        \end{gathered}
    \end{equation*}

    \noindent
    We stress here that this definition assumes a \emph{non-black-box} extractor, i.e., the extractor gets full access to the adversary's state.

    \noindent
    \item \fakeparagraph{Succinctness} For any \sInstance and \sWitness, the length of the transcript is given by $|\tau| = \poly \cdot \pcpolynomialstyle{polylog}(|\sInstance| + |\sWitness|)$.

    \item \fakeparagraph{Honest Verifier Zero-Knowledge} There exists a PPT simulator $\sdv = (\textsf{Setup}, \textsf{Sim})$ such that $\sdv.\textsf{Setup}$ outputs simulated keys $(\ipk, \ivk)$ and a trapdoor \td{}; on input \td{}, \sInstance and the coins $\vec{\theta}$ of \zkVerifier, $\sdv.\textsf{Sim}$ outputs a simulated transcript $\tau$, and for all PPT adversaries $\adv = (\adv_1, \adv_2)$,
    \begin{align*}
        &\left|\condprob{
            \begin{gathered}
            (\sInstance, \sWitness) \in \sRelation \\
            {}\land{} \\
            \adv_2(\mathsf{st}, \tau) = 1
            \end{gathered}
        }{
            \begin{gathered}
            (\ipk, \ivk) \gets \zkSetup(\aip) \\
            (\sInstance, \sWitness, \mathsf{st}) \gets \adv_1(\ipk, \ivk) \\
            \text{sample coins } \vec{\theta} \text{ for } \zkVerifier \\
            \tau \gets \langle \zkProver(\ipk, \sInstance, \sWitness), \\
            \qquad \zkVerifier(\ivk, \sInstance; \vec{\theta}) \rangle
            \end{gathered}
        }
        - \right.
        \\
        &\left.\condprob{
            \begin{gathered}
            (\sInstance, \sWitness) \in \sRelation \\
            {}\land{} \\
            \adv_2(\mathsf{st}, \tau) = 1
            \end{gathered}
        }{
            \begin{gathered}
            (\ipk, \ivk, \td) \gets \sdv.\textsf{Setup}(\aip) \\
            (\sInstance, \sWitness, \mathsf{st}) \gets \adv_1(\ipk, \ivk) \\
            \text{sample coins } \vec{\theta} \text{ for } \zkVerifier \\
            \tau \gets \sdv.\textsf{Sim}(\td, \sInstance, \vec{\theta})
            \end{gathered}
        }
        \right| \\
        &= \negl
    \end{align*}

    \noindent
    This notion is zero-knowledge for an \emph{honest verifier}, because the simulator receives the coins $\vec{\theta}$ of the verifier.
    After the Fiat-Shamir transform, it yields zero-knowledge in the random oracle model.

\end{algos}
\end{definition}

\begin{lemma}[Multivariate DeMillo--Lipton--Schwartz--Zippel~\cite{Demillo1978-SZDL}]
\label{lemma:schwartz_zippel}
Let $f \in \sfield[X_1,\ldots,X_m]$ be a non-zero polynomial of total degree $\delta$ and let $S \subseteq \sfield$ be finite.
For $r_1,\ldots,r_m \sample S$ independently, $\Pr[f(r_1,\ldots,r_m) = 0] \le \delta / |S|$.
\end{lemma}

\begin{figure*}
\centering
\begin{tcolorbox}[]
\small
\fakeparagraphnovskip{\textnormal{\hornerGate}$(\aip, \sIndexSetCom)$}
On input an I-Plonkish instance \aip (\Cref{def:iplonkish:concise}) with \nRounds rounds, $n$ rows, $n_f$ fixed columns, $n_a$ advice columns, gate constraints $B$ and copy constraints $(P_\sigma, \sigma)$, and an index set \sIndexSetCom.
For $t \in [\ell]$, let $\mathsf{cells}^t$ be the ordered list of advice cells in which \aip sends the $t$-th committed sub-witness, i.e., the image of $\sIndexSetCom^t$ under \msgMap, where $(c, r) \in \mathsf{cells}^t$ designates the cell $A[c,r]$ in column $c$ and row $r$.
Let $m \coloneqq \max_t \ceil{\abs{\mathsf{cells}^t} / n}$ be the number of columns per commitment, let $n_\text{horner} \coloneqq \max_t \ceil{\abs{\mathsf{cells}^t} / m}$ be the number of active rows, and assume $n_\text{horner} < n$ (otherwise, increase $m$ by one).
Let $\idxWit \coloneqq n_a + 1$, $\idxMu \coloneqq n_a + \ell m + 1$, $\idxRho \coloneqq n_a + \ell m + 2$, and let $\idxWit[t, j] \coloneqq \idxWit + (t-1) m + j - 1$ denote the $j$-th column of commitment $t$.
Extend \aip to the $(\nRounds + 2)$-round I-Plonkish instance \aipExt as follows:
\begin{itemize}
    \item \textbf{Round $\nRounds + 1$ (witness copies and mask).}
    Add a phase with $n_a^{(\nRounds+1)} \coloneqq \ell m + 1$ advice columns $\idxWit, \ldots, \idxMu$ and $\nChal_{\zkc_{\nRounds+1}} \coloneqq 2$ challenges, whose challenge variables we name \Calpha and \Cbeta.
    The honest prover sends $\pMsg_{\nRounds+1} = A_{\nRounds+1}$ containing, for each $t \in [\ell]$, the committed sub-witness $\scommittedWitness_t = (A[c,r])_{(c,r) \in \mathsf{cells}^t}$ laid out in row-major order over the columns $\idxWit[t, 1], \ldots, \idxWit[t, m]$ and rows $1, \ldots, n_\text{horner}$, zeros in the remaining cells of these columns, and the mask \sMask in the cell $A[\idxMu, 1]$.
    The verifier responds with $(\alpha, \sEvalPoint) \sample \sfield^2$.
    \item \textbf{Round $\hRound \coloneqq \nRounds + 2$ (evaluation).}
    Add a phase with a single advice column \idxRho and $\nChal_{\zkc_{\nRounds+2}} \coloneqq 0$.
    The honest prover sends $\pMsg_{\nRounds+2} = A_{\nRounds+2}$ with $A[\idxRho, n_\text{horner} + 1] \coloneqq 0$ and, for $r = n_\text{horner}$ down to $1$,
    \begin{equation*}
        A[\idxRho, r] \coloneqq [r = 1] \cdot \sMask + \sum_{j=1}^{m} \sum_{t=1}^{\ell} A[\idxWit[t,j], r] \cdot \alpha^t \cdot \sEvalPoint^{j-1} + \sEvalPoint^m \cdot A[\idxRho, r+1],
    \end{equation*}
    so that $A[\idxRho, 1] = \sMask + \sum_{t=1}^{\ell} \alpha^t \cdot \mathbf{w}_t(\sEvalPoint) \eqqcolon \sEvalValue$.
    Let $\hIdx$ be the position of the cell $A[\idxRho, 1]$ in $\pMsg_{\nRounds+2}$.
    \item \textbf{Fixed columns.}
    Extend $F$ with three columns: $F[n_f+1, \cdot]$ is $1$ in the active rows $1, \ldots, n_\text{horner}$, $F[n_f+2, \cdot]$ is $1$ in row $1$, and $F[n_f+3, \cdot]$ is $1$ in row $n_\text{horner} + 1$; all three are $0$ elsewhere.
    \item \textbf{Copy constraints.}
    For each $t \in [\ell]$ and each $o \in [\abs{\mathsf{cells}^t}]$ with $(c, r) = \mathsf{cells}^t[o]$, add the cells $A[c, r]$ and $A[\idxWit[t, ((o-1) \bmod m) + 1], \ceil{o / m}]$ to $P_\sigma$ and join their cycles in $\sigma$, so that the copies equal $\msgProjI(\pMsgs)$.
    Join every remaining cell in rows $1, \ldots, n_\text{horner}$ of the columns $\idxWit, \ldots, \idxWit + \ell m - 1$ with the fixed cell $F[n_f + 1, n]$, which is zero.
    \item \textbf{Gate constraint.}
    Extend $B$ with the Horner's method gate $\polyHornerGate$ over the rotations $\Delta_\mathsf{h} = \{0, +1\}$:
    \begin{equation*}
        \begin{split}
            \polyHornerGate(r, \Calpha, \Cbeta) \coloneqq\ & F[n_f+1, r] \cdot \left( \sum_{j=1}^{m} \sum_{t=1}^{\ell} A[\idxWit[t,j], r] \cdot \Calpha^{t} \cdot \Cbeta^{j-1} + \Cbeta^{m} \cdot A[\idxRho, r+1] - A[\idxRho, r] \right) \\
            & + F[n_f+2, r] \cdot A[\idxMu, r] + F[n_f+3, r] \cdot A[\idxRho, r] .
        \end{split}
    \end{equation*}
    In the active rows, $\polyHornerGate$ enforces the Horner recurrence of round $\nRounds + 2$; in row $1$, it adds the mask; and in row $n_\text{horner} + 1$, it enforces the base case $A[\idxRho, n_\text{horner} + 1] = 0$.
\end{itemize}
Output $(\aipExt, \hRound, \hIdx)$.
\end{tcolorbox}
\caption{The extension transform \hornerGate for I-Plonkish. The transform adds the evaluation of the committed polynomials with Horner's method to an I-Plonkish instance. The mask occupies a separate advice column.}
\label{fig:compiler:artemis:iplonkish}
\end{figure*}

\section{Experimental Setup}
\label{apx:expsetup}
We run the benchmarks on AWS EC2 instances running Ubuntu 24.04,
choosing the instance according to the memory requirements of each benchmark.
We describe each benchmark and its instance below.
\fakeparagraph{zkML} We use the following models.
\begin{enumerate}[topsep=0.5em,leftmargin=0pt,align=left,labelsep=3pt,labelwidth=!,label=\textbf{\arabic*$\mathbf{)}$}]
\setlength\itemsep{8pt}
    \item \gls{mnist}:  A minimal CNN~\cite{Grimov2018-minmnist} with $8.1$K parameters and $444.9$K FLOPs, trained on the MNIST image classification task, evaluated on an r6i.8xlarge instance (32 vCPUs, 256 GB RAM).

    \item \gls{resnet18}: An image classifier~\cite{He2016-resnet} trained on CIFAR-10~\cite{Krizhevsky2009-cifar10}, with $280.9$K parameters and $81.9$M FLOPs, evaluated on an r6i.8xlarge instance  (32 vCPUs, 256 GB RAM).

    \item \gls{mobilenet}: A mobile-optimized image classifier~\cite{Sandler2018-mobilenet} trained on ImageNet, with $3.5$M parameters and $601.8$M FLOPs, evaluated on an r6i.16xlarge instance (64 vCPUs, 512 GB RAM).

    \item \gls{vgg}: A CNN~\cite{Simonyan2014-vgg} with $15.2$M parameters and $627.9$M FLOPs, trained on CIFAR-10~\cite{Krizhevsky2009-cifar10}, evaluated on an r6i.32xlarge instance (128 vCPUs, 1024 GB RAM).

    \item \gls{diffusion}: A small text-to-image Stable Diffusion model~\cite{Rombach2021-diffusion}, with $19.5$M parameters and $22.9$B FLOPs, evaluated on an x2idn.metal instance (128 vCPUs, 2048 GB RAM).

    \item \gls{gpt2}: A distilled transformer-based language model optimized for inference~\cite{Radford2019-gpt2}, with $81.3$M parameters and $188.9$M FLOPs, evaluated on an r6i.32xlarge instance (128 vCPUs, 1024 GB RAM).
\end{enumerate}

\fakeparagraph{TLSNotary}
We run the TLSNotary benchmarks on an AWS EC2 r6i.8xlarge instance (32 vCPUs, 256 GB RAM) running Ubuntu 24.04.
Each workload extends the Noir age-verification example of TLSNotary~\cite{tlsn_examples_zk} and sets the committed response size to the \textsf{maxRecvData} of an application plugin:
a bank account balance (520\,B), a Spotify top artist (2.3\,KiB), X (Twitter) account settings (3.9\,KiB), a Discord direct message (9.8\,KiB), an Uber rider profile (16\,KiB), and ID.me identity credentials (64\,KiB).

\section{Arithmetization Instantiations}
\label{apx:compiler:inst}
We provide the transformation \hornerGate for I-Plonkish in \Cref{fig:compiler:artemis:iplonkish}.
This transform replaces the abstract Steps~\ref{step:setup:aip}--\ref{step:setup:arith} of \zkSetup in \Cref{fig:compiler:artemis:meta}.
In round $\nRounds+1$, the prover sends \sMask in the cell $A[\idxMu, 1]$ together with copies of the committed sub-witnesses, and in round $\nRounds+2$ sends \sEvalValue in the cell $A[\idxRho, 1]$ together with the intermediate Horner sums.
The copy constraints force the copies to equal $\msgProjI(\pMsgs)$, and unrolling the Horner gate from its base case shows that every accepted grid satisfies $A[\idxRho, 1] = A[\idxMu, 1] + \sum_{t=1}^{\ell} \alpha^t \cdot \mathbf{w}_t(\sEvalPoint)$, which is the check that $\dCircuit^\prime$ adds in Step~\ref{step:setup:ext}.
The proof of \Cref{thm:artemis:aip} therefore carries over with $\pMsg_{\nRounds+1}$ and $\pMsg_{\nRounds+2}$ replaced by these two cells; the copies and intermediate sums are hidden messages that the reveal protocol does not open.

\ifdefined\separateproof
\section{Security Proof for \oursystem}
\label{apx:securityproof}
We provide this appendix as supplementary material, \texttt{\proofpdffile}, in our repository at \url{\proofrepourl}.
It defines the properties of \glspl{zksark} that \oursystem requires (challenge exposure, message extraction, and sound revelation) and contains the full proof of \Cref{thm:artemis:aip}.
\else

\fi

\fi %
\end{document}